\documentclass[twoside,11pt]{article}

\usepackage[preprint]{jmlr2e}
\usepackage{enumitem}
\usepackage[figuresright]{rotating}
\usepackage{newfloat}
\usepackage{framed}
\usepackage{amsmath,bm,bbm}
\usepackage{mathtools}
\usepackage{xcolor}

\newtheorem{property}[theorem]{Property}

\DeclareFloatingEnvironment[name=Algorithm]{myalgorithm} 
\DeclareMathOperator*{\logit}{logit}
\DeclareMathOperator*{\diag}{diag}
\DeclareMathOperator*{\tr}{tr}
\DeclareMathOperator*{\rank}{rank}

\DeclarePairedDelimiter\set{\{}{\}}
\DeclarePairedDelimiterX{\norm}[1]{\lVert}{\rVert}{#1}
\DeclarePairedDelimiterX{\abs}[1]{\lvert}{\rvert}{#1}
\DeclarePairedDelimiterX{\mean}[1]{\langle}{\rangle}{#1}

\DeclarePairedDelimiterX{\expectarg}[1]{[}{]}{%
    \ifnum\currentgrouptype=16 \else\begingroup\fi
    \activatebar#1
    \ifnum\currentgrouptype=16 \else\endgroup\fi
}

\DeclarePairedDelimiterX{\variancearg}[1]{(}{)}{%
    \ifnum\currentgrouptype=16 \else\begingroup\fi
    \activatebar#1
    \ifnum\currentgrouptype=16 \else\endgroup\fi
}

\newcommand{\innermid}{\nonscript\;\delimsize\vert\nonscript\;}
\newcommand{\activatebar}{%
    \begingroup\lccode`\~=`\|
    \lowercase{\endgroup\let~}\innermid
    \mathcode`|=\string"8000
}

\newcommand{\Prob}{\mathbb{P} \, \variancearg}
\newcommand{\E}{\mathbb{E} \, \expectarg}
\newcommand{\Eq}[2]{\mathbb{E}_{#1}\left[#2\right]}
\newcommand{\Var}{\operatorname{Var}\variancearg}
\newcommand{\Cov}{\operatorname{Cov}\variancearg}

\newcommand{\iidsim}{\overset{\text{iid}}\sim}
\newcommand{\indsim}{\overset{\text{ind.}}\sim}

\newcommand{\Bern}{\operatorname{Bernoulli}}
\newcommand{\InvGamma}{\operatorname{\Gamma^{-1}}}
\newcommand{\PG}{\operatorname{PG}}

\def\mOne{{\mathbbm{1}}}
\newcommand{\ind}[1]{\mOne_{\{#1\}}}

\newcommand{\Reals}[1]{\mathbb{R}^{#1}}

\newcommand{\bone}{\mathbf{1}}

\newcommand{\bY}{\mathbf{Y}}
\newcommand{\by}{\mathbf{y}}

\newcommand{\bX}{\mathbf{X}}
\newcommand{\cX}{\mathcal{X}}
\newcommand{\bx}{\mathbf{x}}
\newcommand{\bZ}{\mathbf{Z}}
\newcommand{\bz}{\mathbf{z}}

\newcommand{\bw}{\mathbf{w}}

\newcommand{\bv}{\mathbf{v}}

\newcommand{\bb}{\mathbf{b}}
\newcommand{\bc}{\mathbf{c}}

\newcommand{\be}{\mathbf{e}}
\newcommand{\bs}{\mathbf{s}}

\newcommand{\btheta}{\boldsymbol\theta}
\newcommand{\bTheta}{\boldsymbol\Theta}
\newcommand{\blambda}{\boldsymbol\lambda}

\newcommand{\bdelta}{\boldsymbol\delta}
\newcommand{\boldeta}{\boldsymbol\eta}
\newcommand{\bphi}{\boldsymbol\phi}
\newcommand{\bmu}{\boldsymbol\mu}
\newcommand{\bnu}{\boldsymbol\nu}

\newcommand{\bomega}{\boldsymbol\omega}
\newcommand{\boldm}{\mathbf{m}}

\newcommand{\dyads}{\mathcal{D}}

\jmlrheading{22}{2021}{1-55}{3/21}{}{loyal21}{Joshua Daniel Loyal and Yuguo Chen}

\ShortHeadings{An Eigenmodel for Dynamic Multilayer Networks}{Loyal and Chen}
\firstpageno{1}

\begin{document}

\title{An Eigenmodel for Dynamic Multilayer Networks}

\author{\name Joshua Daniel Loyal \email jloyal2@illinois.edu \\
        \name Yuguo Chen \email yuguo@illinois.edu \\
        \addr Department of Statistics\\
        University of Illinois at Urbana-Champaign\\
        Champaign, IL 61820, USA}

\editor{}

\maketitle

\begin{abstract}
Dynamic multilayer networks frequently represent the structure of multiple co-evolving relations; however, statistical models are not well-developed for this prevalent network type. Here, we propose a new latent space model for dynamic multilayer networks. The key feature of our model is its ability to identify common time-varying structures shared by all layers while also accounting for layer-wise variation and degree heterogeneity. We establish the identifiability of the model's parameters and develop a structured mean-field variational inference approach to estimate the model's posterior, which scales to networks previously intractable to dynamic latent space models. We demonstrate the estimation procedure's accuracy and scalability on simulated networks. We apply the model to two real-world problems: discerning regional conflicts in a data set of international relations and quantifying infectious disease spread throughout a school based on the student's daily contact patterns.
\end{abstract}


\section{Introduction} \label{sec:intro}

Dynamic multilayer networks are a prevalent form of relational data with applications in epidemiology, sociology, biology, and other fields~\citep{boccaletti2014}. Unlike static single-layer networks, which are limited to recording one dyadic relation among a set of actors at a single point in time, dynamic multilayer networks contain several types of dyadic relations, called layers, observed over a sequence of times. For instance, social networks contain several types of social relationships jointly evolving over time: friendship, vicinity, coworker-ship, partnership, and others. Also, international relations unfold through daily political events involving two countries, e.g., offering aid, verbally condemning, or participating in military conflict~\citep{hoff2015}. Lastly, the spread of information on social media occurs on a dynamic multilayer network, e.g., hourly interactions among Twitter users such as liking, replying to, and re-tweeting each other's content~\citep{domenico2013}. Proper statistical modeling of dynamic multilayer networks is essential for an accurate understanding of these complex systems.

The statistical challenge in modeling multiple co-evolving networks is maintaining a concise representation while also adequately describing important network characteristics. These characteristics include the dyadic dependencies in each individual static relation, such as degree heterogeneity and transitivity, the autocorrelation of the individual dyadic time series, and the common structures shared among the various relations. We provide an example of such network characteristics in Figure \ref{fig:iraq_networks}, which displays the monthly time series of four dyadic relations between Iraq and other countries from 2009 to 2017. A complete description of the data can be found in Section \ref{sec:realdata}. Within a layer (e.g., verbal cooperation), the individual time series (rows) are correlated with each other while also exhibiting strong autocorrelation. Furthermore, the four relations share a clear homogeneous structure, which is made especially evident after the abrupt change in all dyadic time series in late 2014 due to an American-led intervention in Iraq. We explore this event in more detail in Section~\ref{sec:realdata}. A statistical network model should decompose these dependencies in an interpretable way.

\begin{figure}[tbp]
\centering
\includegraphics[width=\textwidth]{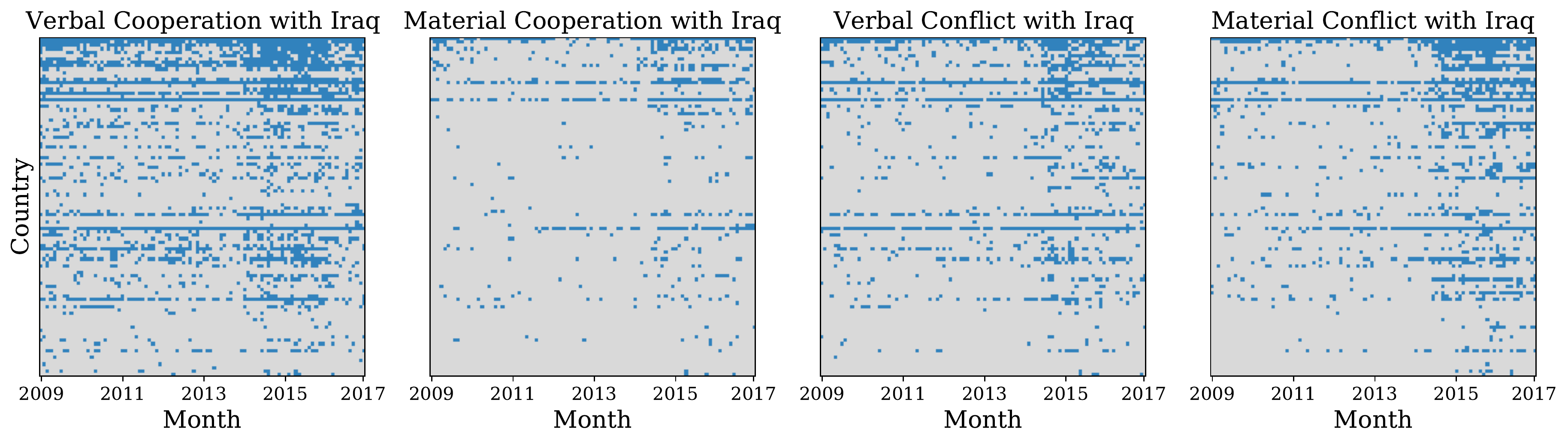}
\caption{Monthly cooperation and conflict relations between Iraq and other countries from 2009 to 2017. A blue (gray) square indicates a relation occurred (did not occur) between Iraq and that nation during that month.}
\label{fig:iraq_networks}
\end{figure}

To date, the statistics literature contains an expansive collection of network models designed to capture specific network properties. See \citet{goldenberg2010} and \citet{loyal2020} for a comprehensive review. An important class of network models is latent space models (LSMs) proposed in \citet{hoff2002}. The key idea behind LSMs is that each actor is assigned a vector in some low-dimensional latent space whose pairwise distances under a specified similarity measure determine the network's dyad-wise connection probabilities. The LSM interprets these latent features as an actor's unmeasured characteristics such that actors that are close in the latent space are more likely to form a connection. This interpretation naturally explains the high levels of homophily (assortativity) and transitivity in real-world networks. A series of works expanded the network characteristics captured by LSMs \citep{handcock2007, hoff2008, krivitsky2009, hoff2005, ma2020}, such as community structure, degree heterogeneity, heterophily (disassortativity), etc. Furthermore, researchers have adopted the LSM formulation to model both dynamic networks~\citep{sakar2006, durante2014, sewell2015, he2019} and static multilayer networks~\citep{gollini2016, townshend2017, dangelo2019, wang2019b, zhang2020}.

Currently, the statistical methodology for modeling dynamic multilayer networks is limited. \citet{snijders2013} introduced a stochastic actor-oriented model which represents the networks as co-evolving continuous-time Markov processes. In addition, \citet{hoff2015} introduced a multilinear tensor regression framework where dynamic multilayer networks are modeled through tensor autoregression. To our knowledge, the only existing LSM for dynamic multilayer networks is the Bayesian nonparametric model proposed in \citet{durante2017}. Although highly flexible, this model lacks interpretability due to strong non-identifiable issues. Furthermore, this model's applications are limited to small networks with only a few dozen nodes and time points due to the model's high computational complexity. Currently, the LSM literature lacks models that decompose the complexity of dynamic multilayer networks into interpretable components and scale to the large networks commonly analyzed in practice.

To address these needs, we develop a new Bayesian dynamic bilinear latent space model that is flexible, interpretable, and computationally efficient. Our approach identifies a common time-varying structure shared by all layers while also accounting for layer-wise variation. Intuitively, our model posits that actors have intrinsic traits that influence how they connect in each layer. Specifically, we identify a common structure in which we represent each node by a single latent vector shared across layers. Also, we introduce node-specific additive random effects (or socialities) to adjust for heavy-tailed degree distributions~\citep{rastelli2016}. The model accounts for layer-wise heterogeneity in two ways. First, the layers assign different amounts of homophily (heterophily) to each latent trait. Second, to capture the dependence of an actor's degree on relation type, we allow the additive random effects to vary by layer. Lastly, we propagate the latent variables through time via a discrete Markov process. These correlated changes capture the network's structural evolution and temporal autocorrelation.

To estimate our model, we derive a variational inference algorithm~\citep{wainwright2008, blei2017} that scales to networks much larger than those analyzed by previous approaches. We base our inference on a structured mean-field approximation to the posterior. Our approximation improves upon previous variational approximations found in the dynamic latent space literature~\citep{sewell2017} by retaining the latent variable's temporal dependencies. Furthermore, we derive a coordinate ascent variational inference algorithm that consists of closed-form updates. Our work leads to a novel approach to fitting dynamic latent space models using techniques from the linear Gaussian state space model (GSSM) literature.

The structure of our paper is as follows. In Section~\ref{sec:methodology}, we present our Bayesian parametric model for dynamic multilayer networks and discuss identifiability issues. Section~\ref{sec:estimation} outlines our structured mean-field approximation and the coordinate ascent variational inference algorithm used for estimation. Section~\ref{sec:simulation} demonstrates the accuracy and scalability of our inference algorithm on simulated networks of various sizes. In Section~\ref{sec:realdata}, we apply our model to two real-world networks taken from international relations and epidemiology. Finally, Section~\ref{sec:disc} concludes with a discussion of various model extensions and future research directions. The Appendices contain proofs, in-depth derivations of the variational inference algorithm, implementation details, and additional results and figures.

\vspace{1em}
\noindent {\it Notation.} We write $[N] = \set{1, \dots, N}$. We use the notation $B_{1:L}$ to refer to the sequence $(B_1, B_2, \dots, B_L)$ where $B_{\ell}$ is any indexed object. Also, for objects with a double index, we use the notation $C_{1:M, 1:N}$ to refer to the collection $(C_{mn})_{(m, n) \in [M] \times [N]}$. We use $\ind{x = a}$ to denote the Boolean indicator function, which evaluates to 1 when $x = a$ and 0 otherwise. We denote an $n$-dimensional vector of ones by $\mathbf{1}_n$ and the $n \times n$ identity matrix by $I_n$. Furthermore, given a vector $\bv \in \Reals{d}$, we use $\diag(\bv)$ to indicate a $d \times d$ diagonal matrix with the elements of $\bv$ on the diagonal. Lastly, we use $I_{p,q} = \diag(1, \dots, 1, -1, \dots, -1)$ to denote a diagonal matrix with $p$ ones followed by $q$ negative ones on the diagonal.

\section{An Eigenmodel for Dynamic Multilayer Networks}\label{sec:methodology}

In this section, we develop our Bayesian model for dynamic multilayer networks. To begin, we formally introduce dynamic multilayer network data. Dynamic multilayer networks consist of $K$ relations measured over $T$ time points between the same set of $n$ nodes (or actors). We collect these relations in binary adjacency matrices $\bY^k_{t} \in \set{0, 1}^{n \times n}$ for $1 \leq k \leq K$ and $1 \leq t \leq T$. The entries $Y_{ijt}^k$ indicate the presence ($Y_{ijt}^k = 1$) or absence ($Y_{ijt}^k = 0$) of an edge between actors $i$ and $j$ in layer $k$ at time $t$. This article only considers undirected networks without self-loops so that $\bY^k_t$ is a symmetric matrix. We discuss extensions of our model to weighted and directed networks in Section \ref{sec:disc}.

\subsection{The Model} \label{sec:method}

Here, we propose our new eigenmodel for dynamic multilayer networks with the goal of capture the correlations between different dyads within a network, the dyads' autocorrelation over time, and the dependence between layers. Specifically, we assume that the dyads are independent Bernoulli random variables conditioned on the latent parameters:
\begin{equation*}
\Prob{\bY^1_{1:T}, \dots, \bY^K_{1:T} \mid \bdelta_{1:K, 1:T}, \Lambda_{1:K}, \mathcal{X}_{1:T}} = \prod_{k=1}^K  \prod_{t=1}^T \prod_{j < i} \Prob{Y_{ijt}^k | \delta_{k,t}^i, \delta_{k,t}^j, \Lambda_k, \bX_t^i, \bX_t^j},
\end{equation*}
where
\begin{equation}\label{eq:logit_reg}
\logit\left[\Prob{Y_{ijt}^k = 1 | \delta_{k,t}^i, \delta_{k,t}^j, \Lambda_k, \bX_t^i, \bX_t^j}\right] = \Theta_{ijt}^k = \delta_{k,t}^i + \delta_{k,t}^j  + \bX_t^{i \, \rm T} \Lambda_k \bX_t^j.
\end{equation}

In Equation (\ref{eq:logit_reg}), layer $k$'s log-odds matrix at time $t$, $\bTheta_t^k \in \Reals{n \times n}$ with elements $\Theta_{ijt}^k$, contains two latent random effects that induce essential unconditional dependencies in the dynamic multilayer network's dyads. We defer specification of their distributions until the next section. First, the sociality effects $\bdelta_{k, t} = (\delta_{k,t}^1, \dots, \delta_{k,t}^n)^{\rm T} \in \Reals{n}$ model degree heterogeneity and node-level autocorrelation. Second, the time-varying latent positions $\mathcal{X}_{t} = (\bX_t^1, \dots, \bX_t^n)^{\rm T} \in \Reals{n \times d}$, where $d$ is the latent space's dimension, induce clusterability~\citep{hoff2008} in the networks and dyadic autocorrelation. Finally, the homophily coefficients $\Lambda_k = \diag(\blambda_k) \in \Reals{d \times d}$ are diagonal matrices that quantify each relation's level of homophily along a latent dimension. For each $\lambda_{kh}$ ($1 \leq h \leq d$), positive values ($\lambda_{kh} > 0$) indicate homophily along the $h$th latent dimension in layer $k$, negative values ($\lambda_{kh} < 0$) indicate heterophily along the $h$th latent dimension in layer $k$, and a zero value ($\lambda_{kh} = 0$) indicates the $h$th latent dimension does not contribute to the connection probability in layer $k$. Furthermore, the model captures common structures among the layers by sharing a common set of latent trajectories.

In matrix form, the log-odds matrices are
\begin{equation*}
\bTheta_t^k = \bdelta_{k,t} \mathbf{1}_n^{\rm T} + \mathbf{1}_n \bdelta_{k,t}^{\rm T} + \mathcal{X}_t \Lambda_k \mathcal{X}_t^{\rm T}.
\end{equation*}
To ensure identifiability of the model parameters, we require both a centered latent space, that is $J_n \cX_t = \cX_t$ where $J_n = I_n - (1/n) \mathbf{1}_n \mathbf{1}_n^{\rm T}$, and $\Lambda_r = I_{p,q}$, where $p + q = d$, for some reference layer $r \in \set{1, \dots, K}$. In an applied setting, one could take a particular interesting layer as the reference layer. Otherwise, as in this work, we select $r = 1$ as the reference. As we elaborate in Section \ref{subsec:identify}, we use these conditions to identify the socialities $\bdelta_{k,t}$ and to identify  $\cX_t$ up to a common linear transformation of its rows. At the same time, we show that the bilinear term, $\cX_t \Lambda_k \cX_t^{\rm T}$, is directly identifiable.

Overall, the proposed model's parameterization reduces the dimensionality of dynamic multi-relational data. The model contains $nTK + nTd + Kd$ parameters, which, for typical values of $d$, is much less than the $KTn(n-1)/2$ dyads that originally summarized the dynamic multilayer network. In this work, we fix $d = 2$, which allows us to use the latent space for network visualization. For a discussion on data-driven choices of $d$, see Section \ref{sec:disc}. Next, we elaborate on the interpretation of the model's parameters and our inclusion of temporal correlation in the random effects.

\subsection{Layer-Specific Social Trajectories}

An actor's sociality, $\delta_{k,t}^i$, represents their global popularity in layer $k$ at time $t$. In particular, holding all other parameters fixed, the larger an actor's sociality $\delta_{k,t}^i$, the more likely they are to connect with other nodes in the $k$th layer at time $t$ regardless of their position in the latent space. Formally, $\delta_{k,t}^i$ is the conditional log-odds ratio of actor $i$ forming a connection with another actor in layer $k$ at time $t$ compared to an actor with the same latent position as actor $i$ but with $\delta_{k,t}^i = 0$. Hub nodes are an example of nodes with a high sociality, while isolated nodes have a low sociality. An actor's sociality can differ between layers. We find this flexibility necessary to model real-world multilayer relations. For example, in the international relations network presented in the introduction, a peaceful nation might participate in many cooperative relations while rarely engaging in conflict relations.

The $i$th actor's social trajectory in layer $k$, $\delta_{k,1:T}^i$, measures their time-varying sociality in the $k$th layer. For example, a nation's propensity to engage in militaristic relations might increase after a regime change. We assume that the social trajectories are independent across layers $k$ and individuals $i$ and propagate them through time via a shared Markov process:
\begin{align*}
\delta_{k,1}^i \iidsim N(0, \tau_{\delta}^2), \qquad \delta_{k,t}^i \sim N(\delta_{k,t-1}^i, \sigma_{\delta}^2), \qquad t = 2, \dots, T, \quad k = 1, \dots, K,
\end{align*}
where iid stands for independent and identically distributed. In the previous expression, $\tau_{\delta}^2$ measures the sociality effects' initial variation over all layers. Similarly, $\sigma_{\delta}^2$ measures the sociality effects' variation over time. In particular, a small value of $\sigma_{\delta}^2$ indicates that most social trajectories are flat with little dynamic variability. We place the following conjugate priors on the variance parameters: $\tau_{\delta}^2 \sim \InvGamma(a_{\tau^2_{\delta}}/2, b_{\tau^2_{\delta}}/2)$ and $\sigma_{\delta}^2 \sim \InvGamma(c_{\sigma^2_{\delta}}/2, d_{\sigma^2_{\delta}}/2)$.

\subsection{Dynamic Latent Features Shared Between Layers}\label{subsubsec:latent_space}

Like other latent space models, we assume that the probability of two actors forming a connection depends on their latent representations in an unobserved latent space. Specifically, we assign every actor a latent feature $\bX_t^i \in \Reals{d}$ at each time point. Relations in dynamic networks typically have strong autocorrelations wherein the dyadic relations and latent features slowly vary over time. These autocorrelations are captured by a distribution that assumes the latent positions propagate through time via a shared Markov process~\citep{sakar2006, sewell2015}:
\begin{align*}
\bX_1^i \iidsim N(0, \tau^2 I_d), \qquad \bX_t^i \sim N(\bX_{t-1}^i, \sigma^2 I_d), \qquad t = 2, \dots, T.
\end{align*}
Intuitively, these dynamics assume that changes in the network's connectivity patterns are partly due to changes in the actor's latent features. Like the social trajectories' dynamics, the $\tau^2$ parameter in the previous expression measures the latent space's initial variation or size. Also, $\sigma^2$ measures the step size of each latent position's Gaussian random-walk. We place the following conjugate priors on the variance parameters: $\tau^2 \sim \InvGamma(a_{\tau^2}/2, b_{\tau^2}/2)$ and $\sigma^2 \sim \InvGamma(c_{\sigma^2}/2, d_{\sigma^2}/2)$.

\subsection{Layer-Specific Homophily Levels}
The proposed model posits that the multilayer networks are correlated because they share a single set of latent positions $\mathcal{X}_t$ among all layers. Mathematically, this restriction allows the model to capture common structures across layers. The homophily coefficients $\Lambda_k = \diag(\blambda_k)$ for $\blambda_k \in \Reals{d}$ allow for variability between the layers. Intuitively, the model assumes that two relations differ because they put distinct weights on the latent features. For example, homophilic features in a friendship relation may be heterophilic in a combative relation. For interpretability, we restrict $\Lambda_k$ to a diagonal matrix. We place independent multivariate Gaussian priors on the diagonal elements
\begin{equation*}
\blambda_k \iidsim N(0, \sigma_{\lambda}^2 I_d),\qquad k = 2,\dots,K.
\end{equation*}

For identifiability reasons, the homophily coefficients take values of $\pm 1$ in the first layer. We enforce this reference layer constraint by re-parameterizing the reference layer's diagonal elements in terms of Bernoulli random variables
\begin{equation*}
\lambda_{1h} = 2 u_h - 1,\qquad u_h \iidsim \Bern(\rho),\qquad h = 1, \dots, d,
\end{equation*}
where $\rho$ is the prior probability of an assortative relationship along a latent dimension. Under this constraint, we interpret the other layer's homophily coefficients in comparison to the reference. For example, if $\lambda_{11} = -1$ and $\lambda_{21} = 2$, then the second layer weights dimension one twice as heavily as the reference layer while exhibiting homophily instead of heterophily.

\subsection{Identifiability and Interpretability}\label{subsec:identify}

Here, we present sufficient conditions for identifiability and their implications on inference. For bilinear latent space models with sociality effects, it is natural to require the matrix of latent positions to be centered and full rank~\citep{zhang2020, macdonald2020}. In addition to the previous conditions, Proposition~\ref{prop:identify} shows that restricting the reference layer's homophily coefficients to take values $\pm 1$ is sufficient to identify our model up to a restricted linear transformation of the latent space. However, the form of this linear transformation is difficult to interpret. As such, we provide stronger conditions that are sufficient to restrict the linear transformations to interpretable forms. We provide the proofs in Appendix~\ref{subsec:proofs}.

\begin{proposition}[Identifiability Conditions]\label{prop:identify}
Suppose that two sets of parameters $\set{\bdelta_{1:K, 1:T},\allowbreak \cX_{1:T},\allowbreak \Lambda_{1:K}}$ and $\set{\tilde{\bdelta}_{1:K, 1:T}, \tilde{\cX}_{1:T}, \tilde{\Lambda}_{1:K}}$ satisfy the following conditions:
\begin{enumerate}[label={\bf A\arabic*.}, ref=A\arabic*]
\item \label{assmp:A1} $J_n \mathcal{X}_t = \mathcal{X}_t$ and $J_n \tilde{\mathcal{X}}_t = \tilde{\mathcal{X}}_t$ for $t = 1, \dots, T$ where $J_n = I_n - \frac{1}{n}\mathbf{1}_n \mathbf{1}_n^{\rm T}$.
\item \label{assmp:A2} $\rank(\cX_t) = \rank(\tilde{\cX}_t) = d$ for $t = 1, \dots, T$.
\item \label{assmp:A3} For at least one $r \in \set{1, \dots, K}$, $\Lambda_r = I_{p,q}$ and $\tilde{\Lambda}_r = I_{p', q'}$.
\end{enumerate}
Then the model is identifiable up to a linear transformation of the latent space, that is, if for all $1 \leq k \leq K$ and $1 \leq t \leq T$ we have that
\begin{equation*}
\bdelta_{k,t}\bone_n^{\rm T} + \bone_n \bdelta_{k,t}^{\rm T} + \cX_t \Lambda_k \cX_t^{\rm T} = \tilde{\bdelta}_{k,t} \bone_n^{\rm T}+ \bone_n \tilde{\bdelta}_{k,t}^{\rm T} + \tilde{\cX}_t \tilde{\Lambda}_k \tilde{\cX}_t^{\rm T},
\end{equation*}
then for all $1 \leq k \leq K$ and $1 \leq t \leq T$ we have that
\begin{equation*}
\tilde{\bdelta}_{k,t} = \bdelta_{k,t}, \ \tilde{\cX}_t = \cX_t M_t, \ \tilde{\Lambda}_{k} = M_t^{\rm T} \Lambda_k M_t,
\end{equation*}
where each $M_t \in \Reals{d \times d}$ satisfies $M_t I_{p',q'} M_t^{\rm T} = I_{p,q}\ $ for $1 \leq t \leq T$.
\end{proposition}

Assumption \ref{assmp:A1} alone, which centers the latent space, is sufficient to remove any confounding between the social trajectories and the latent positions. This issue arises because the likelihood is invariant to translations in the latent space. Indeed,
\begin{equation*}
\begin{split}
 \delta_{k,t}^i + \delta_{k,t}^j + \bX_t^{i \, \rm T}\Lambda_k \bX_t^j &= \delta_{k,t}^i + \delta_{k,t}^j + (\bX_t^i - \bc + \bc)^{\rm T}\Lambda_k(\bX_t^j - \bc + \bc), \\
 &= \tilde{\delta}_{k,t}^i + \tilde{\delta}_{k,t}^j + \tilde{\bX}_t^{i \, \rm T} \Lambda_k \tilde{\bX}_t^j,
 \end{split}
\end{equation*}
where $\tilde{\bX}_t^i = \bX_t^i - \bc$ and $\tilde{\delta}_{k,t}^i = \delta_{k,t}^i + \tilde{\bX}_t^{i \, \rm T}\Lambda_k \bc + \bc^{\rm T} \Lambda_k \bc / 2$.
Such confounding is present in previous bilinear latent space models that treat the latent positions as random effects~\citep{hoff2002, krivitsky2009}. Our prior specification does not directly enforce the centering constraint. Instead, we let all the parameters float because including this redundancy can speed up the variational algorithm proposed in the next section~\citep{liu2012, vandyke2001, qi2006}. However, when summarizing the model, we want to identify the social and latent trajectories so that the sociality effects are no longer confounded with the bilinear term. Therefore, after estimation, we perform posterior inference on $\tilde{\bX}_t^i$ and $\tilde{\delta}_{k,t}^i$ with $\bc = (1/n) \sum_{i=1}^n \bX_t^i$. See Section \ref{subsec:identifiable_parameters} for details.

Under Proposition~\ref{prop:identify}, the latent space is identifiable up to a restricted set of linear transformations, $M_t$, that are difficult to interpret. For this reason, we provide the following proposition, which reduces the set of linear transformations to a well-studied group of transformations when the latent space dimension $d \leq 3$, which is often the case in practice.
\begin{proposition}\label{prop:indefinite_ortho}
Consider the same setup as in Proposition~\ref{prop:identify} and assume that conditions \ref{assmp:A1}---\ref{assmp:A3} are satisfied. If $1 \leq d \leq 3$, then $I_{p,q} = I_{p',q'}$ so that each $M_t$ is in the indefinite orthogonal group, i.e., $M_t I_{p,q} M_t^{\rm T} = I_{p,q}\ $ for $1 \leq t \leq T$.
\end{proposition}

Invariance under the indefinite orthogonal group is common in LSMs that allow a disassortative latent space~\citep{rubin2017}. Furthermore, this group reduces to the orthogonal group, another source of non-identifiability in many LSMs, when $p$ or $q$ equals $d$. The most notable property of the indefinite orthogonal group is that it does not preserve Euclidean distances. This implies that any inference based on Euclidean distances in the latent space is not well-defined. For a detailed discussion, we refer the reader to \citet{rubin2017}, who studied inference under such a non-identifiability in the context of a generalized random dot product graph (RDPG) model. In particular, they showed that any post hoc clustering of the latent positions should use a Gaussian mixture model with elliptical covariance matrices since the clustering results are invariant to indefinite orthogonal transformations. This observation is essential if one intends to use the latent positions for community detection.

For a general $d$, a mild condition on the homophily coefficients is enough to restrict $M_t$ to a signed-permutation matrix. In essence, we show that requiring the layers to measure different types of relations removes the latent space's invariance under the general indefinite orthogonal group. We state the result in the following proposition.
\begin{proposition}\label{prop:permutation}
Consider the same setup as in Proposition~\ref{prop:identify} and assume that conditions \ref{assmp:A1}---\ref{assmp:A3} are satisfied. In addition, suppose the following condition is satisfied:
\begin{enumerate}[label={\bf A4.}, ref=A4]
\item \label{assmp:A4} For at least one layer $k \neq r$, $\rank(\Lambda_k) = \rank(\tilde{\Lambda}_k) = d$ and both $\Lambda_k \Lambda_r$ and $\tilde{\Lambda}_k \tilde{\Lambda}_r$ have distinct diagonal elements.
\end{enumerate}
Then $I_{p,q} = I_{p',q'}$ and each $\set{M_t}_{t=1}^T$ is a signed permutation matrix, i.e., $M_t = P \diag(\bs)$ where $\bs \in \set{\pm 1}^d$ and $P$ is a $d \times d$ permutation matrix.
\end{proposition}

Proposition~\ref{prop:permutation} says that Assumptions \ref{assmp:A1}---\ref{assmp:A4} are strong enough to identify $\cX_t$ up to sign-flips and permutations of its columns and the homophily coefficients $\Lambda_k = \diag(\blambda_k)$ up to the same set of permutations applied to the rows of $\blambda_k$. Intuitively, Assumption~\ref{assmp:A4} asserts that at least one layer should measure a homophily pattern distinct from the reference layer. For example, Assumption~\ref{assmp:A4} is not satisfied for layers whose homophily coefficients satisfy $\Lambda_k = \alpha I_{p,q}$ for any scalar $\alpha$. While mildly restrictive, we expect Assumption~\ref{assmp:A4} to hold when the layers measure different phenomena. For example, we expect the cooperation and conflict layers that make up the international relation networks studied in the real data analysis of Section~\ref{sec:realdata} to have distinct homophily patterns. Most importantly, Assumption~\ref{assmp:A4} holds with probability one under our choice of priors. As such, we assume each $M_t$ is restricted to a signed permutation matrix for any value of $d$ going forward.

\section{Variational Inference}\label{sec:estimation}

We presented the eigenmodel for multilayer dynamic networks and discussed issues of identifiability. Now, we turn to the problem of parameter estimation and inference. We take a Bayesian approach to inference with the goal of providing both posterior mean and credible intervals for the model's parameters. However, the large amount of dyadic relations that comprise multilayer dynamic networks makes Markov chain Monte Carlo inference impractical for all but small networks. For this reason, we employ a variational approach~\citep{wainwright2008}. For notational convenience, we collect the latent variables in $\btheta = \set{\bdelta_{1:K, 1:T}, \Lambda_{1:K}, \mathcal{X}_{1:T}}$  and the state space parameters in $\bphi = \set{\tau^2, \sigma^2,\tau_{\delta}^2, \sigma_{\delta}^2}$.

We aim to approximate the intractable posterior distribution $p(\btheta, \bphi \mid \bY_{1:T}^1, \dots, \bY_{1:T}^K)$ with a tractable variational distribution $q(\btheta, \bphi)$ that minimizes the KL divergence between $q(\btheta, \bphi)$ and $p(\btheta, \bphi \mid \bY_{1:T}^1, \dots, \bY_{1:T}^K)$. It can be shown that minimizing this divergence is equivalent to maximizing the evidence lower bound (ELBO), a lower bound on the data's marginal log-likelihood
\begin{equation*}
\mathcal{L}(q) = \Eq{q(\btheta, \bphi)}{\log p(\bY_{1:T}^1, \dots, \bY_{1:T}^K, \btheta, \bphi) - \log q(\btheta, \bphi)}  \leq \log p(\bY_{1:T}^1, \dots, \bY_{1:T}^K).
\end{equation*}
In general, the ELBO is not concave; however, optimization procedures often converge to a reasonable optimum. One still has the flexibility to specify the variational distribution's form, although the need to evaluate and sample from it often guides this choice. A convenient form is the structured mean-field approximation, which factors $q(\btheta, \bphi)$ into a product over groups of dependent latent variables. Furthermore, when the model consists of conjugate exponential family distributions, this form lends itself to a simple coordinate ascent optimization algorithm with optimal closed-form coordinate updates. For an introduction to variational inference, see~\citet{blei2017}.

In what follows, we present a structured mean-field variational inference algorithm that preserves the eignmodel's essential statistical dependencies and maintains closed-form coordinate updates. Normally, the absence of conditional conjugacy in latent space models poses a challenge for closed-form variational inference. Indeed, popular solutions require additional approximations of the expected log-likelihood~\citep{townshend2013, gollini2016}, which may bias parameter estimates. Another challenge is that the standard mean-field variational approximation is inadequate for our model due to the latent variable's temporal dependencies. Our solution employs P\'olya-gamma augmentation~\citep{polson2013} and variational Kalman smoothing~\citep{beal2003} to produce a new and widely applicable variational inference algorithm for bilinear latent space models for dynamic networks.

\subsection{P\'olya-gamma Augmentation}

As previously mentioned, we use P\'olya-gamma augmentation to render the model conditionally conjugate. For each dyad in the dynamic multilayer network, we introduce auxiliary P\'olya-gamma latent variables $\omega_{ijt}^k \iidsim \PG(1, 0)$, where $\PG(b, c)$ denotes a P\'olya-gamma distribution with parameters $b > 0$ and $c \in \mathbb{R}$. For convenience, we use $\bomega$ to denote the collection of all P\'olya-gamma auxiliary variables. As shown in \citet{polson2013}, the joint distribution is now proportional to
\begin{equation*}
p(\bY_{1:T}^1, \dots, \bY_{1:T}^k, \btheta, \bphi, \bomega) \propto p(\btheta) p(\bphi) p(\bomega) \prod_{k=1}^K \prod_{t=1}^T \prod_{j < i} \exp\left\{z_{ijt}^k \psi_{ijt}^k - \omega_{ijt}^k (\psi_{ijt}^k)^2/2\right\},
\end{equation*}
where $z_{ijt}^k = Y_{ijt}^k - 1/2$ and $\psi_{ijt}^k = \delta_{k,t}^i + \delta_{k,t}^j + \bX_t^{i\, \rm T} \Lambda_k \bX_t^j$. This joint distribution results in each latent variable's full conditional distribution lying within the exponential family, a property sufficient for closed-form variational inference.

\subsection{The Structured Mean-Field Approximation}
We use the following structured mean-field approximation to the augmented model's posterior
\begin{align}\label{eq:variational_q}
q(\btheta, \bphi, \bomega) &= \left[\prod_{h=1}^d q(\lambda_{1h})\right] \left[\prod_{k=2}^K q(\blambda_k)\right]\left[\prod_{k=1}^K \prod_{i=1}^n q(\delta_{k, 1:T}^i)\right] \left[\prod_{i=1}^n q(\bX_{1:T}^i) \right] \left[\prod_{k=1}^K \prod_{t=1}^T \prod_{j < i} q(\omega_{ijt}^k) \right] \nonumber \\
    &\qquad\qquad \times q(\tau^2) q(\sigma^2) q(\tau_{\delta}^2) q(\sigma_{\delta}^2).
\end{align}
This factorization is attractive because it maintains the essential temporal dependencies in the posterior distribution. Since we use optimal variational factors, preserving these dependencies increases the approximate posterior distribution's accuracy.

\subsection{Coordinate Ascent Variational Inference Algorithm}

To maximize the ELBO, we employ coordinate ascent variational inference (CAVI). CAVI performs coordinate ascent on one variational factor at a time, holding the rest fixed. The optimal coordinate updates take a simple form: set each variational factor to the corresponding latent variable's expected full conditional probability under the remaining factors. For example, the update for $q(\bX_{1:T}^i)$ is given by
\begin{equation*}
\log q(\bX_{1:T}^i) = \Eq{-q(\bX_{1:T}^i)}{\log p(\bX_{1:T}^i \mid \cdot)} + c,
\end{equation*}
where $\Eq{-q(\bX_{1:T}^i)}{\cdot}$ indicates an expectation taken with respect to all variational factors except $q(\bX_{1:T}^i)$, $p(\bX_{1:T}^i\mid \cdot)$ is the full conditional distribution of $\bX_{1:T}^i$, and $c$ is a normalizing constant. When the full conditionals are members of the exponential family, a coordinate update involves calculating the natural parameter's expectations under the remaining variational factors.

The CAVI algorithm alternates between optimizing $q(\bomega)$, $q(\bdelta_{1:K, 1:T})$, $q(\mathcal{X}_{1:T})$, $q(\Lambda_{1:K})$, and $q(\bphi)$. Algorithm \ref{alg:cavi} outlines the full CAVI algorithm and defines some notation used throughout the rest of the article. We summarize each variational factor's coordinate update in the following sections. Appendix \ref{sec:variational_updates} and Appendix \ref{sec:variational_smoother} provide the full details and derivations of the coordinate updates and the variational Kalman smoothers, respectively.

\begin{myalgorithm}[p]
\begin{framed}
Define the following expectations taken with respect to the full variational posterior:
\begin{equation*}
\begin{split}
\E{\bX_t^i} &= \bmu_t^i, \quad \Var{\bX_t^i} = \Sigma_t^i, \quad \Cov{\bX_t^i, \bX_{t+1}^i} = \Sigma_{t, t+1}^i, \\
\E{\delta_{k,t}^i} &= \mu_{\delta_{k,t}^i},\quad \Var{\delta_{k,t}^i} = \sigma_{\delta_{k,t}^i}^2,\quad \Cov{\delta_{k,t}^i, \delta_{k,t+1}^i} = \sigma_{\delta_{k,t,t+1}^i}^2, \\
 \E{\blambda_k} &= \bmu_{\blambda_k},\quad \Var{\blambda_k} = \Sigma_{\blambda_k},\quad \E{\omega_{ijt}^k} = \mu_{\omega_{ijt}^k}.
 \end{split}
\end{equation*}

Iterate the following steps until convergence:
\begin{enumerate}
\item Update each $q(\omega_{ijt}^k) = \PG(1, c_{ijt}^k)$ as in Algorithm \ref{alg:cavi_omega}.
\item Update
\begin{itemize}
\item[] $q(\delta_{k, 1:T}^i)$ : a Gaussian state space model for $i \in \set{1, \dots, n}$ and $k \in \set{1, \dots, K}$,
\item[] $q(\tau^2_{\delta}) = \InvGamma(\bar{a}_{\tau^2_{\delta}}/2, \bar{b}_{\tau^2_{\delta}}/2)$,
\item[] $q(\sigma_{\delta}^2) = \InvGamma(\bar{c}_{\sigma^2_{\delta}}/2, \bar{d}_{\sigma^2_{\delta}}/2)$,
\end{itemize}
using a variational Kalman smoother as in Algorithm \ref{alg:cavi_delta}.
\item Update
\begin{itemize}
\item[] $q(\bX_{1:T}^i)$ : a Gaussian state space model for $i \in \set{1, \dots, n}$,
\item[] $q(\tau^2) = \InvGamma(\bar{a}_{\tau^2}/2, \bar{b}_{\tau^2}/2)$,
\item[] $q(\sigma^2) = \InvGamma(\bar{c}_{\sigma^2} / 2, \bar{d}_{\sigma^2}/2)$,
\end{itemize}
using a variational Kalman smoother as in Algorithm \ref{alg:cavi_latent_space}.
\item Update $q(\lambda_{1h}) = p_{\lambda_{1h}}^{\ind{\lambda_{1h} = 1}} \ (1 - p_{\lambda_{1h}})^{\ind{\lambda_{1h} = -1}}$ for $h \in \set{1, \dots, d}$ as in Algorithm \ref{alg:cavi_lambda}.
\item Update $q(\blambda_k) = N(\bmu_{\blambda_k}, \Sigma_{\blambda_k})$ for $k \in \set{2, \dots, K}$ as in Algorithm \ref{alg:cavi_lambda}.
\end{enumerate}
\end{framed}
\caption{Coordinate ascent variational inference for the eigenmodel for dynamic multilayer networks. Appendix \ref{sec:variational_updates} contains the details of Algorithms \ref{alg:cavi_omega}---\ref{alg:cavi_lambda}. Iterations are performed until successive differences of the expected log-likelihood, Equation (\ref{eq:expect_loglik}), drop below a tolerance threshold.}
\label{alg:cavi}
\end{myalgorithm}

\subsubsection{Updating $q(\omega_{ijt}^k)$}

By the exponential tilting property of the P\'olya-gamma distribution, we have
\begin{equation*}
\log q(\omega_{ijt}^k) = \Eq{-q(\omega_{ijt}^k)}{p_{\text{PG}}(\omega_{ijt}^k \mid 1, \psi_{ijt}^k)} + c,
\end{equation*}
where $p_{PG}(\omega \mid b, c)$ is the density of $\text{PG}(b, c)$ random variable. This density is a member of the exponential family with natural parameter $-(\psi_{ijt}^k)^2/2$. We provide the full coordinate update, which involves taking the expectation of $(\psi_{ijt}^k)^2$, in Algorithm \ref{alg:cavi_omega} of Appendix \ref{sec:variational_updates}.

\subsubsection{Updating $q(\delta_{k, 1:T}^i)$, $q(\tau_{\delta}^2)$, $q(\sigma_{\delta}^2)$}

Under the P\'olya-gamma augmentation scheme, the conditional distributions of the social trajectories take the form of linear Gaussian state space models. In particular,
\begin{equation}\label{eq:delta_gssm}
\log q(\delta_{k,1:T}^i) = \log h(\delta_{k,1}^i) + \sum_{t=2}^T \log h(\delta_{k,t}^i \mid \delta_{k,t-1}^i) + \sum_{t=1}^T \log h(\bz_{k,t}^i \mid \delta_{k,t}^i)  + c,
\end{equation}
where
\begin{align*}
\log h(\delta_{k,1}^i) &= \Eq{q(\tau_{\delta}^2)}{\log N(\delta_{k,1}^i \mid 0, \tau_{\delta}^2)}, \\
\log h(\delta_{k,t}^i \mid \delta_{k,t-1}^i) &= \Eq{q(\sigma_{\delta}^2)}{\log N(\delta_{k,t}^i \mid \delta_{k,t-1}^i, \sigma_{\delta}^2)}, \\
\log h(\bz_{k,t}^i \mid \delta_{k,t}^i) &= \Eq{-q(\delta_{k,1:T}^i)}{\sum_{j \neq i} \log N(z_{ijt}^k \mid \omega_{ijt}^k \, \delta_{k,t}^i + \omega_{ijt}^k (\delta_{k,t}^j + \bX_t^{i\, \rm T}\Lambda_k \bX_t^j),\  \omega_{ijt}^k)}.
\end{align*}
In the previous expressions, $\bz_{k,t}^i \in \Reals{n-1}$ is a vector that consists of stacking $z_{ijt}^k$ for $j \neq i$ and $N(\bx \mid \bmu, \Sigma)$ is the density of a $N(\bmu, \Sigma)$ random variable. Because all densities involved are Gaussian, the expectations yield Gaussian densities with natural parameters that depend on the remaining variational factors. Thus, we recognize the optimal variational distribution as a GSSM. The expected sufficient statistics needed to update the remaining variational factors can be computed with either the variational Kalman smoother \citep{beal2003} or a standard Kalman smoother under an augmented state space model \citep{barber2007}. We use the variational Kalman smoother. Furthermore, the inverse-gamma priors on the state space parameters result in fully conjugate coordinate updates for $\tau_{\delta}^2$ and $\sigma_{\delta}^2$. The update for the social trajectories is presented in Algorithm \ref{alg:cavi_delta} of Appendix \ref{sec:variational_updates}.

\subsubsection{Updating $q(\bX_{1:T}^i)$, $q(\tau^2)$, $q(\sigma^2)$}

Similar to the social trajectories, the conditional distributions of the latent trajectories are also GSSMs. Specifically,
\begin{equation}\label{eq:x_gssm}
\log q(\bX_{1:T}^i) = \log h(\bX_1^i) + \sum_{t=2}^T \log h(\bX_t^i \mid \bX_{t-1}^i) + \sum_{t=1}^T \log h(\bz_{t}^i \mid \bX_{t}^i)  + c,
\end{equation}
where
\begin{align*}
\log h(\bX_1^i) &= \Eq{q(\tau^2)}{\log N(\bX_1^i \mid 0, \tau^2)}, \\
\log h(\bX_t^i \mid \bX_{t-1}^i) &= \Eq{q(\sigma^2)}{\log N(\bX_t^i \mid \bX_{t-1}^i, \sigma^2)}, \\
\log h(\bz_{t}^i \mid \bX_t^i) &= \Eq{-q(\bX_{1:T}^i)}{\sum_{k=1}^K \sum_{j \neq i} \log N(z_{ijt}^k \mid \omega_{ijt}^k (\delta_{k,t}^i + \delta_{k,t}^j) + \omega_{ijt}^k \bX_t^{j\, \rm T}\Lambda_k \bX_t^i,\ \omega_{ijt}^k)}.
\end{align*}
In the previous expressions, $\bz_t^i \in \Reals{K(n-1)}$ is a vector formed by stacking $z_{ijt}^k$ for $j \neq i$ and $k = 1, \dots, K$. Once again, we recognize that $q(\bX_{1:T}^i)$ is a GSSM; therefore, we can calculate the expected sufficient statistics with the variational Kalman smoother. Also, the inverse-gamma priors on $\tau^2$ and $\sigma^2$ result in closed form coordinate updates. The updates for the latent trajectories are presented in Algorithm \ref{alg:cavi_latent_space} of Appendix \ref{sec:variational_updates}.

\subsubsection{Updating $q(\Lambda_k)$}

Given the augmented model's conjugacy, the homophily coefficients will be Bernoulli for the reference layer and Gaussian for the other layers. The corresponding coordinate updates, which involve calculating the Bernoulli probabilities and performing standard Bayesian linear regression, are presented in Algorithm \ref{alg:cavi_lambda} of Appendix \ref{sec:variational_updates}.

\subsection{Convergence Criteria}\label{subsec:convergence}

Although it is possible to calculate the ELBO to determine convergence, evaluating the state space terms is computationally expensive. Instead, we monitor the expected log-likelihood
\begin{align}\label{eq:expect_loglik}
\mathcal{F}(q) = \sum_{k=1}^K \sum_{t=1}^T \sum_{j < i} (Y_{ijt}^k - 1/2) \Eq{q(\btheta)}{\psi_{ijt}^k} - \frac{1}{2} \Eq{q(\omega_{ijt}^k)}{\omega_{ijt}^k}\Eq{q(\btheta)}{(\psi_{ijt}^k)^2},
\end{align}
which upper bounds the ELBO. We say the algorithm converged when the difference in the expected log-likelihood is less than $10^{-2}$ between iterations or the number of iterations exceeded 1,000. Due to the ELBO's non-convexity, we run the algorithm with ten different random initializations and choose the model with the highest expected log-likelihood. For details on our initialization procedure and hyper-parameter settings, see Appendix \ref{sec:init}.

\subsection{Inference of Identifiable Parameters}\label{subsec:identifiable_parameters}

Recall that a centered latent space is a sufficient condition for parameter identifiability. As such, we make inference on the following parameters based on the approximate posterior:
\begin{align*}
\tilde{\bX}_t^i &= \bX_t^i - \frac{1}{n} \sum_{j=1}^n \bX_t^j, \\
\tilde{\delta}_{k,t}^i &= \delta_{k,t}^i + \tilde{\bX}_t^{i\, \rm T} \Lambda_k \bc + \frac{1}{2} \bc^{\rm T} \Lambda_k \bc,
\end{align*}
where $\bc = (1/n) \sum_{j=1}^n \bX_t^j$. Under our approximation, the marginal posterior distributions of the $\tilde{\bX}_t^i$'s are Gaussian with moments
\begin{equation}\label{eq:mu_tilde}
\begin{split}
\Eq{q(\btheta, \bphi, \bomega)}{\tilde{\bX}_t^i} &= \tilde{\bmu}_t^i = \bmu_t^i - \frac{1}{n} \sum_{j=1}^n \bmu_t^j, \\
\Var{\tilde{\bX}_t^i} &= \tilde{\Sigma}_t^i = \left(1 - \frac{1}{n}\right)^2 \Sigma_t^i + \left(\frac{1}{n}\right)^2 \sum_{j\neq i} \Sigma_t^j,
\end{split}
\end{equation}
where the variance is respect to $q(\btheta, \bphi, \bomega)$ as well. We calculate each $\tilde{\delta}_{k,t}^i$'s posterior mean and 95\% credible interval using 2,500 samples from the approximate posterior distribution because their approximate posterior distributions lack an analytic form.

\section{Simulation Studies} \label{sec:simulation}

This section presents a simulation study designed to assess the scaling of the proposed algorithm's estimation error and dyad-wise prediction error. We considered three scenarios: {\it Scenario 1.}\ an increase in the number of nodes with $(n,K,T) \in \set{50, 100, 200, 500, 1000} \times \set{5} \times \set{10}$, {\it Scenario 2.}\ an increase in the number of layers with $(n,K,T) \in \set{100} \times \set{5, 10, 20} \times \set{10}$, and {\it Scenario 3.}\ an increase in the number of time points with $(n,K,T) \in \set{100} \times \set{5} \times \set{10, 50, 100}$. For each scenario, we sampled 30 independent parameter settings as follows:
\begin{enumerate}
\item Generate the reference homophily coefficients: $\lambda_{1h} = 2 u_h - 1$ for $1 \leq h \leq d$, where $u_h \iidsim \Bern(0.5)$.
\item Generate the remaining homophily coefficients: $\blambda_k \iidsim U[-2, 2]^d$ for $2 \leq k \leq K$.
\item Generate initial sociality effects: $\delta_{k,1}^i \iidsim U[-4, 4]$ for $1 \leq i \leq n$ and $1 \leq k \leq K$.
\item Generate the social trajectories: For $t = 2, \dots, T$, sample $\delta_{k,t}^i \sim N(\delta_{k,t-1}^i, 0.1)$ for $1 \leq i \leq n$ and $1 \leq k \leq K$.
\item Generate initial latent positions: $\bX_1^i \iidsim N(0, 4 I_d)$ for $1 \leq i \leq n$.
\item Generate the latent trajectories: For $t = 2, \dots, T$, sample $\bX_t^i \sim N(\bX_{t-1}^i, 0.05 I_d)$ for $1 \leq i \leq n$.
\item Center the latent space: For $t=1, \dots, T$, set $\tilde{\bX}_t^i = \bX_t^i - (1/n) \sum_{j=1}^n \bX_t^j$ for $1 \leq i \leq n$.
\end{enumerate}
We set the dimension of the latent space $d = 2$. We sampled a single undirected adjacency matrix for each generated model using the dyad-wise probabilities in Equation (\ref{eq:logit_reg}).

To evaluate the estimated model's accuracy, we computed relative errors according to the Frobenius norm of the difference between the parameters' posterior means and their true values. Because the true homophily coefficients are distinct, Proposition \ref{prop:permutation} states that the latent positions are identifiable up to column permutations and sign-flips. To account for this invariance, we calculated the latent position's time-averaged relative error as
\begin{equation*}
\frac{1}{T} \sum_{t=1}^T \min_{P \in \Pi_d, \, \mathbf{s} \in \set{-1, 1}^{d}} \frac{\norm{\tilde{\mathcal{X}}_t -  \hat{\tilde{\mathcal{X}}}_t P \diag(\mathbf{s})}^2_F}{\norm{\tilde{\mathcal{X}_t}}_F^2},
\end{equation*}
where $\Pi_d$ is the set of permutation matrices on $d$ elements, $\norm{\cdot}_F$ is the Frobenius norm, $\tilde{\mathcal{X}}_t = (\tilde{\bX}_t^1, \dots, \tilde{\bX}_t^i)^{\rm T}$, and $\hat{\tilde{\mathcal{X}}}_t = (\tilde{\bmu}_t^1, \dots, \tilde{\bmu}_t^n)^{\rm T}$ where $\tilde{\bmu}_t^i$ is defined in Equation (\ref{eq:mu_tilde}). Similarly, we computed the relative error of the homophily coefficients accounting for invariance under simultaneous permutations of their rows and columns:
\begin{equation*}
\min_{P \in \Pi_d} \frac{\sum_{k=1}^K \norm{\Lambda_k - P^{\rm T}\diag(\bmu_{\blambda_k}) P}_F^2}{\sum_{k=1}^K \norm{\Lambda_k}_F^2}.
\end{equation*}
Lastly, we calculated the relative errors for the centered social trajectories and the dyad-wise probabilities, both of which do not have identifiability issues. For computational expediency, we calculated the dyad-wise probabilities by plugging-in the posterior means into Equation (\ref{eq:logit_reg}), e.g.,
\begin{equation*}
\widehat{\mathbb{P}}(Y_{ijt}^k = 1 \mid \mu_{\delta_{k,t}^i}, \mu_{\delta_{k,t}^j}, \bmu_{\blambda_k}, \bmu_t^i, \bmu_t^j) = \text{logit}^{-1}\left[\mu_{\delta_{k,t}^i} + \mu_{\delta_{k,t}^j} + \bmu_t^{i\rm T} \diag(\bmu_{\blambda_{k}}) \bmu_t^j\right],
\end{equation*}
which is an upper-bound on the approximate posterior mean of the dyad-wise probability. We can use Monte Carlo to estimate the dyad-wise probabilities’ approximate posterior mean by sampling from the approximate posterior if desired.

The estimation errors for varying $n$, $K$, $T$ are displayed in the boxplots in Figure \ref{fig:rel_error_nodes}, Figure \ref{fig:rel_error_layers}, and Figure \ref{fig:rel_error_time}, respectively. Overall, the CAVI algorithm recovers the model's parameters with high accuracy. The starkest improvement in estimation accuracy occurs as the number of nodes increases. This improvement is partly due to the more accurate estimation of the homophily coefficients. Due to the model's ability to pool information across layers, the latent positions' relative error decreases as $K$ increases. Such an improvement is not observed for the social trajectories because the number of social trajectories grows with the number of layers. Surprisingly, the homophily coefficients' estimation error does not improve as $T$ increases, although the estimation error is already low at roughly $10^{-3}$. Since the relative error of the latent positions and social trajectories is on the order of $10^{-2}$, we conclude that algorithm’s ability to estimate the latent positions and social trajectories accurately dominates the error. Furthermore, the latent positions' estimation error slightly degrades as the number of time steps increases. Such deterioration is typical in smoothing problems.

\begin{figure}[tbp!]
\centering
\includegraphics[width=\textwidth]{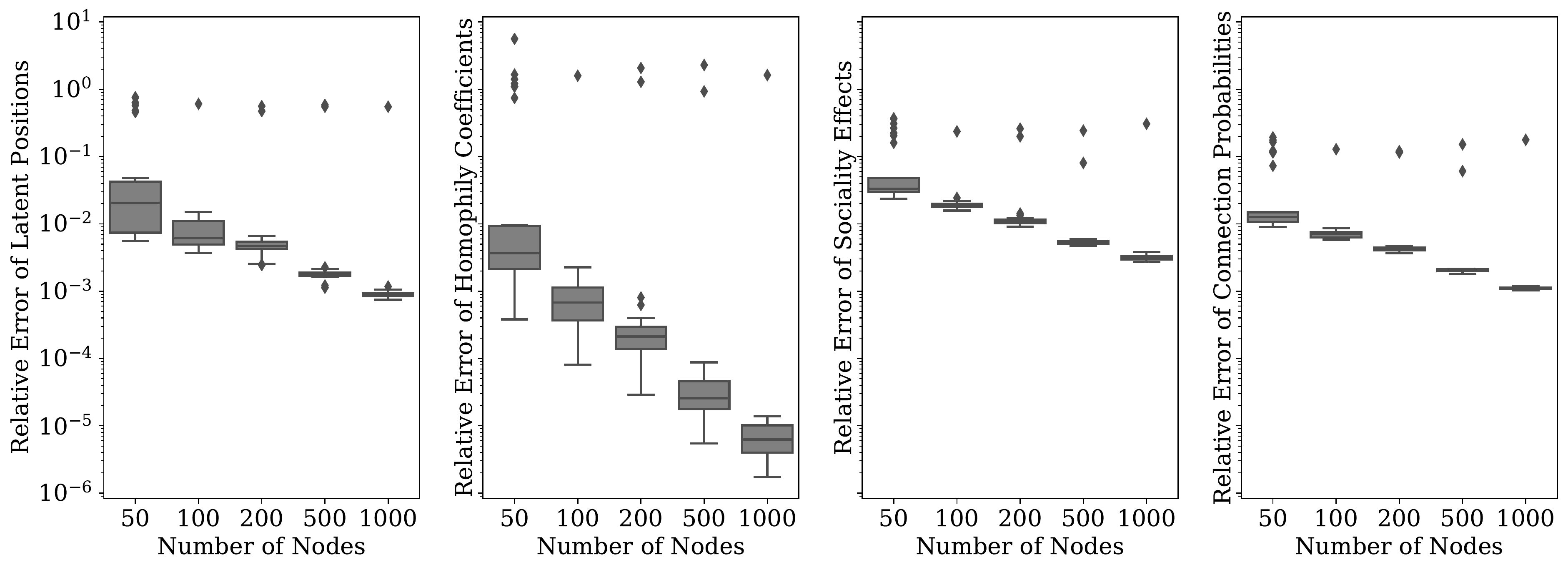}
\caption{Relative estimation errors of the model's parameters as the number of nodes $n$ increases. Boxplots show the distribution over 30 simulations.}
\label{fig:rel_error_nodes}
\end{figure}

\begin{figure}[tbp]
\centering
\includegraphics[width=\textwidth]{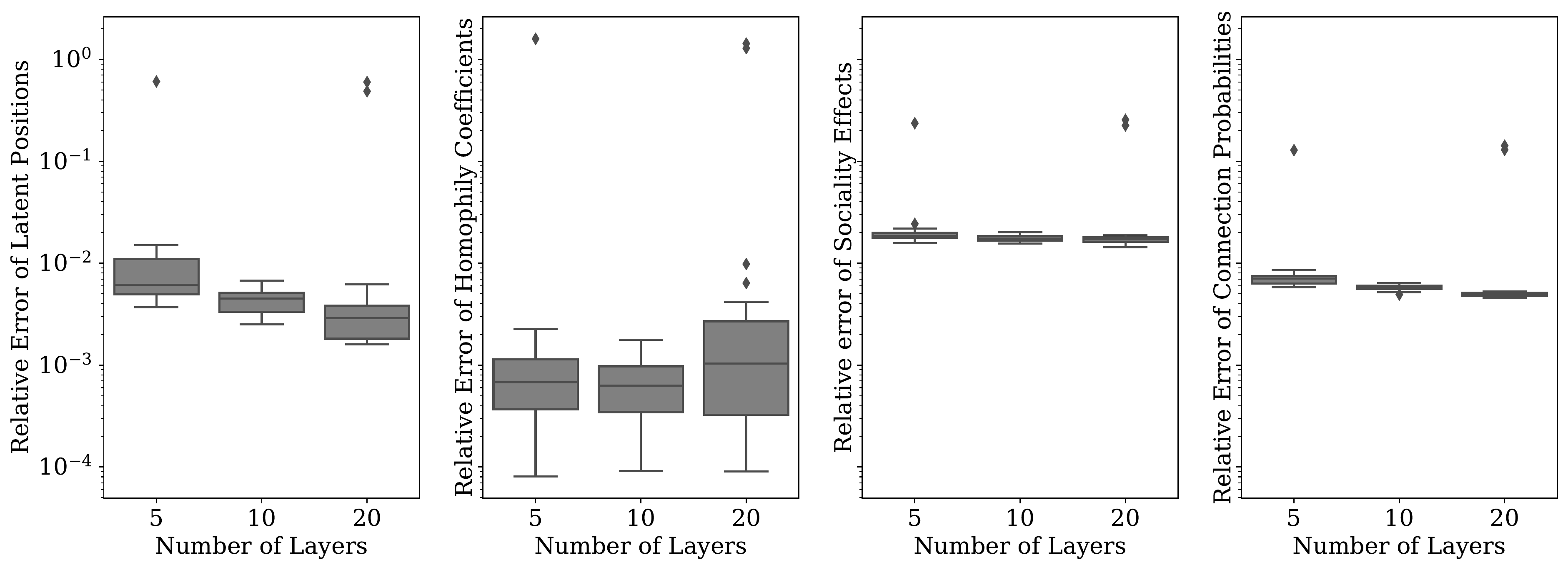}
\caption{Relative estimation errors of the model's parameters as the number of layers $K$ increases. Boxplots show the distribution over 30 simulations.}
\label{fig:rel_error_layers}
\end{figure}

\begin{figure}[tbp]
\centering
\includegraphics[width=\textwidth]{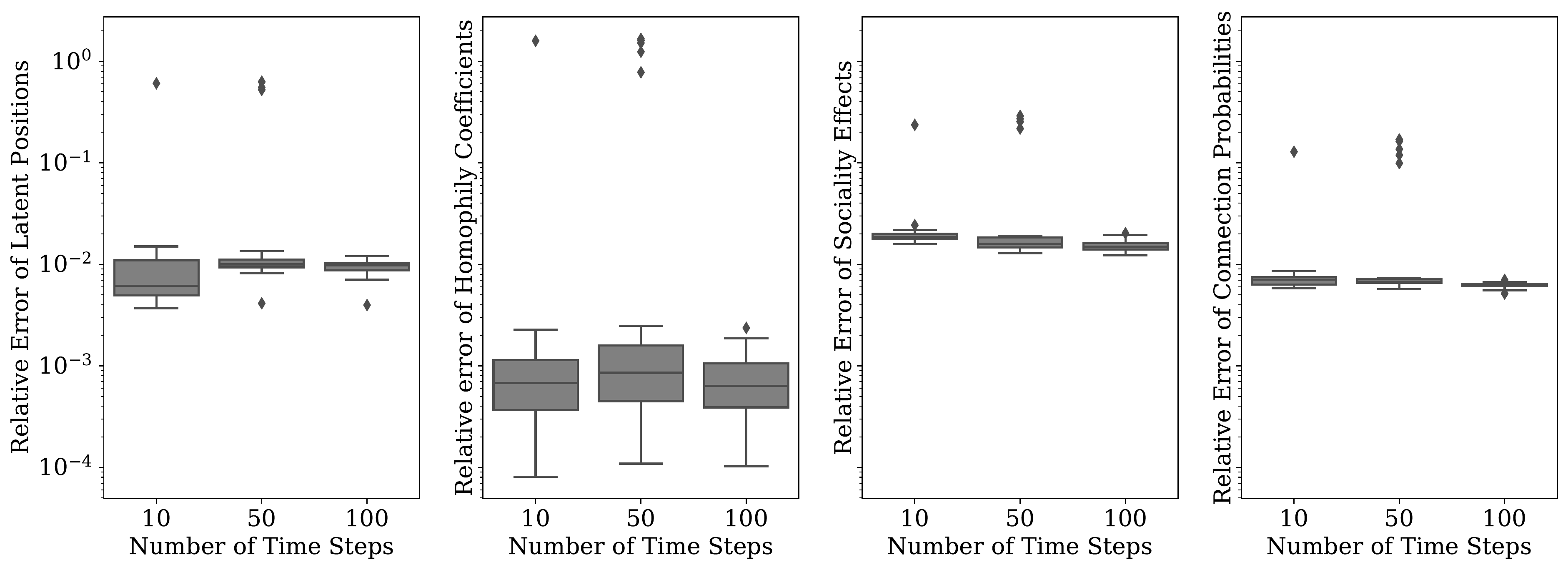}
\caption{Relative estimation errors of the model's parameters as the number of time steps $T$ increases. Boxplots show the distribution over 30 simulations.}
\label{fig:rel_error_time}
\end{figure}

Next, we evaluated the model's predictive performance by calculating the area under the receiver operating characteristic curve (AUC) for in-sample and held-out dyads. To evaluate held-out predictions, we removed 20\% of the dyads randomly from each layer and time step during estimation. Figure \ref{fig:perf_auc}'s boxplots summarize the prediction errors for increasing $n$, $K$, and $T$. The in-sample and holdout AUC are close to the maximum value of one for all scenarios. Overall, the simulations demonstrate that the CAVI algorithm is scalable and accurate.

\begin{figure}[tbp]
\centering
\includegraphics[width=\textwidth]{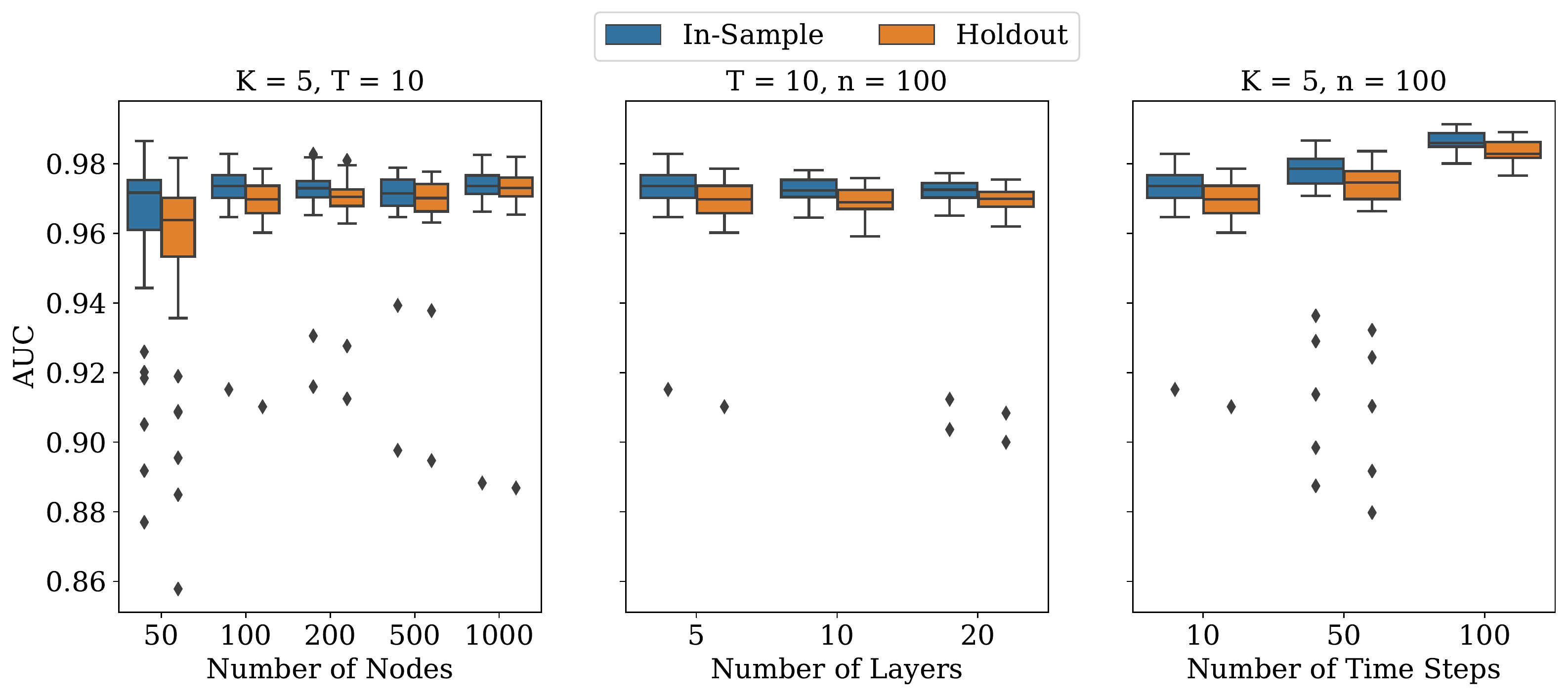}
\caption{Dyad-wise predictive performance measured by the in-sample and holdout AUC as the number of nodes $n$ (left), the number of layers $K$ (center), and the number time steps $T$ (right) increases. Boxplots show the distribution over 30 simulations.}
\label{fig:perf_auc}
\end{figure}

\section{Real Data Applications}\label{sec:realdata}

In this section, we demonstrate how to use the proposed model to analyze real-world data sets. We consider networks from political science and epidemiology. The first example studies a time series of different international relations between 100 countries over eight years. The second example applies the model to a contact network of 242 individuals at a primary school measured over two days to quantify heterogeneities in infectious disease spread throughout the school day.

\subsection{International Relations} \label{subsec:icews}

This application explores the temporal evolution of different relations between socio-political actors. The raw data consists of (source actor, target actor, event type, time-stamp) tuples collected by the Integrated Crisis Early Warning System (ICEWS) project \citep{icews}, which automatically identifies and extracts international events from news articles. The event types are labeled according to the CAMEO taxonomy~\citep{gerner2008}. The CAMEO scheme includes twenty labels ranging from the most neutral ``1 --- make public statement" to the most negative ``20 --- engage in unconventional mass violence."

Our sample consists of monthly event data between countries during the eight years of the Obama administration (2009 - 2017). We grouped the event types into four categories known as ``QuadClass''~\citep{duval1980}. These classes split events along four dimensions: (1) {\it verbal cooperation} (labels 2 to 5), (2) {\it material cooperation} (labels 6 to 7), (3) {\it verbal conflict} (labels 8 to 16), and (4) {\it material conflict} (labels 17 to 20). At a high-level, the first two classes represent friendly relations such as ``5 --- engage in diplomatic cooperation" and ``7 --- provide aid", while the last two classes reflect hostile relations such as ``13 --- threaten" and ``19 --- assault".

\subsubsection{Statistical Network Analysis of the ICEWS Data}

We structured the ICEWS data as a dynamic multilayer network recording which four event types occurred between nations each month from 2009 until the end of 2016. Each event type is a layer in the multilayer networks. We chose verbal cooperation as the reference layer because it contains the densest networks. We limited the actors to the 100 most active countries during this period. This preprocessing resulted in a dynamic multilayer network with $K = 4$ layers, $T = 96$ time steps, and $n = 100$ actors. An edge ($Y_{ijt}^k = 1$) means that country $i$ and country $j$ had at least one event of type $k$ during the $t$th month, where $t = 1$ corresponds to January 2009. We fit the model using the procedure described in Section~\ref{sec:estimation}. The model's in-sample AUC was 0.90, which indicates a good fit to the data.

\subsubsection{Detection of Historical Events During the Obama Administration}

We validate the model by demonstrating that the inferred social trajectories and latent space dynamics reflect major international events. We focus on three events: the Arab Spring, the American-led intervention in Iraq, and the Crimea Crisis. Specifically, we concentrate on interpreting the latent parameters for Libya, Syria, Iraq, the United States, Russia, and Ukraine since they played a large role in these events.

Because these events involve conflict, we start by analyzing each country's material conflict social trajectory, i.e., $\delta_{4,1:T}^i$ for $1 \leq i \leq n$. Figure \ref{fig:icews_social_trajectories} plots these social trajectories' posterior means with a few select countries highlighted. Appendix \ref{app:figures} contains the same plot for the remaining three layers. Most social trajectories are relatively flat. Indeed, the 95\% credible interval for the step size standard deviation $\sigma_{\delta}$ is $(0.0619, 0.0628)$, which is much smaller than that of the initial standard deviation $\tau_{\delta}$, which equals $(2.10, 2.42)$. However, the social trajectories of Iraq, Syria, and Libya demonstrate dramatic changes. Specifically, Libya and Syria both increase their material conflict sociality at the start of the Arab Spring in 2011. In particular, Libya's sociality spikes during the Libyan Civil War in 2011 that saw Muammar Gaddafi's regime overthrown. Iraq's sociality increases leading up to and throughout the United States' escalated military presence in 2014. Note that the Crimea Crisis, which began with Russia annexing the Crimea Peninsula in February 2014, is not reflected in Ukraine’s or Russia's social trajectory. This conflict is missing because an actor's social trajectory reflects their global standing in the network while the Crimea Crisis is primarily a regional conflict. In contrast, the latent space, which captures local transitive effects, should reflect this more localized conflict.

\begin{figure}[bt!]
\centering
\includegraphics[width=0.85\textwidth]{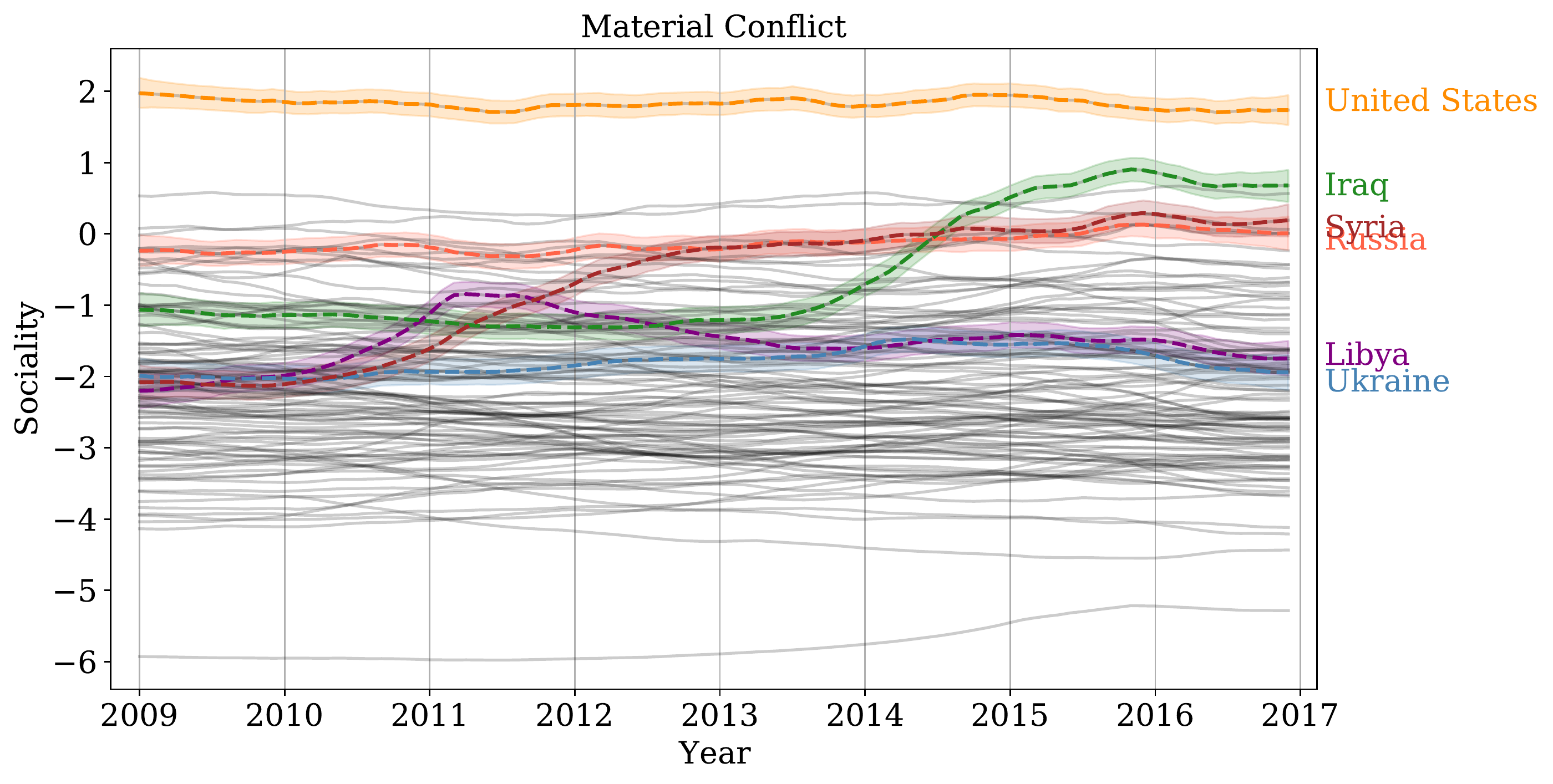}
\caption{Posterior means of the material conflict social trajectories. Select countries are highlighted in color with bands that represent 95\% credible intervals. The remaining countries' social trajectories are displayed with gray curves.}
\label{fig:icews_social_trajectories}
\end{figure}

We begin analyzing the latent space by interpreting the estimated homophily coefficients, $\Lambda_k$ (Figure \ref{fig:icews_homophily}). All layers exhibit assortativity along both latent dimensions. Interestingly, we notice similarities in how the cooperation and the conflict layers use the latent space. The homophily coefficients' 95\% credible intervals overlap along the first dimension for the verbal conflict and the material conflict layers. Also, the credible intervals overlap along the second dimension for the verbal cooperation and the material cooperation layers. Furthermore, the conflict layers have larger homophily coefficients than the cooperation layers. To interpret this result, we visualize the latent space's layout.

\begin{figure}[bt!]
\centering
\includegraphics[width=0.9\textwidth]{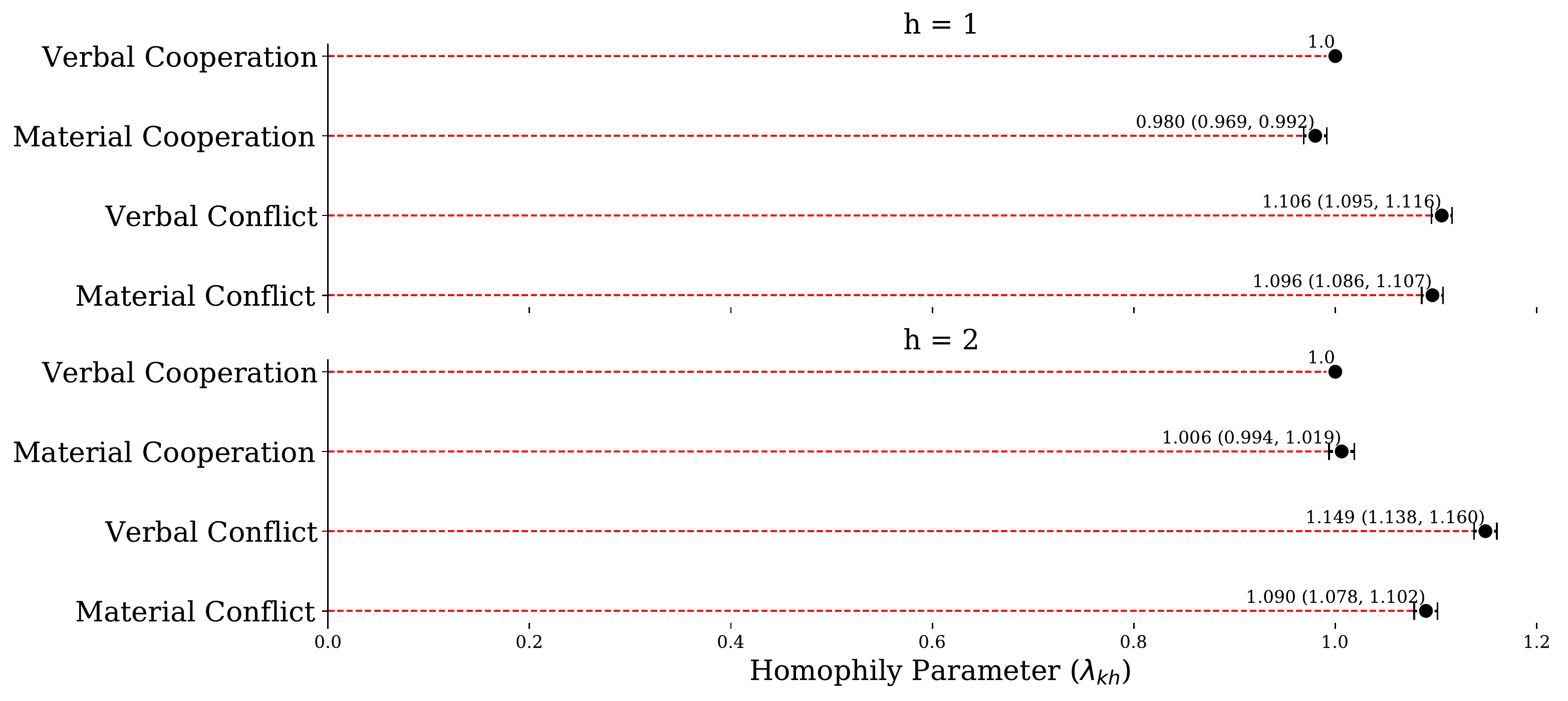}
\caption{The homophily coefficients' posterior means and 95\% credible intervals for the ICEWS network's four relations. The top and bottom plots give estimates for the degree of homophily along the first and second latent dimensions, respectively.}
\label{fig:icews_homophily}
\end{figure}

\begin{sidewaysfigure}[p!]
\centering
\includegraphics[width=\textwidth]{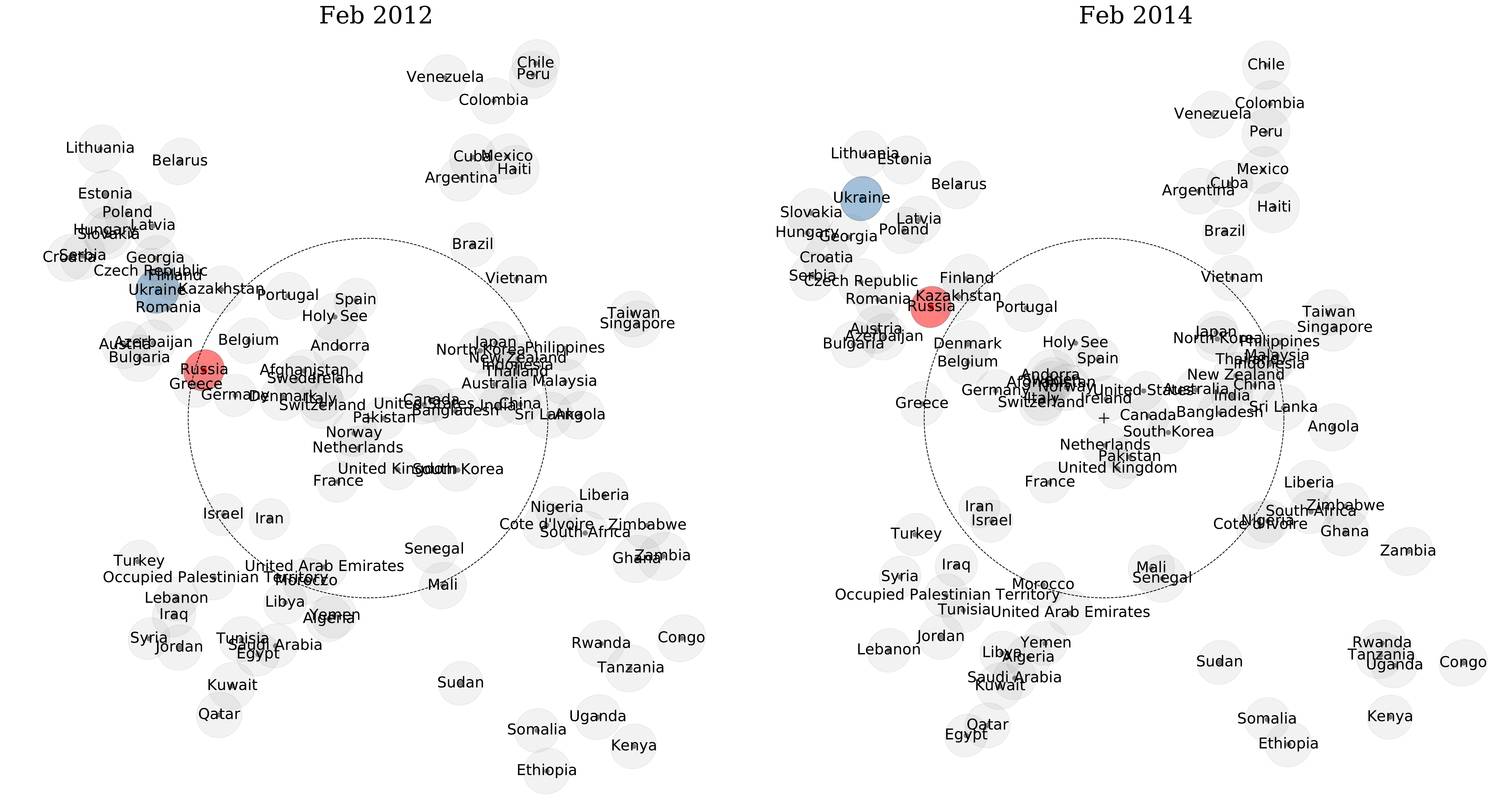}
\caption{Estimated latent space for the ICEWS networks on February 2012 (left) and February 2014 (right). The origin of the latent space is denoted by a $+$ and the initial variance $\tau^2$ is displayed as a dotted one-standard deviation ellipse. The names of each nation are annotated. The ellipses are two standard deviation ($\sim 95\%$) credible ellipses for each actor's latent position. Ukraine and Russia are highlighted in blue and red, respectively.}
\label{fig:icews_networks}
\end{sidewaysfigure}

Figure \ref{fig:icews_networks} displays the estimated latent space during February 2012 and February 2014. The latent space encodes the geographic locations of the countries. Due to the positive homophily of the relations, actors are more likely to connect when their latent positions share a common angle. Eastern European nations are on the top left, Latin American nations are on the top right, African nations are on the bottom right, and Middle Eastern nations are on the bottom left. Furthermore, highly sociable nations, such as the United States, are near the center of the latent space because their high sociality explains most of their interactions. Overall, we conclude that the higher values of the conflict homophily coefficients indicate that regional (geographic) effects play a more prominent role in predicting conflict than cooperation.

Finally, we demonstrate how the latent space reflects the regional Crimea Crisis between Russia and Ukraine in early 2014. Figure \ref{fig:icews_latent_trajectories} displays the latent trajectories for the two nations. Unlike the actor's social trajectories, their latent trajectories are highly variable and encompass the Crimea Crisis. Around the second half of 2013, Ukraine's latent feature along the second dimension increases significantly, reaching a maximum in early 2014. During this time, Russia's second latent feature also increased. Comparing Ukraine and Russia's latent positions in February 2012 to those in February 2014 in Figure~\ref{fig:icews_networks}, we see that they align themselves while moving toward the periphery of the latent space. These dynamics result in an increased connection probability between the two nations in all layers during the crisis, see Figure \ref{fig:icews_probas}. Overall, we conclude that the latent trajectories reflect regional events in the ICEWS data.

\begin{figure}[b!]
\centering
\includegraphics[width=\textwidth]{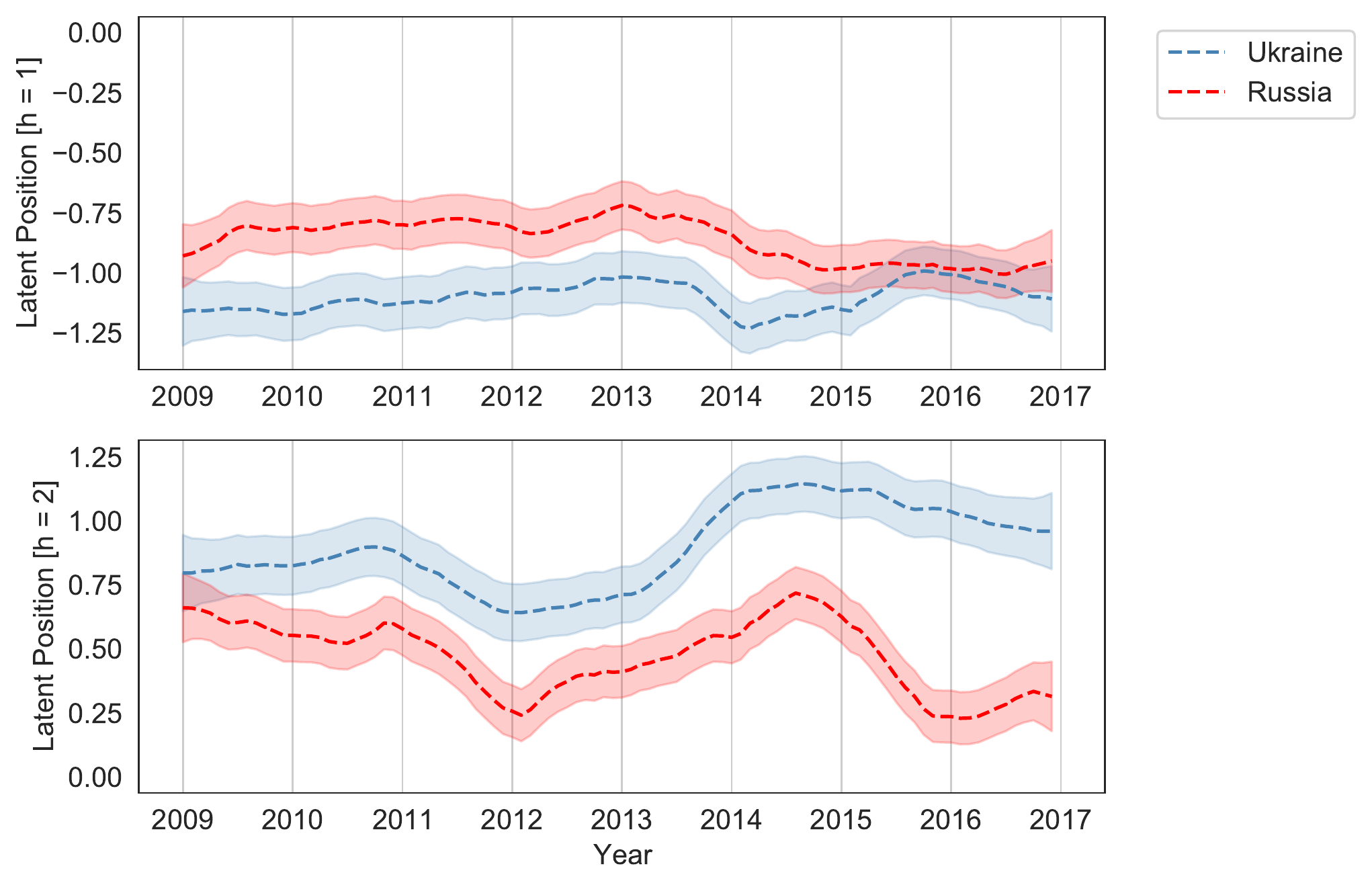}
\caption{Posterior means and 95\% credible intervals of Ukraine and Russia's latent trajectories. The top and bottom plots give estimates for the first and second latent dimensions, respectively.}
\label{fig:icews_latent_trajectories}
\end{figure}

\begin{figure}[bt!]
\centering
\includegraphics[width=\textwidth]{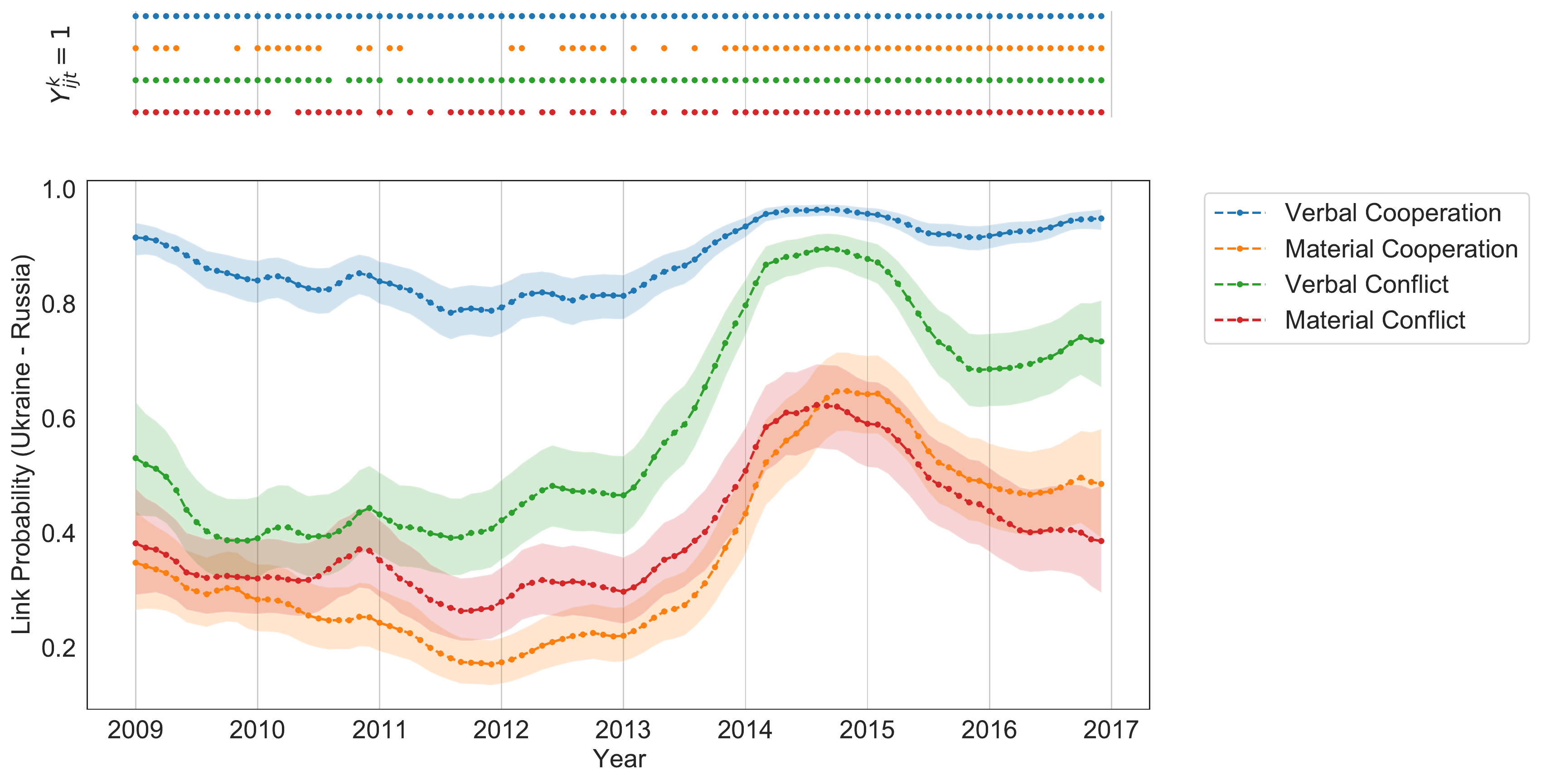}
\caption{Raw values of the adjacency matrices where a dot indicates that Ukraine and Russia had a particular relation during that month (top). The posterior means and 95\% credible intervals for the monthly link probability between the two nations across the four international relations (bottom). Estimates are calculated using 1,000 samples from the approximate posterior.}
\label{fig:icews_probas}
\end{figure}

\subsection{Epidemiological Face-to-Face Contact Networks} \label{subsec:epi}

This case study uses our proposed model to analyze longitudinal face-to-face contact networks drawn from an epidemiological survey of students at a primary school (grades 1 to 5) in Lyon, France. Such contact networks influence mathematical models of infectious disease spread in varying populations~\citep{wallinga2006, zagheni2008}. Also, the analysis of these contact patterns allows school administrators to mitigate infectious disease spread in classrooms by determining the times during the day when spread is most prevalent. In the exploratory phase, these analyses often have difficulty visualizing the complicated dynamic networks. Furthermore, they often do not formally quantify the uncertainty in network statistics. In this section, we demonstrate how our model provides a meaningful network visualization and quantification of uncertainty.

The contact networks were collected by the SocioPatterns collaboration~(\url{http://www.sociopatterns.org}) and initially analyzed in \citet{stehle2011}. Contact data is available for 242 individuals (232 children and 10 teachers) belonging to grades 1 through 5. Each grade is split into two sections (A and B) so that there are ten classes overall. Each class has its own classroom and teacher. The school day runs from 8:30 am to 4:30 pm, with a lunch break from 12:00 pm to 2:00 pm and two breaks of 20 to 25 minutes around 10:30 am and 3:30 pm.

The face-to-face contacts occurred over two days: Thursday, October 1st, 2009, and Friday, October 2nd, 2009. Data was collected from 8:45 am to 5:20 pm on the first day and from 8:30 am to 5:05 pm on the second day. Radio-frequency identification (RFID) devices measured the contacts between individuals. The RFID sensor registered a contact when two individuals were within 1 to 1.5 meters during a 20-second interval. This distance range was chosen to correspond to the range over which a communicable infectious disease could spread. For a detailed description of the data collection technology, see~\citet{cattuto2010}.

\subsubsection{Statistical Network Analysis of the School Contact Network}

We structured the face-to-face contact data as a dynamic multilayer network recording face-to-face interactions each day. We treated each day as a layer so that the layers correspond to Thursday and Friday. We set Thursday as the reference layer. In concordance with the analysis in \citet{stehle2011}, we divided the daily contact networks into 20-minute time intervals between 9:00 am and 5:00 pm and extended the first and last time intervals to accommodate the different starting and ending times of the experiment on the two days. This preprocessing resulted in a dynamic multilayer network with $K = 2$ layers, $T = 24$ time steps, and $n = 242$ actors. Specifically, an edge ($Y_{ijt}^k = 1$) means that actor $i$ and actor $j$ had at least one registered interaction during the $t$th 20-minute interval on day $k$. We fit the model using the procedure detailed in Section~\ref{sec:estimation}. The model's in-sample AUC was 0.96, which indicates a good fit to the data.

\subsubsection{Dynamics of the Epidemic Branching Factor}

Here, we demonstrate how to use our model to (1) determine periods in the school day most susceptible to the spread of infectious disease and (2) identify differences in the contact patterns between the two days. To quantify a network's contribution to the spread of infectious disease, we use the epidemic branching factor~\citep{andersson1998}, defined as
\begin{equation*}
\kappa = \frac{\sum_{i=1}^n d_i^2 / n}{\sum_{i=1}^n d_i / n},
\end{equation*}
where $d_i$ is the $i$th node's degree. The epidemic branching factor is related to the basic reproduction number, $R_0$, which is (loosely) equal to the number of secondary infections caused by a typical infectious individual during an epidemic's early stages~\citep{anderson1991}. In network-based susceptible-infected-recovered (SIR) models, $R_0$ equals $\tau (\kappa - 1) / (\tau + \gamma)$, where $\tau$ and $\gamma$ are infection and recovery rates, respectively~\citep{andersson1997}. This relation implies that larger branching factors lead to more massive epidemics.

\begin{figure}[bt]
\centering
\includegraphics[width=0.9\textwidth]{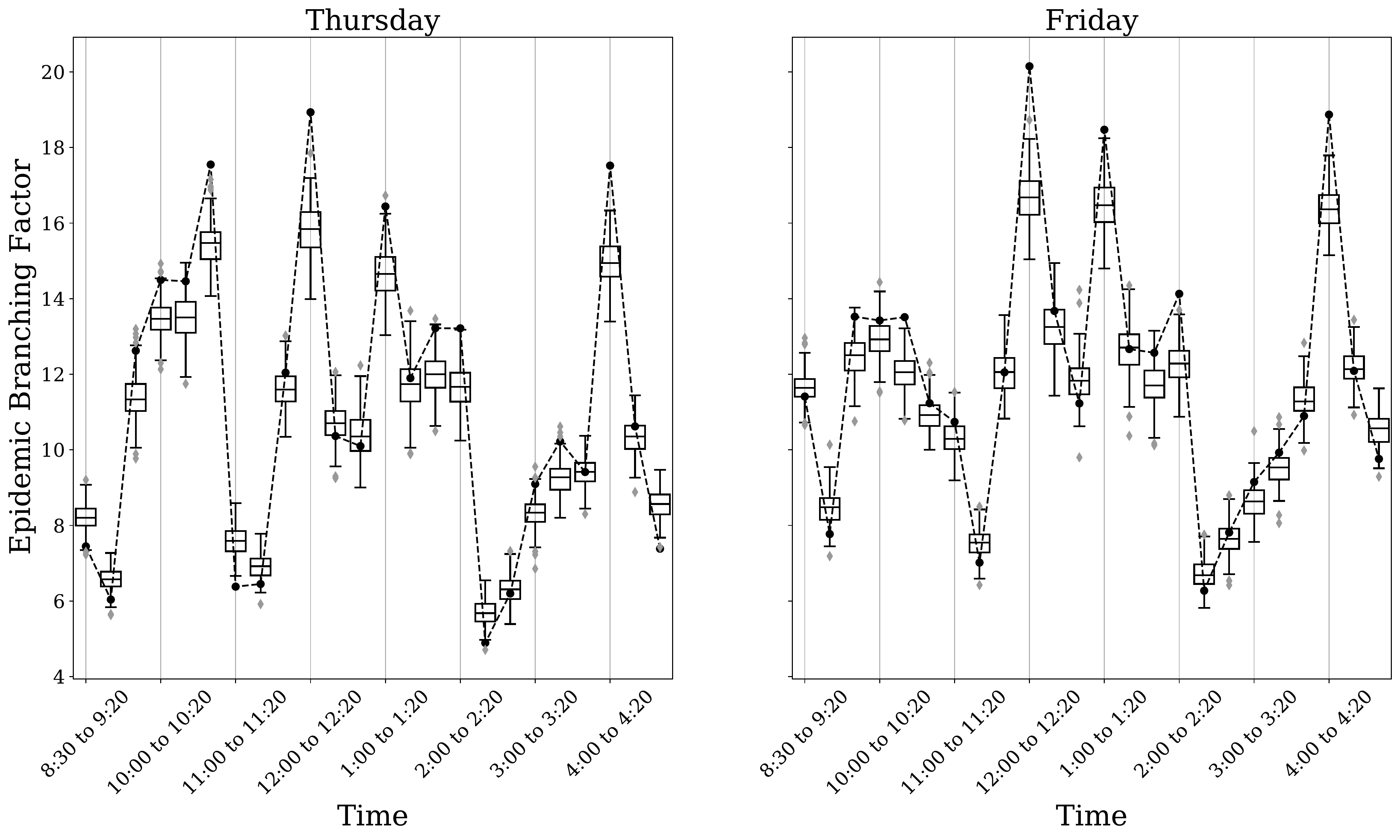}
\caption{Epidemic branching factors for the face-to-face contact networks on Thursday (left) and Friday (right) at different times throughout the school day. The dashed black curves depict the observed network's branching factor. Boxplots show the range of the branch factor's posterior distribution.}
\label{fig:primary_branching_factor}
\end{figure}

Figure~\ref{fig:primary_branching_factor} depicts the posterior distribution of the epidemic branching factor. The boxplots contain 250 networks, each sampled from a different set of latent variables drawn from the model’s approximate posterior. The model matches the observed branching factor for most time steps; however, it underestimates the most dramatic changes at 10:40 am to 11:00 am, 12:00 pm to 12:20 pm, 1:00 pm to 1:20 pm, and 4:00 pm to 4:20 pm. Regardless, the model still captures these four spikes in the branching factor. Intuitively, the timings of these spikes occur during lunchtime (12:00 pm to 2:00 pm) and the two short breaks (around 10:30 am and 3:30 pm). We expect such events to lead to increased disease spread because they allow students from different classrooms to mix. More surprisingly, the branching factor's dynamics differ between the two days. The most apparent difference is the spike from 10:40 am to 11:00 am on Thursday that is not present on Friday. The difference in branching factors between Thursday and Friday from 10:40 am to 11:00 am is significantly greater than zero, with the difference's 95\% credible interval equaling (3.11, 5.78). To understand what caused this difference, we analyzed the shared latent space. We defer a discussion of the actor's social trajectories to Appendix \ref{app:figures}.

Figure \ref{fig:primary_latent_space} depicts the latent positions' posterior means and the observed edges on Thursday and Friday during the first short break from 10:40 am to 11:00 am. The inferred homophily coefficients are all positive and significantly different between layers (see Figure \ref{fig:primary_homophily} in Appendix \ref{app:figures}). The latent space accurately clusters the students into their ten classrooms. The two layers share the same classroom structure, which affirms our choice of a shared latent space. The difference in branching factors is due to the varying mixing patterns between the classrooms on the two days. Specifically, the classrooms that interact on the two days are different. On Thursday, there are many contacts between students in classes 1A, 1B, 2A, 3A, 3B, and 4B. In contrast, on Friday, classes 1A, 2A, 2B, 4B, and 5B interact. Furthermore, the number of edges between classrooms is much lower on Friday than on Thursday. This observation implies a simple intervention to mitigate disease spread: stagger each classroom's break time in order to limit contacts between students of different classes, which will lower the epidemic branching factor.

\begin{figure}[bt]
\centering
\includegraphics[width=\textwidth]{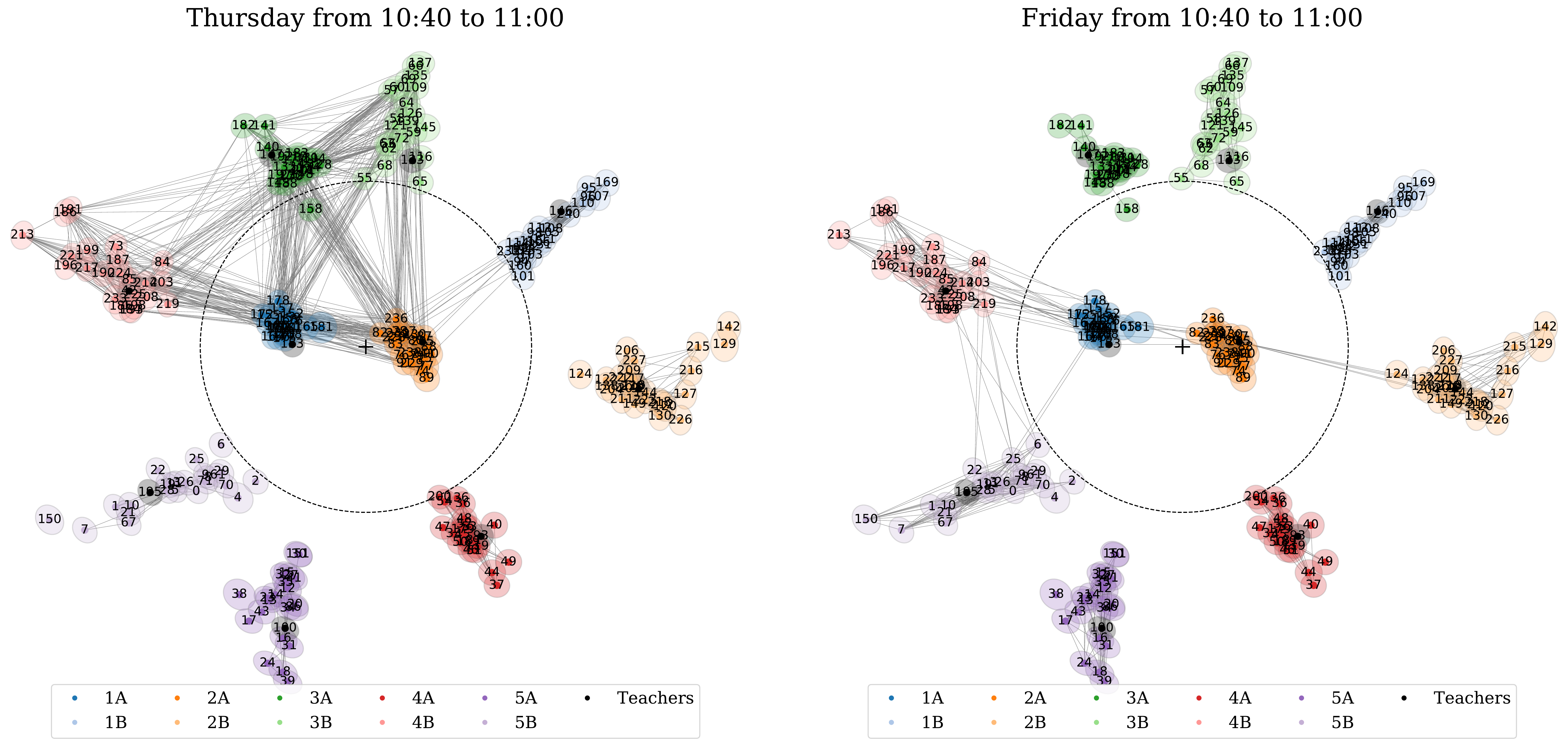}
\caption{Latent positions' 95\% credible ellipses for the primary school face-to-face contact networks from 10:40 am to 11:00 am. The gray lines indicate observed edges on Thursday (left) and Friday (right). The students are colored by their classroom and section, while teachers are displayed in black. The `+' denotes the origin of the latent space and the dotted circle indicates a one standard deviation ellipse with variance $\tau^2$.}
\label{fig:primary_latent_space}
\end{figure}

\section{Discussion}\label{sec:disc}

This article proposed a flexible, interpretable, and computationally efficient latent space model for dynamic multilayer networks. Our eigenmodel for dynamic multilayer networks decomposes the dyadic data into a common time-varying latent space used differently by the layers through layer-specific homophily levels and additive node-specific social trajectories that account for further degree heterogeneity. Also, we determined and corrected for various identifiability issues. This accomplishment allows for an intuitive interpretation of the latent space, unlike previous nonparametric models~\citep{durante2017}. Next, we developed an efficient variational inference algorithm for parameter estimation. Unlike previous variational approaches, we maintain the essential temporal dependencies in the posterior approximation. Furthermore, our variational algorithm is widely applicable to general dynamic bilinear latent space models. A simulation study established the effectiveness of our estimation procedure to scale to various network sizes. Finally, we demonstrated how to use our model to analyze international relations from 2009 to 2017 and understand the spread of an infectious disease in a primary school contact network.

In this work, we always set the latent space dimension $d = 2$, which allows for visualization; however, one may want a data-driven choice of $d$. One possibility is to use information criteria such as the Akaike information criteria (AIC), deviance information criteria (DIC), or Bayesian information criteria (BIC) to perform model selection. When the purpose of the model is to predict unobserved dyads, cross-validation procedures are a reasonable solution. In this case, we can either perform dyad-wise $V$-fold cross-validation~\citep{hoff2005} or network cross-validation~\citep{chen2018}. Lastly, one can examine the posterior predictive distribution of statistics of interest, $T(\bY_{1:T}^1, \dots, \bY_{1:T}^K)$, and select the smallest $d$ such that there is no substantial lack of fit. Such posterior predictive checks are standard in the social network literature~\citep{hunter2008}. A theoretically sound and easy-to-compute model selection criteria would be beneficial for bilinear LSMs.

Many real-world networks contain non-binary relations. One can adopt the proposed model to networks with non-binary edges with minor changes. For example, replacing the Bernoulli likelihood in Equation (\ref{eq:logit_reg}) with a Gaussian likelihood can model real-valued networks with minimal changes to the variational algorithm. However, extending the variational algorithm to general exponential family likelihoods, such as Poisson or negative binomial, is a direction for future research.

Relations are often directed in nature; therefore, it is natural to generalize the model to directed networks. Such a model needs to allow for varying levels of reciprocity in the directed relations. A simple extension of our model to directed networks is
\begin{equation*}
Y_{ijt}^k \indsim \Bern\left(\text{logit}^{-1}\left[\Theta_{ijt}^k\right]\right), \quad \text{ with } \quad \Theta_{ijt}^k = \delta_{k,t}^i + \gamma_{k,t}^j + \bX_t^{i \, \rm T} \Lambda_k \bZ_t^j,
\end{equation*}
where $Y_{ijt}^k = 1$ ($Y_{ijt}^k = 0$) denotes the presence (absence) of a directed edge from $i$ to $j$ in layer $k$ at time $t$. The latent variables' distributions are
\begin{equation*}
\gamma_{k,1}^i \iidsim N(0, \tau_{\gamma}^2), \quad \gamma_{k,t}^i \sim N(\gamma_{k,t-1}^i, \sigma^2_{\gamma}), \quad \bZ_1^i \iidsim N(0, \tau^2_z I_d), \quad \bZ_{t}^i \sim N(\bZ_{t-1}^i, \sigma^2_z I_d),
\end{equation*}
and the priors on the remaining parameters are left unchanged from the undirected case. In this case, $\delta_{k,t}^i, \gamma_{k,t}^i \in \mathbb{R}$ model degree heterogeneity in outgoing and incoming edges, respectively. The asymmetric latent positions $\bX_t^i, \bZ_t^i \in \mathbb{R}^d$ allow an actor's features to differ depending on whether they are receiving or initiating the relation. The variational inference algorithm for this model remains mostly unchanged. However, a model that does not drastically increase the number of parameters compared to the undirected case, such as the one in \citet{sewell2015}, is an area of research interest.

Further research directions include increasing the algorithm's scalability through stochastic variational inference~\citep{hoffman2013, aliverti2020} and exploring the variational estimates' asymptotics. Overall, our proposed eigenmodel for dynamic multilayer networks is an interpretable statistical network model with applications to various real-world scientific problems. A repository for the replication code is available on Github~\citep{multidynet2021}.


\acks{This work was supported in part by National Science Foundation grant DMS-2015561 and a grant from Sandia National Laboratories.}

\appendix
\section{Proofs of Propositions \ref{prop:identify}, \ref{prop:indefinite_ortho}, and \ref{prop:permutation}}\label{subsec:proofs}

This section demonstrates the identifiability of our model under the conditions proposed in Propositions \ref{prop:identify}, \ref{prop:indefinite_ortho}, and \ref{prop:permutation}. Before stating the proofs, we need the following lemma.

\begin{lemma}\label{lemma:sum_zero}
For any $\bv = (v_1, \dots, v_n)^{\rm T} \in \Reals{n}$, if $\bv \mathbf{1}_n^{\rm T} \mathbf{1}_n + \mathbf{1}_n \bv^{\rm T} \mathbf{1}_n = 0$, then $\bv = 0$.
\end{lemma}
\begin{proof}
The condition can be written as
\begin{equation*}
n\begin{pmatrix}
v_1 \\
\vdots \\
v_n
\end{pmatrix} +
\begin{pmatrix}
\sum_{i=1}^n v_i \\
\vdots \\
\sum_{i=1}^n v_i
\end{pmatrix} = \begin{pmatrix}
0 \\
\vdots \\
0
\end{pmatrix},
\end{equation*}
which implies $v_1 = \dots = v_n = -(1/n) \sum_{i=1}^n v_i$. Thus, we have $\bv = 0$.
\end{proof}

\begin{proof}[Proof of Proposition \ref{prop:identify}]

We begin by showing that under Assumption \ref{assmp:A1}, the social trajectories $\bdelta_{k,t}$ are identifiable. Under Assumption \ref{assmp:A1}, $J_{n} \mathcal{X}_t = \mathcal{X}_t$ and $J_{n} \tilde{\mathcal{X}}_t = \tilde{\mathcal{X}}_t$, which implies that $\mathcal{X}_t \Lambda_k \mathcal{X}_t^{\rm T} \mathbf{1}_{n} = \tilde{\mathcal{X}}_t \Lambda_k \tilde{\mathcal{X}}_t^{\rm T}\mathbf{1}_{n} = 0$. Now assume two sets of parameters satisfy
\begin{equation}\label{eq:identify}
\bdelta_{k,t}\bone_n^{\rm T} + \bone_n \bdelta_{k,t}^{\rm T} + \cX_t \Lambda_k \cX_t^{\rm T} = \tilde{\bdelta}_{k,t} \bone_n^{\rm T}+ \bone_n \tilde{\bdelta}_{k,t}^{\rm T} + \tilde{\cX}_t \tilde{\Lambda}_k \tilde{\cX}_t^{\rm T}
\end{equation}
for $k = 1, \dots, K$ and $t = 1, \dots, T$. Right multiplying $\mathbf{1}_{n}$ on both sides of the above equation gives
\begin{equation*}
\bdelta_{k,t} \mathbf{1}_{n}^{\rm T} \mathbf{1}_{n} + \mathbf{1}_{n} \bdelta_{k,t}^{\rm T} \mathbf{1}_{n} = \tilde{\bdelta}_{k,t} \mathbf{1}_{n}^{\rm T} \mathbf{1}_{n} + \mathbf{1}_{n} \tilde{\bdelta}_{k,t}^{\rm T} \mathbf{1}_{n},
\end{equation*}
or
\begin{equation*}
(\bdelta_{k,t} - \tilde{\bdelta}_{k,t}) \mathbf{1}_{n}^{\rm T} \mathbf{1}_{n} + \mathbf{1}_{n} (\bdelta_{k,t} - \tilde{\bdelta}_{k,t})^{\rm T} \mathbf{1}_{n} = 0.
\end{equation*}
Applying Lemma \ref{lemma:sum_zero}, we conclude that
\begin{equation}\label{eq:soc_ident}
\bdelta_{k,t} = \tilde{\bdelta}_{k,t}
\end{equation}
for all $k$ and $t$.

Now, we focus on the identifiability of the latent space and the homophily coefficients. By Assumption \ref{assmp:A3}, for the reference layer $r$, Equation (\ref{eq:identify}) and Equation (\ref{eq:soc_ident}) imply
\begin{equation}\label{eq:dist_match}
\cX_t I_{p,q} \cX_t^{\rm T} = \tilde{\cX}_t I_{p',q'} \tilde{\cX}_t^{\rm T}.
\end{equation}
By Assumption \ref{assmp:A2}, $\cX_t$ and $\tilde{\cX}_t$ are full rank so have left inverses $B$ and $\tilde{B}$, respectively. In other words, $B \cX_t = \tilde{B} \tilde{\cX}_t = I_d$. Multiplying Equation (\ref{eq:dist_match}) on the right by $\tilde{B}^{\rm T} I_{p', q'}$, we have that
\begin{align}\label{eq:latent_match}
\tilde{\cX}_t &= \cX_t I_{p,q} \cX_t^{\rm T} \tilde{B}^{\rm T} I_{p',q'} = \cX_t M_t,
\end{align}
where $M_t = I_{p,q} \cX_t^{\rm T} \tilde{B}^{T} I_{p',q'} \in \Reals{d \times d}$. More generally, for all layers $k \in \set{1, \dots K}$, we have
\[
\cX_t \Lambda_k \cX_t^{\rm T} = \tilde{\cX}_t \tilde{\Lambda}_k \tilde{\cX}_t^{\rm T} = \cX_t M_t \tilde{\Lambda}_k M_t^{\rm T} \cX_t^{\rm T},
\]
where the last equality used the identity in Equation (\ref{eq:latent_match}). Multiplying each side of the previous identity on the left by $B$ and on the right by $B^{\rm T}$, we conclude that
\begin{equation}\label{eq:inv_lambda}
\Lambda_k = M_t \tilde{\Lambda}_k M_t^{\rm T}.
\end{equation}
However, we want a transformation that takes $\Lambda_k$ to $\tilde{\Lambda}_k$. To proceed, we note that $M_t$ is invertable. Indeed, since Equation (\ref{eq:inv_lambda}) holds for the reference layer, we conclude that $M_t I_{p',q'} M_t^{\rm T} = I_{p,q}$. It is then easy to check that $M^{-1}_t = I_{p',q'} M^{\rm T}_t I_{p,q}$. Therefore, $(M_t^{\rm T})^{-1} = (M_t^{-1})^{\rm T} = I_{p,q} M_t I_{p', q'}$. Multiplying Equation (\ref{eq:inv_lambda}) on the left by $M_t^{-1}$ and on the right by $(M_t^{\rm T})^{-1}$, we find that
\[
\tilde{\Lambda}_k = \left[I_{p',q'} M_t^{\rm T} I_{p,q}\right] \Lambda_k \left[I_{p,q} M_t I_{p',q'}\right].
\]
Lastly, multiplying on the left and the right by $I_{p',q'}$ and noting that $I_{p',q'}\tilde{\Lambda}_k I_{p',q'} = \tilde{\Lambda}_k$ and $I_{p,q}\Lambda_k I_{p,q} = \Lambda_k$, we find that
\[
\tilde{\Lambda}_k = M_t^{\rm T} \Lambda_k  M_t,
\]
which completes the proof.
\end{proof}

\begin{proof}[Proof of Proposition~\ref{prop:indefinite_ortho}]
From Proposition~\ref{prop:identify}, we have that each matrix $M_t$ satisfies $M_t I_{p',q'} M_t^{\rm T} = I_{p,q}$ for $1 \leq t \leq T$. Consider a single matrix $M \in \set{M_t}_{t=1}^T$. Taking the determinant of both sides of $M I_{p',q'} M^{\rm T} = I_{p,q}$, we conclude that $\det(M)^2 (-1)^{q'} = (-1)^{q}$, so that $q - q'$ is an even number. Without loss of generality, assume that $q \geq q'$. We proceed case by case:
\begin{enumerate}
\item[(i)] $d = 1$. Since $q - q'$ can only equal zero, the result is immediate.
\item[(ii)] $d = 2$. In this case, the only non-trivial case is $q - q' = 2$, which corresponds to $q = 2$, $q' = 0$. Now, we show that this combination leads to a contradiction. In this case, $M$ satisfies $MM^{\rm T} = -I_2$, which is a contradiction because $MM^{\rm T}$ is a positive-definite matrix while $-I_2$ is not.  Thus, $q = q'$.
\item[(iii)] $d = 3$. Once again, the only non-trivial case is $q - q' = 2$, where we have the following two cases: $q = 3$, $q' = 1$ and $q = 2$, $q' = 0$. We proceed by showing that both scenarios lead to a contradiction.

For the case $q = 3$, $q' = 1$, we have $M I_{2, 1} M^{\rm T} = - I_3$ which implies $-M^{\rm T}M = I_{2, 1}$, where we used the fact that $M^{\rm T} I_{p,q} M = I_{p',q'}$ shown during the proof of Proposition \ref{prop:identify}. Letting $\boldm_j$ be the $j$th column of $M$, we have that
\[
-\norm{\boldm_1}^2_2 = 1,
\]
which cannot be satisfied by a real vector.

Similarly for $q = 2$, $q' = 0$, we have $MM^{\rm T} = I_{1,2}$. Letting $\tilde{\boldm}_j$ be the $j$th row vector of $M$, we have that
\[
\norm{\tilde{\boldm}_2}_2^2 = -1.
\]
which is impossible for a real vector.
\end{enumerate}
Therefore, $I_{p,q} = I_{p',q'}$ when $1 \leq d \leq 3$, which completes the proof.
\end{proof}

\begin{proof}[Proof of Proposition~\ref{prop:permutation}]
Without loss of generality, consider a single matrix $M \in \set{M_t}_{t=1}^T$. Further let $\Lambda_k = \diag(\blambda_k)$ and $\tilde{\Lambda}_k = \diag(\tilde{\blambda}_k)$, so that by Proposition~\ref{prop:identify} we have that
\begin{equation} \label{eq:id_group}
\diag(\tilde{\blambda}_k) = M^{\rm T} \diag(\blambda_k) M.
\end{equation}
Now, left multiplying $M  I_{p',q'}$ on both sides of Equation (\ref{eq:id_group}) and apply the identity $M I_{p',q'} M^{\rm T} = I_{p,q}$, we have that
\begin{equation}\label{eq:linear_system}
M \diag(\tilde{\blambda}_k) I_{p',q'} = \diag(\blambda_k) I_{p,q} M,
\end{equation}
Denoting the $j$th columns of $M$ by $\boldm_j \in \Reals{d}$, we can re-express the linear system in Equation (\ref{eq:linear_system}) as
\begin{equation}\label{eq:linear_system2}
\left(\tilde{\lambda}_{k,j} (I_{p',q'})_{jj} \, I_d - \diag(\blambda_k) I_{p,q} \right) \boldm_j = \mathbf{0}_d \quad \text{ for } j = 1, \dots, d,
\end{equation}
where $\mathbf{0}_d$ is a $d$-dimensional vector of zeros.

Now, we determine what relationship Equation (\ref{eq:linear_system2}) imposes on $\diag(\blambda_k)$ and $\diag(\tilde{\blambda}_k)$. Let $A_j = \tilde{\lambda}_{j,k} (I_{p',q'})_{jj} \, I_d - \diag(\blambda_k) I_{p,q}$ for $j = 1, \dots, d$. Since $M$ is full rank, $A_j$ must be a singular matrix for $j = 1, \dots, d$. Combining the facts that $\diag(\tilde{\blambda}_k) I_{p',q'}$ and $\diag(\blambda_k)I_{p,q}$ are both full rank with $d$ distinct elements and that there are $d$ singular diagonal matrices $A_j$, it is easy to see that $\diag(\blambda_k) I_{p,q} = P \diag(\tilde{\blambda}_k) I_{p',q'} P^{\rm T}$ for some permutation matrix $P$. In other words, $\diag(\blambda_k) I_{p,q}$ equals $\diag(\tilde{\blambda}_k) I_{p',q'}$ with permuted diagonal entries.

Now we focus on the consequences for $M$. As a result of the argument in the previous paragraph, each $A_j$ is a rank $d-1$ diagonal matrix. This means that each $A_j$ is a diagonal matrix with $d-1$ non-zero entries and a single zero entry on the diagonal. Therefore, Equation (\ref{eq:linear_system2}) holds if and only if $\set{\boldm_j}_{j=1}^d$ are $d$-dimensional vectors with a single non-zero entry where $A_j$ is zero on the diagonal. Also, since $M$ is full rank, $\set{\boldm_1, \dots, \boldm_d}$ are linearly independent. This implies that $M$ is a generalized permutation matrix: $M = P\diag(\bs)$ where $\diag(\bs)$ is a full-rank diagonal matrix and $P$ is a permutation matrix.

To complete the proof, we focus on the diagonal entries of $M I_{p',q'} M^{\rm T} = I_{p,q}$. From the previous paragraph, we have that $M I_{p',q'} M^{\rm T} = P \diag(\bs) I_{p',q'} \diag(\bs) P^{\rm T}$. Let $\sigma : \set{1, \dots, d} \rightarrow \set{1, \dots, d}$, denote the permutation encoded by $P$, so that
\[
P =
\begin{pmatrix}
\be_{\sigma(1)}^{\rm T} \\
\vdots \\
\be_{\sigma(d)}^{\rm T}
\end{pmatrix},
\]
where $\be_1, \dots, \be_d$ are the standard basis vectors. Thus, the diagonal entries must satisfy
\[
(I_{p',q'})_{\sigma(j)\sigma(j)} \, s_{\sigma(j)}^2 = (I_{p, q})_{jj} \text{ for }  j = 1, \dots, d,
\]
which holds if and only if $\bs = \set{\pm 1}^d$, $I_{p,q} = I_{p',q'}$, and $\sigma$ only permutes the first $p$ and last $q$ diagonal elements of $I_{p,q}$.
\end{proof}

\section{Derivation of Variational Updates}\label{sec:variational_updates}

This section contains detailed derivations of the variational updates presented in Section \ref{sec:estimation} of the main text. For notational simplicity, we use the shorthand $\Eq{q(\btheta, \bphi, \bomega)}{\cdot} = \mean{\cdot}$, where $q(\btheta, \bphi, \bomega)$ is defined in Equation~(\ref{eq:variational_q}) of the main text, to denote expectations with respect to the full variational posterior throughout this section. For a definition of the notation used in this section, see Algorithm~\ref{alg:cavi}.

Throughout this section, we encounter the following Gaussian state space model
\begin{align}
\bx_1 &\sim N(0, \tau^2 I_d), \label{eq:gssm-init} \\
\bx_t &= \bx_{t-1} + \bw_t, \quad \bw_t \sim N(0, \sigma^2 I_d), \label{eq:gssm-latent} \\
\by_t &= A_t \bx_t + \bb_t + \bv_t, \quad \bv_t \sim N(0, C_t), \label{eq:gssm-observed}
\end{align}
where $\bx_t \in \Reals{d}$, $\by_t \in \Reals{n}$, $A_t \in \Reals{n \times d}$, $\bb_t \in \Reals{n}$, $C_t \in \Reals{n \times n}$. In this context, $d$ is not necessarily the dimension of the latent space and $n$ is not necessarily the number of nodes in the network. Specifically, the full conditional distributions of the social and latent trajectories are of this form.  Before proceeding, we state a lemma used throughout Appendix \ref{sec:variational_updates} and Appendix \ref{sec:variational_smoother}.
\begin{lemma}\label{lemma:gssm}
For the Gaussian state space model specified by Equations (\ref{eq:gssm-init}) -- (\ref{eq:gssm-observed}), the conditional distribution $p(\bx_{1:T} \mid \by_{1:T})$ is in the exponential family with natural parameters
\begin{equation}\label{eq:lds_natural_params}
\psi = (-1/2\tau^2, -1/2\sigma^2, \Gamma_{1:T}^1, \Gamma_{1:T}^2 / 2),
\end{equation}
where $\Gamma_{t}^1 = A_{t}^{\rm T} C_{t}^{-1}\by_{t} - A_{t} C_{t}^{-1}\bb_{t}$ and $\Gamma_{t}^2 = A_{t}^{\rm T} C_{t}^{-1} A_{t}$ for $1 \leq t \leq T$.
\end{lemma}
\begin{proof}
We have
\begin{align*}
\log p(\bx_{1:T} \mid \by_{1:T}) &\propto -\frac{1}{2\tau^2} \norm{\bx_1}^2_2 - \frac{1}{2\sigma^2} \sum_{t=1}^T \norm{\bx_t - \bx_{t-1}}_2^2 - \nonumber \\
&\qquad \frac{1}{2}\sum_{t=1}^T (\by_t - A_t\bx_t - \bb_t)^{\rm T} C_t^{-1} (\by_t - A_t\bx_t - \bb_t), \\
&\propto -\frac{1}{2\tau^2} \norm{\bx_1}^2_2 - \frac{1}{2\sigma^2} \sum_{t=1}^T \norm{\bx_t - \bx_{t-1}}_2^2 + \nonumber \\
&\qquad (A_t^{\rm T} C_t^{-1} \by_t - A_t C_t^{-1}\bb_t)^{\rm T}\bx_t - \frac{1}{2}\sum_{t=1}^T \tr(A_t^{\rm T} C_t^{-1} A_t \bx_t \bx_t^{\rm T}),
\end{align*}
which is in exponential family form with natural parameters given in Equation (\ref{eq:lds_natural_params}).
\end{proof}

The variational distributions of the social and latent trajectories---Equation~(\ref{eq:delta_gssm}) and Equation~(\ref{eq:x_gssm}) in the main text---are GSSMs that are in the form assumed by Lemma \ref{lemma:gssm}. This observation implies that the expected natural parameters, $\Eq{-q(\bx_{1:T})}{\psi}$, are sufficient for calculating the variational distribution's moments, i.e., $\Eq{q(\bx_{1:T})}{\bx_t}$, $\Eq{q(\bx_{1:T})}{\bx_t \bx_t^{\rm T}}$, and $\Eq{q(\bx_{1:T})}{\bx_t \bx_{t+1}^{\rm T}}$. In Appendix~\ref{sec:variational_smoother}, we derive a variational Kalman smoother that calculates these moments recursively.

\subsection{Derivation of Algorithm \ref{alg:cavi_omega}}

The coordinate updates for $q(\omega_{ijt}^k)$ are given in Algorithm~\ref{alg:cavi_omega}, which we formally derive in the remainder of this section.

\begin{proposition}
Under the eigenmodel for dynamic multilayer networks with the same prior distributions and variational factorization defined in the main text, $q(\omega_{ijt}^k) = \PG(1, c_{ijt}^k)$ where $c_{ijt}^k$ is given in Equation (\ref{eq:c_omega}). Furthermore, the mean of this distribution is given by Equation (\ref{eq:omega_mean}).
\end{proposition}

\begin{myalgorithm}[tbp]
\begin{framed}
\quad Update $q(\omega_{ijt}^k) = \PG(1, c_{ijt}^k)$:

\vspace{0.5em}

\quad For each $k \in \set{1, \dots, K}$, $t \in \set{1, \dots T}$, and $(i, j) \in \set{(i,j) :1 \leq i \leq n, j < i}$:
\begin{align}
c_{ijt}^k &= (\sigma_{\delta_{k,t}^i}^2 + \mu_{\delta_{k,t}^i}^2 + \sigma_{\delta_{k,t}^j}^2 + \mu_{\delta_{k,t}^j}^2 + 2 \mu_{\delta_{k,t}^i} \mu_{\delta_{k,t}^j} + 2 (\mu_{\delta_{k,t}^i} + \mu_{\delta_{k,t}^j}) \bmu_t^{i \, \rm T} \diag(\bmu_{\blambda_k})\bmu_t^j \ + \nonumber \\
&\qquad\quad \norm{(\Sigma_{\blambda_k} + \bmu_{\blambda_k} \bmu_{\blambda_k}^{\rm T}) \odot (\Sigma_t^i + \bmu_t^{i}\bmu_t^{i \, \rm T}) \odot (\Sigma_t^j + \bmu_t^{j}\bmu_t^{j \, \rm T})})^{1/2}, \label{eq:c_omega} \\
\mu_{\omega_{ijt}^k} &= \frac{1}{2 c_{ijt}^k}\left(\frac{e^{c_{ijt}^k} - 1}{1 + e^{c_{ijt}^k}}\right). \label{eq:omega_mean}
\end{align}
\end{framed}
\caption{Coordinate ascent updates for the auxiliary P\'olya-gamma variables. Here, $\odot$ is the Hadamard product between two matrices, i.e., $(A \odot B)_{ij} = A_{ij} B_{ij}$, and $\norm{A} = \sum_i \sum_j A_{ij}$.}
\label{alg:cavi_omega}
\end{myalgorithm}

\begin{proof}
From the exponential tilting property of the P\'olya-gamma distribution, we have that
\begin{equation}
\omega_{ijt}^k \mid \cdot \sim \PG(1, \psi_{ijt}^k),
\end{equation}
where $\psi_{ijt}^k = \delta_{k,t}^i + \delta_{k,t}^j + \bX_t^{i \, \rm T} \Lambda_k \bX_t^j$. This distribution is in the exponential family with natural parameter $-(\psi_{ijt}^k)^2 / 2$. The variational distribution is then a $\PG(1, c_{ijt}^k)$ where $(c_{ijt}^k)^2 = \Eq{-q(\omega_{ijt}^k)}{(\psi_{ijt}^k)^2}$.

It remains to calculate the natural parameter $c_{ijt}^k$. We have that
\begin{align*}
(c_{ijt}^k)^2 = \Eq{-q(\omega_{ijt}^k)}{(\psi_{ijt}^k)^2} &= \mean{(\delta_{k,t}^i + \delta_{k,t}^j + \bX_t^{i \, \rm T} \Lambda_k \bX_t^j)^2},  \\
&= \mean{(\delta_{k,t}^i + \delta_{k,t}^j)^2} + 2 \mean{\delta_{k,t}^i + \delta_{k,t}^j} \mean{\bX_t^i}^{\rm T} \mean{\Lambda_k} \mean{\bX_t^j} + \mean{(\bX_t^{i \, \rm T} \Lambda_k \bX_t^j)^2}, \\
&= \sigma_{\delta_{k,t}^i}^2 + \mu_{\delta_{k,t}^i}^2 + \sigma_{\delta_{k,t}^j}^2 + \mu_{\delta_{k,t}^j}^2 + 2 \mu_{\delta_{k,t}^i} \mu_{\delta_{k,t}^j} + \nonumber \\
&\quad\qquad 2 (\mu_{\delta_{k,t}^i} + \mu_{\delta_{k,t}^j}) \bmu_t^{i \, \rm T} \diag(\bmu_{\blambda_k})\bmu_t^j + \Eq{-q(\omega_{ijt}^k)}{(\bX_t^{i \, \rm T} \Lambda_k \bX_t^j)^2}.
\end{align*}
Note that the last term is equal to
\begin{align*}
\Eq{-q(\omega_{ijt}^k)}{(\bX_t^{i \, \rm T} \Lambda_k \bX_t^j)^2} &= \Eq{-q(\omega_{ijt}^k)}{\sum_{g=1}^d \sum_{h=1}^d \lambda^k_g \lambda^k_h X_{tg}^i X_{th}^i X_{tg}^j X_{th}^j}, \\
&= \sum_{g=1}^d \sum_{h=1}^d \Eq{q(\blambda_k)}{\lambda_g^k \lambda_h^k} \Eq{q(\bX_t^i)}{X_{tg}^i X_{th}^i} \Eq{q(\bX_t^j)}{X_{tg}^j X_{th}^j}, \\
&= \norm{\Eq{q(\blambda_k)}{\blambda_k \blambda_k^{\rm T}} \odot \Eq{q(\bX_t^i)}{\bX_t^i \bX_t^{i \, \rm T}} \odot \Eq{q(\bX_t^j)}{\bX_t^j \bX_t^{j \, \rm T}}}, \\
&= \norm{(\Sigma_{\blambda_k} + \bmu_{\blambda_k} \bmu_{\blambda_k}^{\rm T}) \odot (\Sigma_t^i + \bmu_t^{i}\bmu_t^{i \, \rm T}) \odot (\Sigma_t^j + \bmu_t^{j}\bmu_t^{j \, \rm T})},
\end{align*}
where $\odot$ is the Hadamard product, i.e, $(A \odot B)_{ij} = A_{ij} B_{ij}$, and $\norm{A} = \sum_{i} \sum_{j} A_{ij}$.

Lastly, the moments of the P\'olya-gamma distribution are available in closed form. In particular, we have that
\begin{equation*}
\mu_{\omega_{ijt}^k} = \Eq{q(\omega_{ijt}^k)}{\omega_{ijt}^k} = \frac{1}{2 c_{ijt}^k}\left(\frac{e^{c_{ijt}^k} - 1}{1 + e^{c_{ijt}^k}}\right).
\end{equation*}
\end{proof}

\subsection{Derivation of Algorithm \ref{alg:cavi_delta}}

The coordinate updates for $q(\delta_{1:T}^i), q(\tau_{\delta}^2)$, and $q(\sigma_{\delta}^2)$ are given in Algorithm~\ref{alg:cavi_delta}, which we formally derive in the remainder of this section.

\begin{myalgorithm}[tbp]
\begin{framed}
\begin{enumerate}

\item Update $q(\delta_{k, 1:T}^i)$, a linear Gaussian state space model (GSSM):

For each $k \in \set{1, \dots K}$ and $i \in \set{1, \dots, n}$:
\begin{enumerate}
\item[(a)] For $t \in \set{1, \dots T}$, update the natural parameters of the GSSM:
\begin{align}
\Gamma_t^1 &= \sum_{j \neq i} [Y_{ijt}^k - 1/2 - \mu_{\omega_{ijt}^k} (\mu_{\delta_{k,t}^j} + \bmu_t^{i\, \rm T} \diag(\bmu_{\blambda_k}) \bmu_t^j)], \label{eq:delta_gamma1} \\
\Gamma_t^2 &= \sum_{j \neq i} \mu_{\omega_{ijt}^k}, \\
\mean*{1/\tau^2_{\delta}} &= \bar{a}_{\tau_{\delta}^2} / \bar{b}_{\tau_{\delta}^2},\\
\mean*{1 / \sigma^2_{\delta}} &= \bar{c}_{\sigma_{\delta}^2} / \bar{d}_{\sigma_{\delta}^2}.\label{eq:delta_sigma}
\end{align}

\item[(b)] Update marginal distributions and cross-covariances as in Algorithm \ref{alg:kalman_smoother}:
\begin{align*}
\mu_{\delta_{k,1:T}^i}, \sigma_{\delta_{k,1:T}^i}^2, \set{\sigma_{\delta_{k, t,t+1}^i}^2}_{t=1}^{T-1} = \text{\tt KalmanSmoother}(\Gamma_{1:T}^1, \Gamma_{1:T}^2, \bar{a}_{\tau_{\delta}^2}/\bar{b}_{\tau_{\delta}^2}, \bar{c}_{\sigma_{\delta}^2}/\bar{d}_{\sigma_{\delta}^2}).
\end{align*}
\end{enumerate}

\item Update $q(\tau_{\delta}^2) = \InvGamma(\bar{a}_{\tau_{\delta}^2}/2, \bar{b}_{\tau_{\delta}^2}/2)$:
\begin{align}
\bar{a}_{\tau_{\delta}^2} &= a_{\tau^2_{\delta}} + nK, \label{eq:a_delta} \\
\bar{b}_{\tau_{\delta}^2} &= b_{\tau^2_{\delta}} + \sum_{k=1}^K \sum_{i=1}^n \left(\sigma_{\delta_{k,1}^i}^2 + \mu_{\delta_{k,1}^i}^2\right) \label{eq:b_delta}.
\end{align}

\item Update $q(\sigma_{\delta}^2) = \InvGamma(\bar{c}_{\sigma_{\delta}^2}/2, \bar{d}_{\sigma_{\delta}^2}/2)$:
\begin{align}
\bar{c}_{\sigma_{\delta}^2} &= c_{\sigma^2_{\delta}} + nK(T-1), \label{eq:c_delta} \\
\bar{d}_{\sigma_{\delta}^2} &= d_{\sigma^2_{\delta}} + \sum_{k=1}^K\sum_{t=2}^T \sum_{i=1}^n \Big\{\sigma_{\delta_{k,t}^i}^2 + \mu_{\delta_{k,t}^i}^2 + \sigma_{\delta_{k,t-1}^i}^2 + \mu_{\delta_{k,t-1}^i}^2 - 2(\sigma_{\delta_{k,t-1, t}^i}^2 + \mu_{\delta_{k,t-1}^i}\mu_{\delta_{k,t}^i})\Big\} \label{eq:d_delta}.
\end{align}

\end{enumerate}
\end{framed}
\caption{Coordinate ascent updates for the social trajectories. {\tt KalmanSmoother} is the variational Kalman smoother defined in Algorithm \ref{alg:kalman_smoother} of Appendix \ref{sec:variational_smoother}.}
\label{alg:cavi_delta}
\end{myalgorithm}

\begin{proposition} \label{prop:tau_delta}
Under the eigenmodel for dynamic multilayer networks with the same prior distributions and variational factorization defined in the main text, $q(\tau^2_{\delta}) = \InvGamma(\bar{a}_{\tau_{\delta}^2}/2, \bar{b}_{\tau_{\delta}^2}/2)$ where $\bar{a}_{\tau_{\delta}^2}$ and $\bar{b}_{\tau_{\delta}^2}$ are defined in Equation (\ref{eq:a_delta}) and Equation (\ref{eq:b_delta}), respectively.
\end{proposition}
\begin{proof}
Standard calculations show that
\begin{align*}
p(\tau^2_{\delta} \mid \cdot) &\propto \left(\frac{1}{\tau^2_{\delta}}\right)^{(a_{\tau^2_{\delta}} + n K) / 2} \exp\left(-\frac{1}{2\tau^2_{\delta}} \sum_{k=1}^K \sum_{i=1}^n (\delta_{k,1}^i)^2  - \frac{b_{\tau^2_{\delta}}}{2 \tau^2_{\delta}} \right), \\
&\propto \InvGamma\left(\frac{a_{\tau^2_{\delta}} + nK}{2}, \frac{1}{2}\left\{ \sum_{k=1}^K \sum_{i=1}^n (\delta_{k,1}^i)^2 + b_{\tau^2_{\delta}} \right\}\right).
\end{align*}
Note that $\InvGamma(a/2, b/2)$ is in the exponential family with natural parameters $a$ and $b$. Thus, the variational distribution is also an inverse-gamma distribution with natural parameters
\begin{align*}
\bar{a}_{\tau_{\delta}^2} &= a_{\tau^2_{\delta}} + n K, \\
\bar{b}_{\tau_{\delta}^2} &= b_{\tau^2_{\delta}} + \sum_{k=1}^K \sum_{i=1}^n \Eq{q(\delta_{k, 1:T}^i)}{(\delta_{k,1}^i)^2}, \\
&= b_{\tau^2_{\delta}} + \sum_{k=1}^K \sum_{i=1}^n \left( \sigma_{\delta_{k,1}^i}^2 + \mu_{\delta_{k,1}^i}^2 \right).
\end{align*}
\end{proof}

\begin{proposition} \label{prop:sigma_delta}
Under the eigenmodel for dynamic multilayer networks with the same prior distributions and variational factorization defined in the main text, $q(\sigma^2_{\delta}) = \InvGamma(\bar{c}_{\tau_{\delta}^2}/2, \bar{d}_{\tau_{\delta}^2}/2)$ where $\bar{c}_{\tau_{\delta}^2}$ and $\bar{d}_{\tau_{\delta}^2}$ are defined in Equation (\ref{eq:c_delta}) and Equation (\ref{eq:d_delta}), respectively.
\end{proposition}
\begin{proof}
Standard calculations show that
\begin{align*}
p(\sigma_{\delta}^2 \mid \cdot) &\propto \left(\frac{1}{\sigma^2_{\delta}}\right)^{(c_{\sigma^2_{\delta}} + nK(T-1))/2} \exp\left(-\frac{1}{2 \sigma_{\delta}^2} \sum_{k=1}^K \sum_{t=2}^T \sum_{i=1}^n (\delta_{k,t}^i - \delta_{k,t-1}^i)^2 - \frac{d_{\sigma^2_{\delta}}}{2 \sigma^2_{\delta}} \right), \\
&\propto \InvGamma\left(\frac{c_{\sigma^2_{\delta}} + n K (T - 1)}{2}, \frac{1}{2} \left\{\sum_{k=1}^K \sum_{t=2}^T \sum_{i=1}^n (\delta_{k,t}^i - \delta_{k,t-1}^i)^2 + d_{\sigma^2_{\delta}}\right\}\right).
\end{align*}
Note that $\InvGamma(a/2, b/2)$ is in the exponential family with natural parameters $a$ and $b$. Thus, the variational distribution is also an inverse-gamma distribution with natural parameters
\begin{align*}
\bar{c}_{\sigma_{\delta}^2} &= c_{\sigma^2_{\delta}} + n K (T - 1), \\
\bar{d}_{\sigma_{\delta}^2} &= d_{\sigma^2_{\delta}} + \sum_{k=1}^K \sum_{t=2}^T \sum_{i=1}^n \Eq{q(\delta_{k, 1:T}^i)}{(\delta_{k,t}^i - \delta_{k,t-1}^i)^2}, \\
&= d_{\sigma^2_{\delta}} + \sum_{k=1}^K \sum_{t=2}^T \sum_{i=1}^n \left\{\sigma_{\delta_{k,t}^i}^2 + \mu_{\delta_{k,t}^i}^2 + \sigma^2_{\delta_{k,t-1}^i} + \mu_{\delta_{k,t-1}^i}^2 - 2 (\sigma_{\delta_{k,t-1,t}}^2 + \mu_{\delta_{k,t}^i} \mu_{\delta_{k,t-1}^i})\right\}.
\end{align*}
\end{proof}

\begin{proposition}
Under the eigenmodel for dynamic multilayer networks with the same prior distributions and variational factorization defined in the main text, $q(\delta_{k,1:T}^i)$ is a Gaussian state space model with natural parameters given by Equations (\ref{eq:delta_gamma1}) -- (\ref{eq:delta_sigma}).
\end{proposition}
\begin{proof}
From Equation (\ref{eq:delta_gssm}), we associate $p(\delta_{k,1:T}^i \mid \cdot)$ with a GSSM parameterized by
\begin{align*}
\bx_t &= \delta_{k,t}^i \in \Reals{}, \\
\by_t &= \bz_{k,t}^i \in \Reals{n-1}, \\
A_t &= (\omega_{i1t}^k, \dots, \omega_{i(i-1)t}^k, \omega_{i(i+1)t}^k, \dots, \omega_{int}^k) \in \Reals{n-1}, \\
(\bb_t)_j &= \omega_{ijt}^k (\delta_{k,t}^j + \bX_t^{j\, \rm T} \Lambda_k \bX_t^i),\quad \text{ for } j \neq i, \\
C_t &= \diag(\omega_{i1t}^k, \dots, \omega_{i(i-1)t}^k, \omega_{i(i+1)t}^k, \dots, \omega_{int}^k).
\end{align*}
We then apply Lemma \ref{lemma:gssm} in Appendix \ref{sec:variational_smoother} to identify the natural parameters. Note that $A_t^{\rm T} C_t^{-1} = \mathbf{1}_{n-1}$, so that expected natural parameters are
\begin{align*}
\Gamma_t^1 &= \mean{\mathbf{1}_{n-1}^{\rm T}(\by_t - \bb_t)} = \sum_{j\neq i} (Y_{ijt}^k - 1/2 - \mu_{\omega_{ijt}^k} (\mu_{\delta_{k,t}^j} + \bmu_t^{j\, \rm T} \diag(\bmu_{\blambda_k}) \bmu_t^i) ), \\
\Gamma_t^2 &= \mean{\mathbf{1}_{n-1}^{\rm T}A_t} = \sum_{j \neq i} \mu_{\omega_{ijt}^k}.
\end{align*}
Furthermore, from Proposition \ref{prop:tau_delta} and Proposition \ref{prop:sigma_delta}, we have that  $1/\tau^2_{\delta}$ and $1/\sigma^2_{\delta}$ are gamma distributed with means $\bar{a}_{\tau_{\delta}^2}/\bar{b}_{\tau_{\delta}^2}$ and $\bar{c}_{\sigma_{\delta}^2}/\bar{d}_{\sigma_{\delta}^2}$, respectively. Finally, we can apply the variational Kalman smoothing equations derived in Appendix \ref{sec:variational_smoother} to calculate the moments of $q(\delta_{k,1:T}^i)$.
\end{proof}

\subsection{Derivation of Algorithm \ref{alg:cavi_latent_space}}

The coordinate updates for $q(\bX_{1:T}^i), q(\tau^2)$, and $q(\sigma^2)$ are given in Algorithm~\ref{alg:cavi_latent_space}, which we formally derive in the remainder of this section.

\begin{myalgorithm}[tbp]
\begin{framed}
\begin{enumerate}

\item Update $q(\bX_{1:T}^i)$, a linear Gaussian state space model (GSSM):

For $i \in \set{1, \dots, n}$:
\begin{enumerate}
\item[(a)] For $t \in \set{1, \dots T}$, update the natural parameters of the GSSM:
\begin{align}
\Gamma_t^1 &=\begin{pmatrix}
\sum_{k=1}^K \sum_{j \neq i} \mu_{\lambda_{k1}} \mu_{t1}^j [Y_{ijt}^k - 1/2 - \mu_{\omega_{ijt}^k} (\mu_{\delta_{k,t}^i} + \mu_{\delta_{k,t}^j})] \\
\vdots \\
\sum_{k=1}^K \sum_{j \neq i} \mu_{\lambda_{kd}} \mu_{td}^j [Y_{ijt}^k - 1/2 - \mu_{\omega_{ijt}^k} (\mu_{\delta_{k,t}^i} + \mu_{\delta_{k,t}^j})]
\end{pmatrix}, \label{eq:x_gamma1}\\
\Gamma_t^2 &= \sum_{k=1}^K \sum_{j \neq i} \mu_{\omega_{ijt}^k} \left(\Sigma_{\blambda_k} + \bmu_{\blambda_k}\bmu_{\blambda_k}^{\rm T} \right) \odot \left(\Sigma_t^j  + \bmu_t^j \bmu_t^{j \, \rm T}\right), \\
\mean*{1/\tau^2} &= \bar{a}_{\tau^2} / \bar{b}_{\tau^2}, \\
\mean*{1/\sigma^2} &= \bar{c}_{\sigma^2} / \bar{d}_{\sigma^2}. \label{eq:x_sigma}
\end{align}

\item[(b)] Update marginal distributions and cross-covariances as in Algorithm \ref{alg:kalman_smoother}:
\begin{align*}
\bmu_{1:T}^i, \Sigma_{1:T}^i, \set{\Sigma_{t, t+1}^i}_{t=1}^{T-1} = \text{\tt KalmanSmoother}(\Gamma_{1:T}^1, \Gamma_{1:T}^2, \bar{a}_{\tau^2}/\bar{b}_{\tau^2}, \bar{c}_{\sigma^2}/\bar{d}_{\sigma^2}).
\end{align*}
\end{enumerate}

\item Update $q(\tau^2) = \InvGamma(\bar{a}_{\tau^2}/2, \bar{b}_{\tau^2}/2)$:
\begin{align}
\bar{a}_{\tau^2} &= a_{\tau^2} + nd, \label{eq:a_tau} \\
\bar{b}_{\tau^2} &= b_{\tau^2} + \sum_{i=1}^n \left(\tr(\Sigma_1^i) + \bmu_1^{i \, \rm T} \bmu_1^i\right). \label{eq:b_tau}
\end{align}

\item Update $q(\sigma^2) = \InvGamma(\bar{c}_{\sigma^2}/2, \bar{d}_{\sigma^2}/2)$:
\begin{align}
\bar{c}_{\sigma^2} &= c_{\sigma^2} + nd(T-1), \label{eq:c_sigma} \\
\bar{d}_{\sigma^2} &= d_{\sigma^2} + \sum_{t=2}^T \sum_{i=1}^n \Big\{\tr(\Sigma_t^i) + \bmu_t^{i \, \rm T} \bmu_t^i + \tr(\Sigma_{t-1}^i) + \bmu_{t-1}^{i \, \rm T} \bmu_{t-1}^i \nonumber \\
&\qquad\qquad\qquad\quad - 2(\tr(\Sigma_{t-1, t}^i) + \bmu_{t-1}^{i \, \rm T}\bmu_t^i)\Big\}. \label{eq:d_sigma}
\end{align}
\end{enumerate}
\end{framed}
\caption{Coordinate ascent updates for the latent trajectories. {\tt KalmanSmoother} is the variational Kalman smoother defined in Algorithm \ref{alg:kalman_smoother} of Appendix \ref{sec:variational_smoother}.}
\label{alg:cavi_latent_space}
\end{myalgorithm}

\begin{proposition}\label{prop:tau}
Under the eigenmodel for dynamic multilayer networks with the same prior distributions and variational factorization defined in the main text, $q(\tau^2) = \InvGamma(\bar{a}_{\tau^2}/2, \bar{b}_{\tau^2}/2)$ where $\bar{a}_{\tau^2}$ and $\bar{b}_{\tau^2}$ are defined in Equation (\ref{eq:a_tau}) and Equation (\ref{eq:b_tau}), respectively.
\end{proposition}
\begin{proof}
Standard calculations show that
\begin{align*}
p(\tau^2 \mid \cdot) &\propto \left(\frac{1}{\tau^2}\right)^{(a_{\tau^2} + n d) / 2} \exp\left(-\frac{1}{2\tau^2} \sum_{i=1}^n \norm{\bX_1^i}^2_2  - \frac{b_{\tau^2}}{2 \tau^2} \right), \\
&\propto \InvGamma\left(\frac{a_{\tau^2} + nd}{2}, \frac{1}{2}\left\{ \sum_{i=1}^n \norm{\bX_1^i}^2_2 + b_{\tau^2} \right\}\right).
\end{align*}
Note that $\InvGamma(a/2, b/2)$ is in the exponential family with natural parameters $a$ and $b$. Thus, the variational distribution is also an inverse-gamma distribution with natural parameters
\begin{align*}
\bar{a}_{\tau^2} &= a_{\tau^2} + n d, \\
\bar{b}_{\tau^2} &= b_{\tau^2} + \sum_{i=1}^n \Eq{q(\bX_{1:T}^i)}{\norm{\bX_1^i}^2_2}, \\
&= b_{\tau^2} + \sum_{k=1}^K \sum_{i=1}^n \left( \tr(\Sigma_1^i) + \bmu_1^{i \, \rm T} \bmu_1^i \right).
\end{align*}
\end{proof}

\begin{proposition}\label{prop:sigma}
Under the eigenmodel for dynamic multilayer networks with the same prior distributions and variational factorization defined in the main text, $q(\sigma^2) = \InvGamma(\bar{c}_{\sigma^2}/2, \bar{d}_{\sigma^2}/2)$ where $\bar{c}_{\sigma^2}$ and $\bar{d}_{\sigma^2}$ are defined in Equation (\ref{eq:c_sigma}) and Equation (\ref{eq:d_sigma}), respectively.
\end{proposition}
\begin{proof}
Standard calculations show that
\begin{align*}
p(\sigma^2 \mid \cdot) &\propto \left(\frac{1}{\sigma^2}\right)^{(c_{\sigma^2} + n d (T-1))/2} \exp\left(-\frac{1}{2 \sigma^2} \sum_{t=2}^T \sum_{i=1}^n \norm{\bX_{t}^i - \bX_{t-1}^i}^2_2 - \frac{d_{\sigma^2}}{2 \sigma^2} \right), \\
&\propto \InvGamma\left(\frac{c_{\sigma^2} + n d (T - 1)}{2}, \frac{1}{2} \left\{\sum_{t=2}^T \sum_{i=1}^n \norm{\bX_{t}^i - \bX_{t-1}^i}^2_2 + d_{\sigma^2}\right\}\right).
\end{align*}
Note that $\InvGamma(a/2, b/2)$ is in the exponential family with natural parameters $a$ and $b$. Thus, the variational distribution is also an inverse-gamma distribution with natural parameters
\begin{align*}
\bar{c}_{\sigma^2} &= c_{\sigma^2} + n d (T - 1), \\
\bar{d}_{\sigma^2} &= d_{\sigma^2} + \sum_{t=2}^T \sum_{i=1}^n \Eq{q(\bX_{1:T}^i)}{\norm{\bX_{t}^i - \bX_{t-1}^i}^2_2}, \\
&= d_{\sigma^2} + \sum_{t=2}^T \sum_{i=1}^n \Big\{\tr(\Sigma_t^i) + \bmu_t^{i \, \rm T} \bmu_t^i + \tr(\Sigma_{t-1}^i) + \bmu_{t-1}^{i \, \rm T} \bmu_{t-1}^i \\
&\qquad\qquad\qquad\quad - 2(\tr(\Sigma_{t-1, t}^i) + \bmu_{t-1}^{i \, \rm T}\bmu_t^i)\Big\}.
\end{align*}
\end{proof}

\begin{proposition}
Under the eigenmodel for dynamic multilayer networks with the same prior distributions and variational factorization defined in the main text, $q(\bX_{1:T}^i)$ is a Gaussian state space model with natural parameters given in Equations (\ref{eq:x_gamma1}) -- (\ref{eq:x_sigma}).
\end{proposition}
\begin{proof}
First we layout some notation. Let
\begin{align*}
\btheta_{k,t}^{i} &= \delta_{k,t}^i \mathbf{1}_{n-1} + (\delta_{k,t}^1, \dots, \delta_{k,t}^{i-1}, \delta_{k,t}^{i+1}, \dots, \delta_{k,t}^n)^{\rm T} \in \Reals{n-1}, \\
\bomega_{it}^k &= (\omega_{i1t}^k, \dots, \omega_{i(i-1)t}^k, \omega_{i(i+1)t}^k, \dots \omega_{int}^k)^{\rm T} \in \Reals{n - 1}, \\
X_{k,t}^i &= (\bX_t^1 \Lambda_k, \dots, \bX_t^{i-1} \Lambda_k, \bX_t^{i + 1} \Lambda_k, \dots, \bX_t^n \Lambda_k)^{\rm T} \in \Reals{(n - 1) \times d}.
\end{align*}
We then define the concatenated version of these quantities:
\begin{align*}
\Omega_t^i &= \diag(\bomega_{it}^{1 \, \rm T}, \dots, \bomega_{it}^{K \, \rm T}) \in \Reals{K(n-1) \times K(n-1)}, \\
\btheta_t^i &= (\bdelta_{1,t}^{i \, \rm T}, \dots, \bdelta_{K, t}^{i \, \rm T})^{\rm T} \in \Reals{K (n-1)},
\end{align*}
and $X_t^i \in \Reals{K (n-1) \times d}$ formed by stacking the matrices $X_{k,t}^i$ row-wise for $k = 1, \dots, K$.

From Equation (\ref{eq:x_gssm}), we associate $p(\bX_{1:T}^i \mid \cdot)$ with a GSSM parameterized by
\begin{align*}
\bx_t &= \bX_t^i, \\
\by_t &= \bz_{t}^i, \\
A_t &= \Omega_t^i \, X_t^i, \\
\bb_t &= \Omega_t^i \, \btheta_t^i, \\
C_t &= \Omega_t^i.
\end{align*}
We then apply Lemma \ref{lemma:gssm} in Appendix \ref{sec:variational_smoother} to identify the natural parameters. Note that $A_t^{\rm T} C_t^{-1} = X_t^i$. Taking into account the independence assumptions contained in the approximate posterior, we have
\begin{align*}
\Gamma_t^1 &= \mean{X_t^i}^{\rm T} \left(\bz_t^i - \mean*{\Omega_t^i} \mean*{\btheta_t^i}\right), \\
&=\begin{pmatrix}
\sum_{k=1}^K \sum_{j \neq i} \mu_{\lambda_{k1}} \mu_{t1}^j [Y_{ijt}^k - 1/2 - \mu_{\omega_{ijt}^k} (\mu_{\delta_{k,t}^i} + \mu_{\delta_{k,t}^j})] \\
\vdots \\
\sum_{k=1}^K \sum_{j \neq i} \mu_{\lambda_{kd}} \mu_{td}^j [Y_{ijt}^k - 1/2 - \mu_{\omega_{ijt}^k} (\mu_{\delta_{k,t}^i} + \mu_{\delta_{k,t}^j})]
\end{pmatrix}.
\end{align*}
Next, the individual elements of $\Gamma_t^2 \in \Reals{d \times d}$ are
\begin{equation*}
(\Gamma_t^2)_{gh} = \mean*{X_t^{i\, \rm T}\Omega_t^i X_t^i}_{gh} = \sum_{k=1}^K \sum_{j \neq i} \mean*{\omega_{ijt}^k} \mean*{\lambda_g^k \lambda_h^k} \mean*{X_{tg}^j X_{th}^j},
\end{equation*}
or
\begin{align*}
\Gamma_t^2 &= \sum_{k=1} \sum_{j \neq i} \bmu_{\omega_{ijt}^k} \Eq{q(\blambda_k)}{\blambda_k \blambda_k^{\rm T}} \odot \Eq{q(\bX_{1:T}^j)}{\bX_t^j \bX_t^{j \, \rm T}}, \\
&= \sum_{k=1}^K \sum_{j \neq i} \mu_{\omega_{ijt}^k} (\Sigma_{\blambda_k} + \bmu_{\blambda_k} \bmu_{\blambda_k}^{\rm T}) \odot (\Sigma_t^j + \bmu_t^j \bmu_t^{j\, \rm T}).
\end{align*}
Furthermore, from Proposition \ref{prop:tau} and Proposition \ref{prop:sigma}, we have that  $1/\tau^2$ and $1/\sigma^2$ are gamma distributed with means $\bar{a}_{\tau^2}/\bar{b}_{\tau^2}$ and $\bar{c}_{\sigma^2}/\bar{d}_{\sigma^2}$, respectively. Finally, we can apply the variational Kalman smoothing equations derived in Appendix \ref{sec:variational_smoother} to calculate the moments of $q(\bX_{1:T}^i)$.
\end{proof}

\subsection{Derivation of Algorithm \ref{alg:cavi_lambda}}

The coordinate updates for $q(\blambda_k)$ are given in Algorithm~\ref{alg:cavi_lambda}, which we formally derive in the remainder of this section.

\begin{myalgorithm}[tbp]
\begin{framed}
\begin{enumerate}
\item Update $q(\lambda_{1h}) = p_{\lambda_{1h}}^{\ind{\lambda_{1h} = 1}} (1 - p_{\lambda_{1h}})^{\ind{\lambda_{1h} = -1}}$:

For $h \in \set{1, \dots, d}$:
\begin{align}
\eta_{\lambda_{1h}} &= \log\left[\frac{\rho}{1 - \rho}\right] + 2 \sum_{t=1} \sum_{j < i} \Bigg\{(Y_{ijt}^1 - 1/2 - \mu_{\omega_{ijt}^1} (\mu_{\delta_{1,t}^i} + \mu_{\delta_{1,t}^j}) \mu_{th}^i \mu_{th}^j - \nonumber \\
&\qquad\qquad \mu_{\omega_{ijt}^1} \sum_{g \neq h} \mu_{\lambda_{1g}} ((\Sigma_{t}^i)_{gh} + \mu_{tg}^i \mu_{th}^i) ((\Sigma_t^i)_{gh} + \mu_{tg}^j \mu_{th}^j)\Bigg\}, \\
p_{\lambda_{1h}} &= e^{\eta_{\lambda_{1h}}}/(1 + e^{\eta_{\lambda_{1h}}}), \quad \mu_{\lambda_{1h}} = 2 p_{\lambda_{1h}} - 1, \quad \sigma_{\lambda_{1h}}^2 = 1 - (2 p_{\lambda_{1h}} - 1)^2. \label{eq:lambda_ref_p}
\end{align}

\item Update $q(\blambda_k) = N(\bmu_{\blambda_k}, \Sigma_{\blambda_k})$:

For $k \in \set{2, \dots K}$:
\begin{align}
\Sigma_{\blambda_k} &= \left[\sum_{t=1}^T \sum_{j < i} \mu_{\omega_{ijt}^k} (\Sigma_t^i + \bmu_t^i \bmu_t^{i \, \rm T}) \odot (\Sigma_t^j + \bmu_t^j \bmu_t^{j \, \rm T}) +
\frac{1}{\sigma_{\lambda}^2} I_p \right]^{-1}, \label{eq:lambda_sigma} \\
\bmu_{\blambda_k} &= \Sigma_{\blambda_k}\begin{pmatrix}
                  \sum_{t=1}^T \sum_{j < i} [Y_{ijt}^k - 1/2 - \mu_{\omega_{ijt}^k} (\mu_{\delta_{k,t}^i} + \mu_{\delta_{k,t}^j})] \mu_{t1}^i \mu_{t1}^j \\
                  \vdots \\
                  \sum_{t=1}^T \sum_{j < i} [Y_{ijt}^k - 1/2 - \mu_{\omega_{ijt}^k} (\mu_{\delta_{k,t}^i} + \mu_{\delta_{k,t}^j})] \mu_{td}^i \mu_{td}^j
                  \end{pmatrix}. \label{eq:lambda_mu}
\end{align}
\end{enumerate}
\end{framed}
\caption{Coordinate ascent updates for the homophily coefficients}
\label{alg:cavi_lambda}
\end{myalgorithm}

\begin{proposition}
Consider the eigenmodel for dynamic multilayer networks with the same prior distributions and variational factorization defined in the main text. For $k \in \set{2, \dots, K}$ (non-reference layers), $q(\blambda_k) = N(\bmu_{\blambda_k}, \Sigma_{\blambda_k})$ with parameters given in Equation (\ref{eq:lambda_sigma}) and Equation (\ref{eq:lambda_mu}).
\end{proposition}
\begin{proof}
First we define some notation. Let $\dyads = \set{(i, j) : j < i, \ 1 \leq i \leq n}$ denote the set of dyads. Define $\bX_t^{i \odot j} = \bX_t^i \odot \bX_t^j$ and $\theta^{ij}_{k,t} = \delta_{k,t}^i + \delta_{k,t}^j$. Let $\bomega_t^k = (\omega_{ijt}^k)_{((i,j) \in \dyads)} \in \Reals{\abs{\dyads}}$ be a vector formed by stacking the $\omega_{ijt}^k$ by dyads. Also, let $\Omega_k  = \diag(\bomega_1^k, \dots, \bomega_T^k) \in \Reals{\abs{\dyads}T \times \abs{\dyads}T}$. Finally, let  $\bz_k \in \Reals{\abs{\dyads}T}$, $X \in \Reals{\abs{\dyads}T \times d}$ and $\btheta_k \in \Reals{\abs{\dyads} T}$ be formed by stacking $z_{ijt}^k = Y_{ijt}^k - 1/2$, $\bX_t^{i \odot j}$ and $\theta_{k,t}^{ij}$ first by dyads and then the result by time, respectively. Standard manipulations show that
\begin{equation*}
p(\blambda_k \mid \cdot) \propto p(\blambda_k) N(\bz_k \mid \Omega_k X \blambda_k + \Omega_k \btheta_k, \Omega_k).
\end{equation*}
Since $p(\blambda_k) = N(0, \sigma^2_{\lambda} I_d)$, the full conditional distribution is Gaussian with the following natural parameters:
\begin{align*}
\Lambda &= \left[X^{\rm T} \Omega_k X + \frac{1}{\sigma^2_{\lambda}} I_d \right], \\
\boldeta &= X^{\rm T} (\bz_k - \Omega_k \btheta_k).
\end{align*}
Taking expectations with respect to the approximate posterior and converting back to the mean and covariance parameters, we have
\begin{align*}
\Sigma_{\blambda} &= \left[\mean*{X^{\rm T} \Omega_k X} + \frac{1}{\sigma^2_{\lambda}} I_d \right]^{-1}, \\
\bmu_{\blambda_k} &= \Sigma_{\blambda_k} \mean*{X}^{\rm T} (\bz - \mean*{\Omega_k} \mean*{\btheta_k}),
\end{align*}
which is equivalent to the parameters in Equation (\ref{eq:lambda_sigma}) and Equation (\ref{eq:lambda_mu}).
\end{proof}

\begin{proposition}
Consider the eigenmodel for dynamic multilayer networks with the same prior distributions and variational factorization defined in the main text. For $h \in \set{1, \dots, d}$, $q(\lambda_{1h}) = p_{\lambda_{1h}}^{\ind{\lambda_{1h} = 1}} \ (1 - p_{\lambda_{1h}})^{\ind{\lambda_{1h} = -1}}$ where $p_{\lambda_{1h}}$ is given in Equation (\ref{eq:lambda_ref_p}).
\end{proposition}
\begin{proof}
The full conditional distributions are
\begin{align}
p(\lambda_{1h} \mid \cdot) &\propto p(\lambda_{1h}) \exp\bigg\{\sum_{t=1}^T \sum_{j < i} \bigg[(Y_{ijt}^1 - 1/2 - \omega_{ijt}^k (\delta_{1,t}^i + \delta_{1,t}^j)) \lambda_{1h} X_{th}^i X_{th}^j - \nonumber \\
&\qquad \frac{1}{2} \omega_{ijt}^1 \lambda_{1q}^2 (X_{th}^i)^2 (X_{th}^j)^2 - \omega_{ijt}^1 \sum_{g\neq h} \lambda_{1g} \lambda_{1h} X_{tg}^i X_{th}^i X_{th}^j X_{tg}^j \bigg] \bigg\}.
\end{align}
The natural parameter is then
\begin{align*}
\eta_{\lambda_{1h}} &= \Eq{-q(\lambda_{1h})}{\log p(\lambda_{1h} = 1 \mid \cdot)} - \Eq{-q(\lambda_{1h})}{\log p(\lambda_{1h} = -1 \mid \cdot)}, \\
&= \log\left[\frac{\rho}{1 - \rho}\right] + \\
&\qquad 2 \sum_{t=1}^T \sum_{j < i} \bigg\{(Y_{ijt}^1 - 1/2 \mu_{\omega_{ijt}^1} (\mu_{\delta_{1,t}^i} + \mu_{\delta_{1,t}^j})) \mu_{th}^i \mu_{th}^j - \\
&\quad\qquad \mu_{\omega_{ijt}^1} \sum_{g \neq h} \mu_{\lambda_{1g}} ((\Sigma_t^i)_{gh} + \mu_{tg}^i \mu_{th}^i) ((\Sigma_t^j)_{gh}  + \mu_{tg}^j \mu_{th}^i) \bigg\}.
\end{align*}
Converting back to the standard parameterization, we have
\begin{align*}
p_{\lambda_{1h}} &= e^{\eta_{\lambda_{1h}}} / (1 + e^{\eta_{\lambda_{1h}}}), \\
\mu_{\lambda_{1h}} &= 2 p_{\lambda_{1h}} - 1,\\
\sigma_{\lambda_{1h}}^2 &= 1 - (2 p_{\lambda_{1h}} - 1)^2.
\end{align*}
\end{proof}

\section{Derivation of the Variational Kalman Smoother}\label{sec:variational_smoother}

In this section, we derive the variational Kalman smoother used for inference in our model. Many of our results are based on the work in \citet{beal2003}. The major difference in the two formulations is that we incorporate time-varying state space parameters and non-identity covariance matrices.

Consider the Gaussian state space model specified by Equations (\ref{eq:gssm-init}) -- (\ref{eq:gssm-observed}). Our goal is to perform variational inference based on the factorization $q(\bx_{1:T}) q(\btheta)$ where $\btheta$ contains the parameters of the GSSM. Crucially, the variational distributions of the hidden states are also GSSMs. Indeed,
\begin{align*}
q(\bx_{1:T}) = c \exp(\mean{\log p(\bx_{1:T}, \by_{1:T})}) &\triangleq h(\bx_{1:T}, \by_{1:T}) \\
&= h(\bx_1) h(\by_1 \mid \bx_1) \prod_{t=2}^T h(\bx_t \mid \bx_{t-1}) h(\by_t \mid \bx_t),
\end{align*}
where
\begin{align*}
h(\bx_1) &= c_1 \exp(\mean{\log N(\bx_1 \mid 0, \tau^2 I_d)}), \\
h(\bx_t \mid \bx_{t-1}) &= c_2 \exp(\mean{\log N(\bx_t \mid \bx_{t-1}, \sigma^2 I_d)}), \\
h(\by_t \mid \bx_t) &= c_3 \exp(\mean{\log N(\by_t \mid A_t \bx_t + \bb_t, C_t)}),
\end{align*}
are Gaussian distributions and $\mean{\cdot}$ denotes an expectation with respect to $q(\btheta)$. When calculating the distributions in $q(\bx_{1:T})$, we use the Gaussian distribution's natural parameter form so that the expectations have an analytical form. The trade-off in using this parameterization is that when calculating the moments of these variational GSSMs, the standard Kalman smoothing equations are no longer applicable because we only have access to the natural parameters. Therefore, we derive Kalman smoothing equations that are expressed in terms of the variational distribution's natural parameters.

\subsection{Variational Kalman Filter}

\begin{property}
For the Gaussian state space model specified in Equations (\ref{eq:gssm-init}) -- (\ref{eq:gssm-observed}), the variational filtering distributions $h(\bx_t \mid \by_{1:t}) = N(\bx_t \mid \bmu_t, \Sigma_t)$ with parameters that can be calculated recursively via Algorithm \ref{alg:kalman_filter}.
\end{property}

\begin{myalgorithm}[tbp]
\begin{framed}
Given the natural parameters $\mean*{\Gamma_{1:T}^1} = \mean*{A_{1:T}^{\rm T} C_{1:T}^{-1}} \by_{1:T} - \mean*{A_{1:T}^{\rm T} C_{1:T}^{-1} \bb_{1:T}}$, $\mean*{\Gamma_{1:T}^2} = \mean*{A_{1:T}^{\rm T} C_{1:T}^{-1} A_{1:T}}$, $\mean*{1/\tau^2}$, and $\mean*{1 / \sigma^2}$, calculate the filtering distribution's moments as follows:
\begin{enumerate}
\item For $t = 1$:
\begin{align*}
\Sigma_1 &= \left[\mean*{\frac{1}{\tau^2}} I_d + \mean*{A_1^{\rm T} C_1^{-1} A_1}\right]^{-1}, \\
\bmu_1 &= \Sigma_1 \left[\mean*{A_1^{\rm T} C_1^{-1}} \by_1 - \mean*{A_1^{\rm T} C_1^{-1} \bb_1} \right].
\end{align*}
\item For $t = 2, \dots, T$:
\begin{align*}
\Sigma_{t-1}^{*} &= \left[\mean*{\frac{1}{\sigma^2}} I_d + \Sigma_{t-1}^{-1}\right]^{-1}, \\
\Sigma_t &= \left[\mean*{\frac{1}{\sigma^2}} I_d + \mean*{A_t^{\rm T} C_t^{-1} A_t} - \mean*{\frac{1}{\sigma^2}}^2 \Sigma_{t-1}^{*}\right]^{-1}, \\
\bmu_t &= \Sigma_t \left[\mean*{A_t^{\rm T} C_t^{-1}} \by_t  - \mean*{A_t^{\rm T} C_t^{-1}\bb_t} + \mean*{\frac{1}{\sigma^2}} \Sigma_{t-1}^{*} \Sigma_{t-1}^{-1} \bmu_{t-1}\right].
\end{align*}
\end{enumerate}

Output $\set{\bmu_{1:T}, \Sigma_{1:T}, \Sigma_{1:(T-1)}^{*}}$.
\end{framed}
\caption{Variational Kalman filter.}
\label{alg:kalman_filter}
\end{myalgorithm}

\begin{proof}
Define the forwards message variables as $\alpha_t(\bx_t) = h(\bx_t \mid \by_{1:t})$. The filter proceeds recursively starting at $t = 1$:
\begin{align*}
\alpha_1(\bx_1) &\propto h(\bx_1) h(\by_1 \mid \bx_1), \\
&\propto \exp\left(-\frac{1}{2}\mean*{\bx_1^{\rm T}\frac{1}{\tau^2} I_d \bx_1 + (\by_1 - A_1 \bx_1 - \bb_1)^{\rm T} C_1^{-1} (\by_1 - A_1 \bx_1 - \bb_1)}\right), \\
&\propto \exp\bigg(-\frac{1}{2} \bx_1^{\rm T}(\mean*{\frac{1}{\tau^2}} I_d + \mean*{A_1^{\rm T} C_1^{-1}A_1})\bx_1 + \\
&\qquad\qquad\qquad (\mean*{A_1^{\rm T} C_1^{-1}}\by_1 - \mean*{A_1^{\rm T} C_1^{-1} \bb_1})^{\rm T} \bx_1\bigg), \\
&\propto N(\bx_1 \mid \bmu_1, \Sigma_1),
\end{align*}
where
\begin{align*}
\Sigma_1 &= \left[\mean*{\frac{1}{\tau^2}} I_d + \mean*{A_1^{\rm T} C_1^{-1} A_1}\right]^{-1}, \\
\bmu_1 &= \Sigma_1 \left[\mean*{A_1^{\rm T} C_1^{-1}} \by_1 - \mean*{A_1^{\rm T} C_1^{-1} \bb_1} \right].
\end{align*}
The last line follows from matching the natural parameters of a Gaussian density. This shows that $\alpha_1(\bx_1)$ is Gaussian; therefore, we can proceed by induction. For $t > 1$, we have
\begin{align*}
\alpha_t(\bx_t) &\propto \int d\bx_{t-1} \, h(\bx_t, \by_{1:t}, \bx_{t-1}), \\
&= \int d\bx_{t-1} \, h(\bx_{t-1} \mid \by_{1:t-1}) h(\bx_t \mid \bx_{t-1}) h(\by_t \mid \bx_t), \\
&\propto \int d\bx_{t-1} \, \alpha_{t-1}(\bx_{t-1}) \times  \\
&\qquad \exp\bigg(-\frac{1}{2}\Big[\bx_{t-1}^{\rm T} \mean*{\frac{1}{\sigma^2}} I_d \bx_{t-1} - 2 \bx_{t-1}^{\rm T} \mean*{\frac{1}{\sigma^2}} I_d \bx_t + \nonumber \\
&\qquad \bx_t^{\rm T}(\mean*{\frac{1}{\sigma^2}} I_d + \mean*{A_t^{\rm T} C_t^{-1} A_t})\bx_t - 2 (\mean*{A_t^{\rm T} C_t^{-1}} \by_t  - \mean*{A_t^{\rm T} C_t^{-1}\bb_t})^{\rm T} \bx_t \Big]\bigg).
\end{align*}
Inside the integral is a $N(\bx_{t-1} \mid \bmu_{t-1}^{*}, \Sigma_{t-1}^{*})$ with
\begin{align*}
\Sigma_{t-1}^{*} &= \left[\mean*{\frac{1}{\sigma^2}} I_d + \Sigma_{t-1}^{-1}\right]^{-1}, \\
\bmu_{t-1}^{*} &= \Sigma_{t-1}^{*}  \left[\mean*{\frac{1}{\sigma^2}} \bx_t + \Sigma_{t-1}^{-1} \bmu_{t-1} \right].
\end{align*}
Marginalizing over this distribution leaves the following terms
\begin{align*}
\alpha_t(\bx_t) &\propto \exp(- \frac{1}{2} \bx_t^{\rm T}(\mean*{\frac{1}{\sigma^2}} I_d + \mean*{A_t^{\rm T} C_t^{-1} A_t} - \mean*{\frac{1}{\sigma^2}} \Sigma_{t-1}^{*}\mean*{\frac{1}{\sigma^2}})\bx_t + \nonumber \\
&\quad (\mean*{A_t^{\rm T} C_t^{-1}} \by_t  - \mean*{A_t^{\rm T} C_t^{-1}\bb_t})^{\rm T} \bx_t + \frac{1}{2} \bmu_{t-1}^{* \, \rm T} \Sigma_{t-1}^{* \, -1} \bmu_{t-1}^{*}), \\
&\propto \exp(- \frac{1}{2} \bx_t^{\rm T}(\mean*{\frac{1}{\sigma^2}} I_d + \mean*{A_t^{\rm T} C_t^{-1} A_t} - \mean*{\frac{1}{\sigma^2}} \Sigma_{t-1}^{*}\mean*{\frac{1}{\sigma^2}})\bx_t + \nonumber \\
&\quad (\mean*{A_t^{\rm T} C_t^{-1}} \by_t  - \mean*{A_t^{\rm T} C_t^{-1}\bb_t} + \mean*{\frac{1}{\sigma^2}} \Sigma_{t-1}^{*} \Sigma_{t-1}^{-1} \bmu_{t-1})^{\rm T} \bx_t), \\
&\propto N(\bx_t \mid \bmu_t, \Sigma_t),
\end{align*}
where
\begin{align*}
\Sigma_t &= \left[\mean*{\frac{1}{\sigma^2}} I_d + \mean*{A_t^{\rm T} C_t^{-1} A_t} - \mean*{\frac{1}{\sigma^2}}^2 \Sigma_{t-1}^{*}\right]^{-1}, \\
\bmu_t &= \Sigma_t \left[\mean*{A_t^{\rm T} C_t^{-1}} \by_t  - \mean*{A_t^{\rm T} C_t^{-1}\bb_t} + \mean*{\frac{1}{\sigma^2}} \Sigma_{t-1}^{*} \Sigma_{t-1}^{-1} \bmu_{t-1}\right].
\end{align*}
\end{proof}

\subsection{Variational Kalman Smoother}

\begin{property}
For the Gaussian state space model specified in Equations (\ref{eq:gssm-init}) -- (\ref{eq:gssm-observed}), the backwards message variables $h(\by_{(t+1):T} \mid \bx_t) = N(\bx_t \mid \boldeta_t, \Psi_t)$ with parameters that can be calculated recursively via Algorithm \ref{alg:kalman_smoother}.
\end{property}

\begin{myalgorithm}[bp]
\begin{framed}
Given the natural parameters $\mean*{\Gamma_{1:T}^1} = \mean*{A_{1:T}^{\rm T} C_{1:T}^{-1}} \by_{1:T} - \mean*{A_{1:T}^{\rm T} C_{1:T}^{-1} \bb_{1:T}}$, $\mean*{\Gamma_{1:T}^2} = \mean*{A_{1:T}^{\rm T} C_{1:T}^{-1} A_{1:T}}$, $\mean*{1/\tau^2}$, and $\mean*{1 / \sigma^2}$, calculate the smoothing distribution's moments and cross-covariances as follows:

\begin{enumerate}
\item Calculate $\set{\bmu_{1:T}, \Sigma_{1:T}, \Sigma_{1:(T-1)}^{*}}$ as in Algorithm \ref{alg:kalman_filter}.
\item Set $\bnu_T = \bmu_T$ and $\Upsilon_T = \Sigma_T$.
\item Initialize $\Psi_T^{-1} = 0$ to satisfy $\beta_T(\bx_T) = 1$.
\item For $t = T, \dots, 2$.

$\vartriangleright$ Calculate backwards message variables
\begin{align*}
\Psi_t^{*} &= \left[\mean*{\frac{1}{\sigma^2}} I_d + \mean*{A_t^{\rm T} C_t^{-1} A_t} + \Psi_t^{-1} \right]^{-1}, \\
\Psi_{t-1} &= \left[\mean*{\frac{1}{\sigma^2}} I_d - \mean*{\frac{1}{\sigma^2}}^2 \Psi_t^{*} \right]^{-1}, \\
\boldeta_{t-1} &= \Psi_{t-1}\left[\mean*{\frac{1}{\sigma^2}} \Psi_t^{*} \left(\mean*{A_t^{\rm T} C_t^{-1}} \by_t - \mean*{A_t^{\rm T} C_t^{-1} \bb_t} + \Psi_{t}^{-1}\boldeta_{t} \right) \right].
\end{align*}

$\vartriangleright$ Calculate smoothing distributions and cross-covariances
\begin{align*}
\Upsilon_{t-1} &= \left[\Sigma_{t-1}^{-1} + \Psi_{t-1}^{-1}\right]^{-1}, \\
\Upsilon_{t-1, t} &= \mean*{\frac{1}{\sigma^2}} \Sigma_{t-1}^{*} \left[\Psi_t^{* \, -1} - \mean*{\frac{1}{\sigma^2}}^2 \Sigma_{t-1}^{*} \right]^{-1}, \\
\bnu_{t-1} &= \Upsilon_{t-1} \left[\Sigma_{t-1}^{-1} \bmu_{t-1} + \Psi_{t-1}^{-1} \boldeta_{t-1} \right].
\end{align*}
\end{enumerate}

Output smoothing distributions $\set{\bnu_{1:T}, \Upsilon_{1:T}}$ and cross-covariances $\set{\Upsilon_{t, t+1}}_{t=1}^{T-1}$.
\end{framed}
\caption{Variational Kalman smoother. Throughout the text, we refer to this algorithm as $\text{\tt KalmanSmoother}(\mean*{\Gamma_{1:T}^1}, \mean*{\Gamma_{1:T}^2}, \mean*{1/\tau^2}, \mean*{1/\sigma^2})$.}
\label{alg:kalman_smoother}
\end{myalgorithm}

\begin{proof}
Define the backwards message variables $\beta_t(\bx_t) = h(\by_{(t+1):T} \mid \bx_t)$ and set $\beta_T(\bx_T) = 1$. Note that
\begin{align*}
\beta_{t-1}(\bx_{t-1}) &= \int d\bx_t \, h(\by_{t:T}, \bx_t \mid \bx_{t-1}), \\
&= \int d\bx_t \, h(\by_{(t+1):T} \mid \bx_t) h(\by_t \mid \bx_t) h(\bx_t \mid \bx_{t-1}), \\
&= \int d\bx_t \, \beta_{t}(\bx_t) h(\by_t \mid \bx_t) h(\bx_t \mid \bx_{t-1}).
\end{align*}
The derivation proceeds sequentially backwards in time. For $t = T - 1$, we have
\begin{align*}
\beta_{T-1}(\bx_{T-1}) &\propto \int d\bx_{T} \, \times \\
&\exp\bigg(-\frac{1}{2} \Big\langle(\by_T - A_T\bx_T - \bb_T)^{\rm T} C_T^{-1} (\by_T - A_T\bx_T - \bb_T) + \\
&\qquad (\bx_{T} - \bx_{T-1})^{\rm T} \frac{1}{\sigma^2} I_d (\bx_{T} - \bx_{T-1})\Big\rangle\bigg), \\
&\propto \int d\bx_T \, \exp\bigg(-\frac{1}{2}\Big[\bx_{T-1} \mean*{\frac{1}{\sigma^2}} I_d \bx_{T-1}  + \bx_T^{\rm T} (\mean*{\frac{1}{\sigma^2}} I_d + \mean*{A_T^{\rm T} C_T^{-1} A_T}) \bx_T - \nonumber \\
&\qquad 2 (\mean*{A_T^{\rm T} C_T^{-1}}\by_T - \mean*{A_T^{\rm T} C_T^{-1} \bb_T} + \mean*{\frac{1}{\sigma^2}} \bx_{T-1})^{\rm T} \bx_T \Big]\bigg).
\end{align*}
Inside the integral is a $N(\bx_T \mid \boldeta_T^{*}, \Psi_T^{*})$ with
\begin{align*}
\Psi_T^{*} &= \left[\mean*{\frac{1}{\sigma^2}} I_d + \mean*{A_T^{\rm T} C_T^{-1} A_T}\right]^{-1}, \\
\boldeta_T^{*} &= \Psi_T^{*} \left[\mean*{A_T^{\rm T} C_T^{-1}}\by_T - \mean*{A_T^{\rm T} C_T^{-1} \bb_T} + \mean*{\frac{1}{\sigma^2}} \bx_{T-1}\right].
\end{align*}
Marginalizing over this density, we are left with
\begin{align*}
\beta_{T-1}(\bx_{T-1}) &\propto \exp\bigg(-\frac{1}{2} \bx_{T-1}^{\rm T} \mean*{\frac{1}{\sigma^2}} \bx_{T-1} + \frac{1}{2} \boldeta_T^{* \, \rm T} \Psi_T^{*\, -1} \boldeta_T^{*}\bigg), \\
&\propto \exp\bigg(-\frac{1}{2}\bx_{T-1}^{\rm T} (\mean*{\frac{1}{\sigma^2}} I_d - \mean*{\frac{1}{\sigma^2}}^2 \Psi_T^{*})\bx_{T-1} + \nonumber \\
&\qquad (\mean*{A_T^{\rm T} C_T^{-1}} \by_T - \mean*{A_T^{\rm T} C_T^{-1} \bb_T})^{\rm T} \mean*{\frac{1}{\sigma^2}} \Psi_T^{*} \bx{T-1}\bigg), \\
&\propto N(\bx_{T-1} \mid \boldeta_{T-1}, \Psi_{T-1}),
\end{align*}
where
\begin{align*}
\Psi_{T-1} &= \left[\mean*{\frac{1}{\sigma^2}} I_d - \mean*{\frac{1}{\sigma^2}}^2 \Psi_T^{*} \right]^{-1}, \\
\boldeta_{T-1} &= \Psi_{T-1}\left[\mean*{\frac{1}{\sigma^2}} \Psi_T^{*} \left(\mean*{A_T^{\rm T} C_T^{-1}} \by_T - \mean*{A_T^{\rm T} C_T^{-1} \bb_T}\right) \right].
\end{align*}
Now we proceed inductively. For $2 < t < T - 1$, we have
\begin{align*}
\beta_{t-1}(\bx_{t-1}) &\propto \int d\bx_t \, N(\bx_t \mid \boldeta_t, \Psi_t) \times \nonumber \\
&\exp\bigg(-\frac{1}{2} \Big\langle(\by_t - A_t\bx_t - \bb_t)^{\rm T} C_t^{-1} (\by_t - A_t\bx_t - \bb_t) + \nonumber \\
&\qquad (\bx_{t} - \bx_{t-1})^{\rm T} \frac{1}{\sigma^2} I_d (\bx_{t} - \bx_{t-1})\Big\rangle\bigg).
\end{align*}
In fact, the remaining derivation is exactly the same as before, except
\begin{align*}
\Psi_t^{*} &= \left[\mean*{\frac{1}{\sigma^2}} I_d + \mean*{A_t^{\rm T} C_t^{-1} A_t} + \Psi_t^{-1} \right]^{-1}, \\
\boldeta_t^{*} &= \Psi_t^{*} \left[\mean*{A_t^{\rm T} C_t^{-1}}\by_t - \mean*{A_t^{\rm t} C_t^{-1} \bb_t} + \Psi_t^{-1} \boldeta_t + \mean*{\frac{1}{\sigma^2}} \bx_{t-1}\right],
\end{align*}
so that $\beta_{t-1}(\bx_{t-1}) = N(\bx_{t-1} \mid \boldeta_{t-1}, \Psi_{t-1})$ with
\begin{align*}
\Psi_{t-1} &= \left[\mean*{\frac{1}{\sigma^2}} I_d - \mean*{\frac{1}{\sigma^2}}^2 \Psi_T^{*} \right]^{-1}, \\
\boldeta_{t-1} &= \Psi_{t-1}\left[\mean*{\frac{1}{\sigma^2}} \Psi_t^{*} \left(\mean*{A_t^{\rm T} C_t^{-1}} \by_t - \mean*{A_t^{\rm T} C_t^{-1} \bb_t} + \Psi_{t}^{-1}\boldeta_{t} \right) \right].
\end{align*}
\end{proof}

\begin{property}
For the Gaussian state space model specified in Equations (\ref{eq:gssm-init}) -- (\ref{eq:gssm-observed}), the variational smoothing distributions $h(\bx_t \mid \by_{1:T}) = N(\bx_t \mid \bnu_t, \Upsilon_t)$ with parameters that can be calculated recursively via Algorithm \ref{alg:kalman_smoother}.
\end{property}

\begin{proof}
We define forwards and backwards message variables, $\alpha_t(\bx_t)$ and $\beta_t(\bx_t)$, as in the previous proofs. For $t = T$, $h(\bx_T \mid \by_{1:T}) = \alpha_T(\bx_T) = N(\bx_T \mid \bmu_T, \Sigma_T)$. For $t < T - 1$, we have
\begin{align*}
h(\bx_t \mid \by_{1:T}) &\propto \alpha_t(\bx_t) \beta_t(\bx_t), \\
&\propto N(\bx_t \mid \bmu_t, \Sigma_t) N(\bx_t \mid \boldeta_t, \Psi_t), \\
&\propto N(\bx_t \mid \bnu_t, \Upsilon_t),
\end{align*}
where
\begin{align*}
\Upsilon_t &= \left[\Sigma_t^{-1} + \Psi_t^{-1}\right]^{-1}, \\
\bnu_t &= \Upsilon_t \left[\Sigma_t^{-1} \bmu_t + \Psi_t^{-1} \boldeta_t \right].
\end{align*}
\end{proof}

\subsection{Cross-Covariance Matrices}

\begin{property}
For the Gaussian state space model specified in Equations (\ref{eq:gssm-init}) -- (\ref{eq:gssm-observed}), the variational joint distributions $h(\bx_t, \bx_{t+1} \mid \by_{1:T})$ are Gaussian with cross-covariance matrices $\Upsilon_{t,t+1}$ that can be calculated recursively via Algorithm \ref{alg:kalman_smoother}.
\end{property}

\begin{proof}
We have that
\begin{align*}
h(\bx_t, \bx_{t+1} \mid \by_{1:T}) &\propto h(\bx_t \mid \bx_{1:t}) h(\bx_{t+1} \mid \bx_t) h(\by_{t+1} \mid \bx_{t+1}) h(\by_{(t+2):T} \mid \bx_{t+1}), \\
&\propto \alpha_t(\bx_t) h(\bx_{t+1} \mid \bx_t) h(\by_{t+1} \mid \bx_{t+1}) \beta_{t+1}(\bx_{t+1}).
\end{align*}
To determine $\Upsilon_{t, t+1}$, we identify the cross-terms in the above product. This product is proportional to
\begin{align*}
&\quad\alpha_t(\bx_t) \ \times  \nonumber \\
&\quad\qquad\exp\bigg(-\frac{1}{2} (-\bx_t^{\rm T} \mean*{\frac{1}{\sigma^2}} I_d \bx_{t+1} + \bx_t^{\rm T} \mean*{\frac{1}{\sigma^2}} I_d \bx_t + \bx_{t+1}^{\rm T} \mean*{\frac{1}{\sigma^2}} I_d \bx_{t+1} \nonumber \\
&\quad\quad\qquad \bx_{t+1}^{\rm T} \mean*{A_{t+1}^{\rm T} C_{t+1}^{-1} A_{t+1}} \bx_{t+1} - \nonumber \\
&\quad\qquad 2(\mean*{A_{t+1}^{\rm T}C_{t+1}^{-1}}\by_{t+1} + \mean*{A_{t+1}^{\rm T}C_{t+1}^{-1}\bb_{t+1}})^{\rm T}\bx_{t+1})\bigg) \ \times \nonumber \\
&\quad\qquad\beta_{t+1}(\bx_{t+1}), \\
&\propto \exp\left\{-\frac{1}{2}
\begin{pmatrix}
\bx_t \\
\bx_{t+1}
\end{pmatrix}^{\rm T}
\begin{pmatrix}
\Gamma_t & \Gamma_{t,t+1} \\
\Gamma_{t, t+1}^{\rm T} & \Gamma_{t+1}
\end{pmatrix}
\begin{pmatrix}
\bx_t \\
\bx_{t+1}
\end{pmatrix}
\right\},
\end{align*}
with
\begin{align*}
\Gamma_{t} &= \mean*{\frac{1}{\sigma^2}} I_d + \Sigma_t^{-1} = \Sigma_t^{* \, -1},\\
\Gamma_{t+1} &= \mean*{\frac{1}{\sigma^2}} I_d + \mean*{A_{t+1}^{\rm T} C_{t+1}^{-1} A_{t+1}} + \Psi_{t+1}^{-1} = \Psi_{t+1}^{* \, -1}, \\
\Gamma_{t, t+1} &= -\mean*{\frac{1}{\sigma^2}} I_d,
\end{align*}
where in the last line we only kept terms quadratic in $\bx_t$ and $\bx_{t+1}$. Applying Schur's compliment, we obtain the cross-covariance matrix
\begin{equation*}
\Upsilon_{t, t+1} = \Sigma_t^{*} \mean*{\frac{1}{\sigma^2}} \left[\Psi_{t+1}^{* \, -1} - \mean*{\frac{1}{\sigma^2}}^2 \Sigma_t^{*} \right]^{-1}.
\end{equation*}
\end{proof}

\section{Parameter Initialization and Prior Settings}\label{sec:init}

Due to variational inference algorithm's non-convex objective, appropriate initial values of the approximate posterior parameters can greatly improve convergence. We used the same initialization scheme and prior settings for all experiments and real-data applications.

We initialized the social trajectories $\bdelta_{1:K, 1:T}$ by estimating independent two-way logistic regression models
\[
Y_{ijt}^k \iidsim \Bern\left(\text{logit}^{-1}\left[\delta_{k,t}^i + \delta_{k,t}^j\right]\right),
\]
for each $k = 1, \dots, K$ and $t = 1, \dots T$. Furthermore, we set $\sigma^2_{\delta_{k,t}^i} = 1$ and $\sigma_{\delta_{k,t,t+1}^i}^2 = 1$. We placed broad priors on the state space parameters. To make the prior on $\tau^2_{\delta}$ flat, one can set the shape and scale parameters of the inverse gamma priors to $a_{\tau_{\delta}^2} = 2 (2 + \epsilon)$ and $b_{\tau_{\delta}^2} = 2 (1 + \epsilon) \E{\tau^2_{\delta}}$ for some small $\epsilon > 0$, respectively. We set $\epsilon = 0.05$ and $\E{\tau^2_{\delta}} = 10$. For $\sigma^2_{\delta}$, we set $c_{\sigma^2_{\delta}} = 2$ and $d_{\sigma^2_{\delta}} = 2$.

Next, we initialized the latent trajectories $\cX_{1:T}$ by sampling the entries independently from a $N(0, 1)$. We set the variances and cross-covariances to the identity matrix $I_d$. As for the social trajectories, we can make the prior on $\tau^2$ flat by setting  $a_{\tau^2} = 2 (2 + \epsilon)$ and $b_{\tau^2} = 2 (1 + \epsilon) \E{\tau^2}$. We set $\epsilon = 0.05$ and $\E{\tau^2} = 10$. For $\sigma^2$, we set $c_{\sigma^2} = 2$ and $d_{\sigma^2} = 2$.

For the homophily matrix, we initialized the reference layer differently from the rest. For the reference layer, we set $\mu_{\lambda_{1h}} = 1$ for $h = 1, \dots, d$ and $\rho = 1/2$. We sampled $\bmu_{\blambda_k} \sim N(0, 4 I_d)$ and set $\Sigma_{\blambda_k} = 10 I_d$. In addition, we set the prior variance to $\sigma^2_{\lambda} = 10$.

Lastly, we initialized the mean of the P\'olya-gamma random variables, $\mu_{\omega_{ijt}^k}$, to zero.

\section{Additional Figures}\label{app:figures}

In this section, we include the remaining figures for the data analyzed in the main text's real data applications (Section \ref{sec:realdata}). These figures include the remaining social trajectories for the international relations and primary school networks and the homophily coefficients for the primary school networks.

We start with the remaining figures for the ICEWS data set. Figure \ref{fig:social_trajectory_vcoop} contains the social trajectories for the verbal cooperation networks, Figure \ref{fig:social_trajectory_mcoop} contains the social trajectories for the material cooperation networks, and Figure \ref{fig:social_trajectory_vcon} contains the social trajectories for the verbal conflict networks. The shapes of the social trajectories are similar to the material conflict relation. A notable difference is that Ukraine increases its material cooperation and verbal conflict sociality dramatically leading up to the Crimea Crisis.

\begin{figure}[tbp!]
\centering
\includegraphics[width=0.9\textwidth]{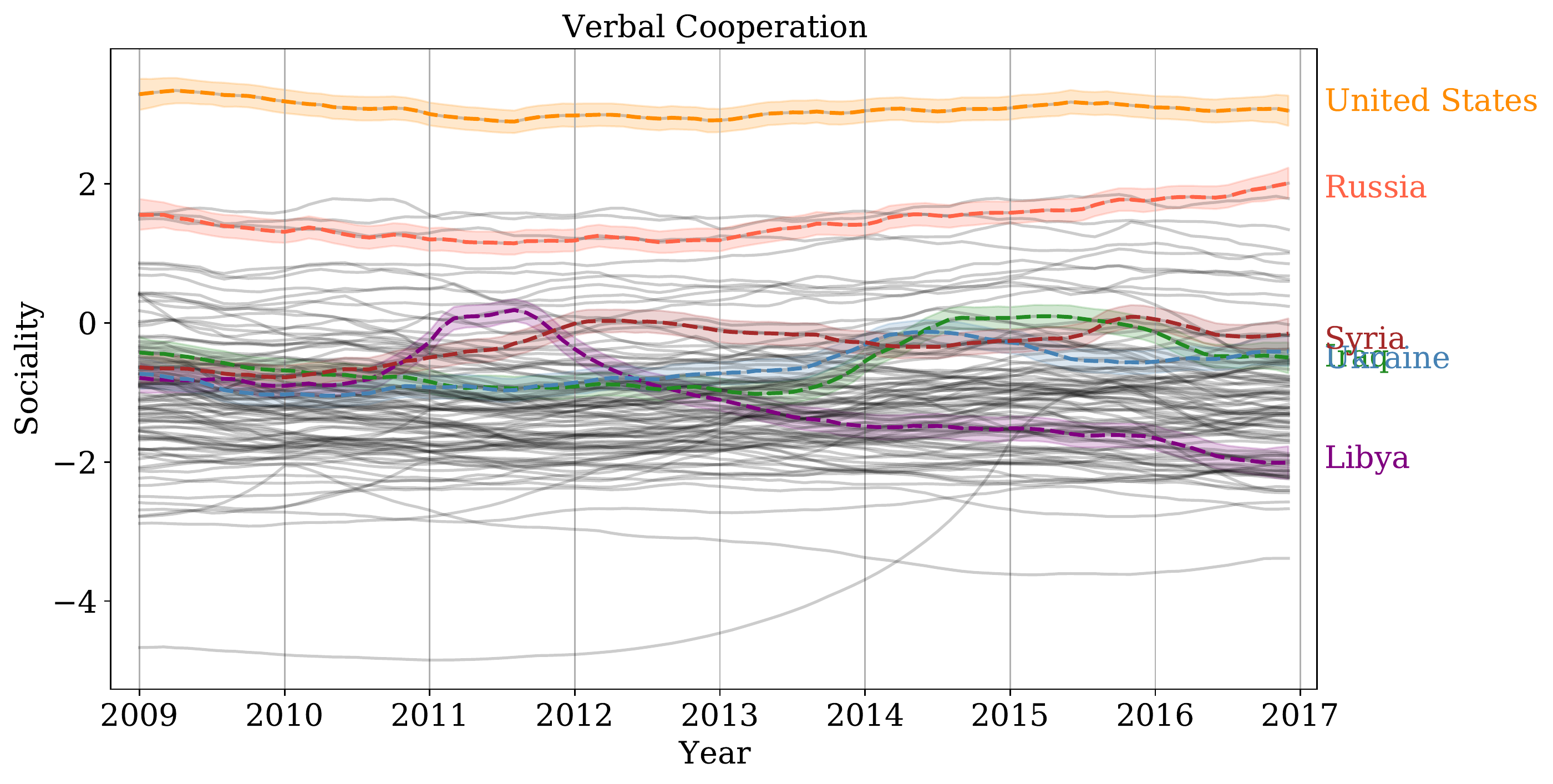}
\caption{Posterior means of the verbal cooperation social trajectories for the ICEWS network. Select countries are highlighted in color with bands that represent 95\% credible intervals. The remaining countries' social trajectories are displayed with gray curves.}
\label{fig:social_trajectory_vcoop}
\end{figure}

\begin{figure}[tbp!]
\centering
\includegraphics[width=0.9\textwidth]{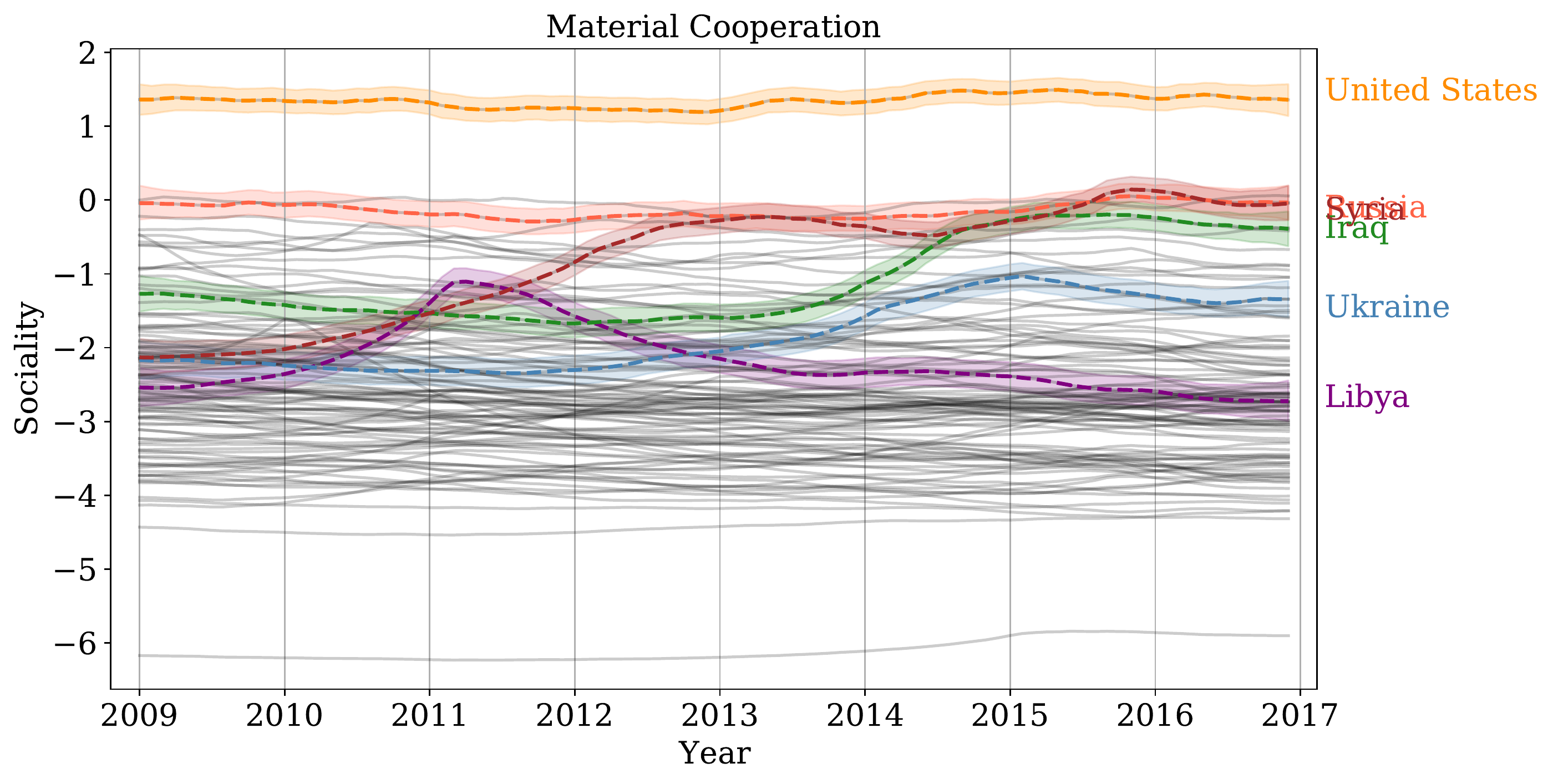}
\caption{Posterior means of the material cooperation social trajectories for the ICEWS network. Select countries are highlighted in color with bands that represent 95\% credible intervals. The remaining countries' social trajectories are displayed with gray curves.}
\label{fig:social_trajectory_mcoop}
\end{figure}

\begin{figure}[tbp!]
\centering
\includegraphics[width=0.9\textwidth]{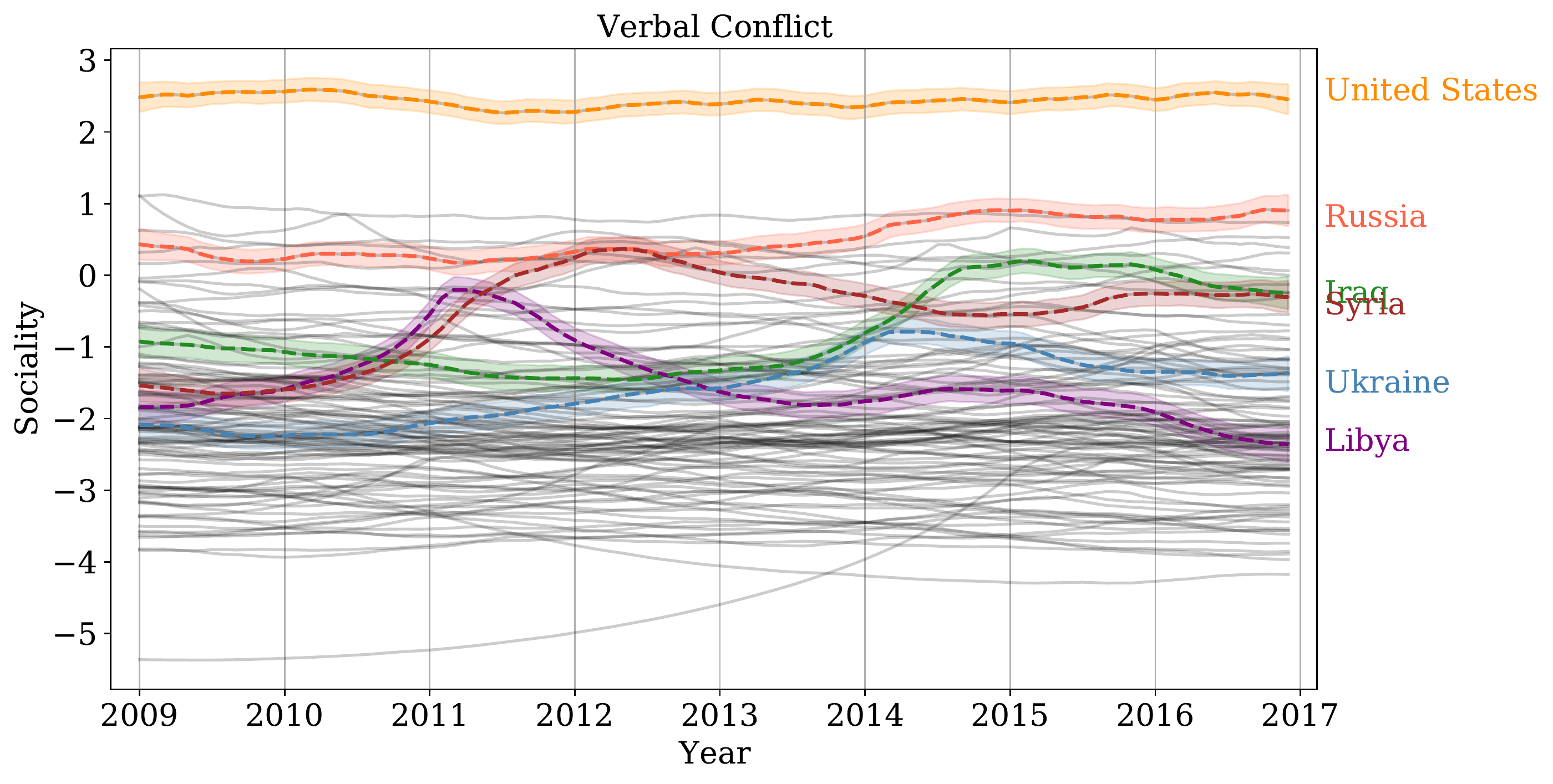}
\caption{Posterior means of the verbal conflict social trajectories for the ICEWS network. Select countries are highlighted in color with bands that represent 95\% credible intervals. The remaining countries' social trajectories are displayed with gray curves.}
\label{fig:social_trajectory_vcon}
\end{figure}

Next, we present the social trajectories and homophily coefficients for the primary school network. Figure \ref{fig:social_trajectory_thursday} and Figure \ref{fig:social_trajectory_friday} contain the actors' social trajectories on Thursday and Friday, respectively. We highlighted the trajectories of three actors. Actor 148 is a teacher, actor 195 is a student in class 3A, and actor 5 is a student in class 5B. Their social trajectories demonstrate three interesting longitudinal patterns. Actor 148, the teacher, is most socially active during class and least active during lunch. Conversely, actor 195 is most sociable during lunch and less sociable during class. Lastly, actor 5's social trajectory differs between the two days because he/she is absent on Friday. Next, Figure \ref{fig:primary_homophily} contains the primary school network's homophily coefficients. All homophily coefficients are positive. Also, the magnitude of the homophily coefficients is larger on Friday than on Thursday.

\begin{figure}[tbp]
\centering
\includegraphics[width=0.8\textwidth]{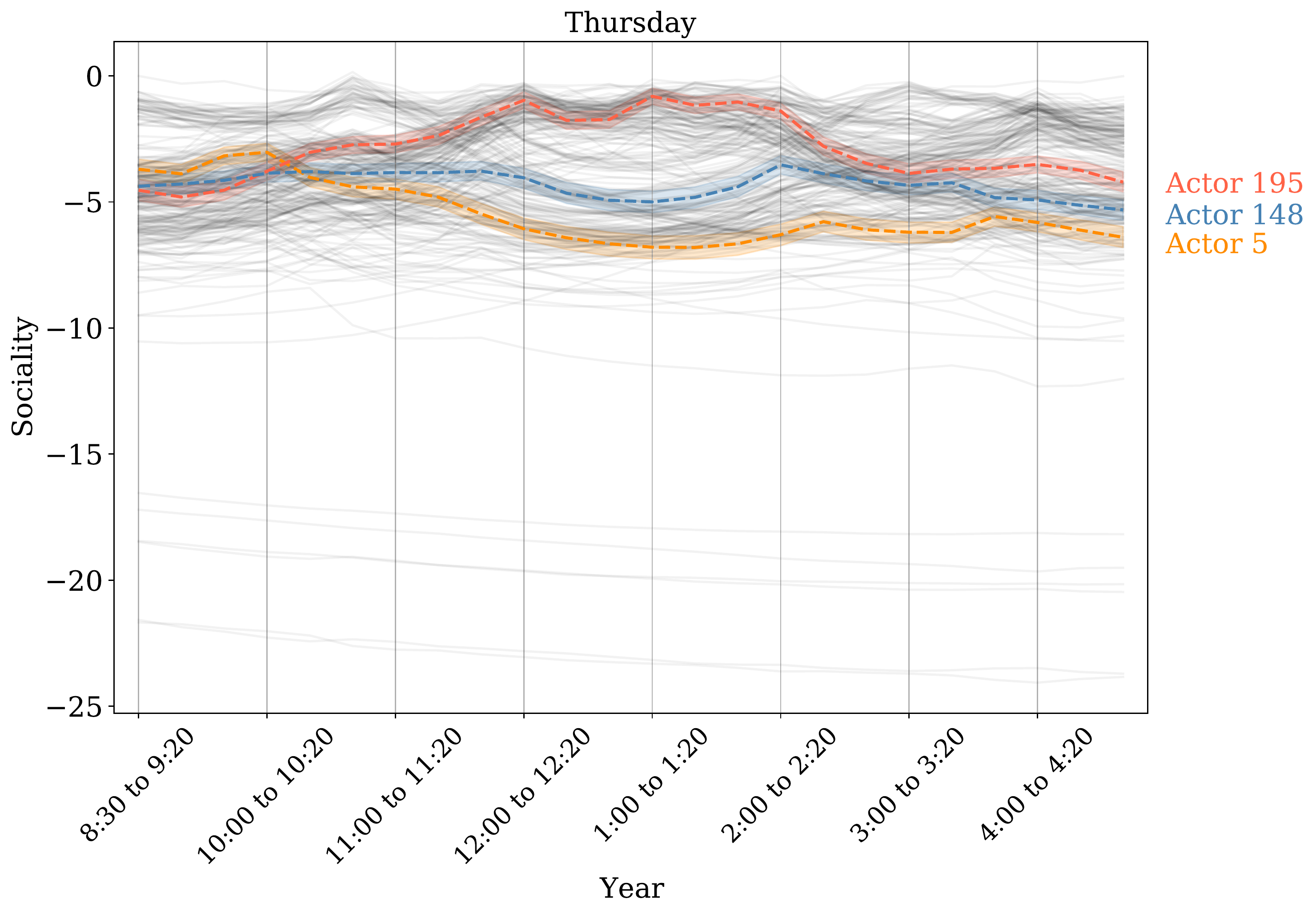}
\caption{Posterior means of the social trajectories on Thursday for the primary school network. Select actors are highlighted in color with bands that represent 95\% credible intervals. The remaining actors' social trajectories are displayed with gray curves.}
\label{fig:social_trajectory_thursday}
\end{figure}

\begin{figure}[tbp]
\centering
\includegraphics[width=0.8\textwidth]{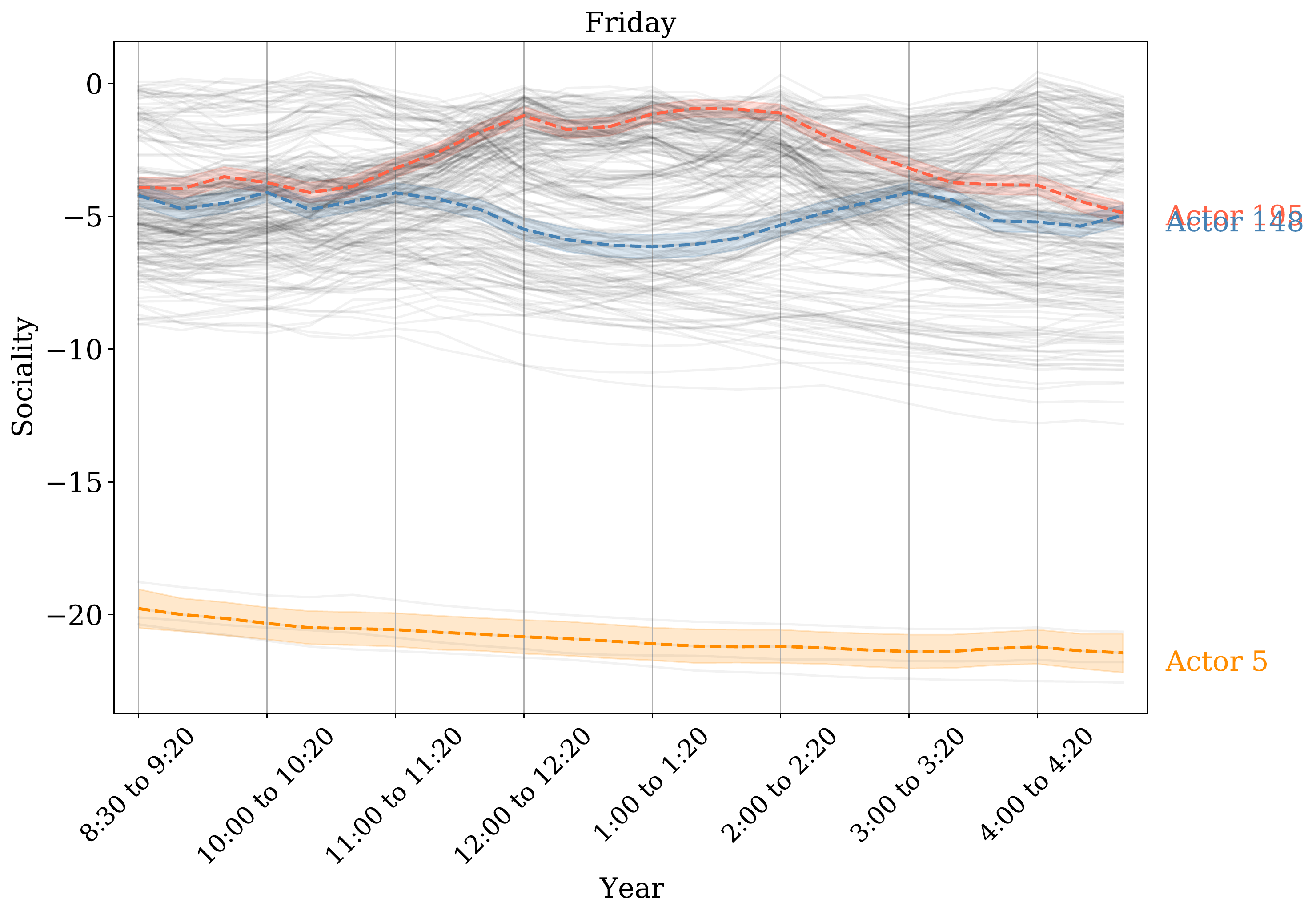}
\caption{Posterior means of the social trajectories on Friday for the primary school network. Select actors are highlighted in color with bands that represent 95\% credible intervals. The remaining actors' social trajectories are displayed with gray curves.}
\label{fig:social_trajectory_friday}
\end{figure}

\begin{figure}[tbp]
\centering
\includegraphics[width=0.8\textwidth]{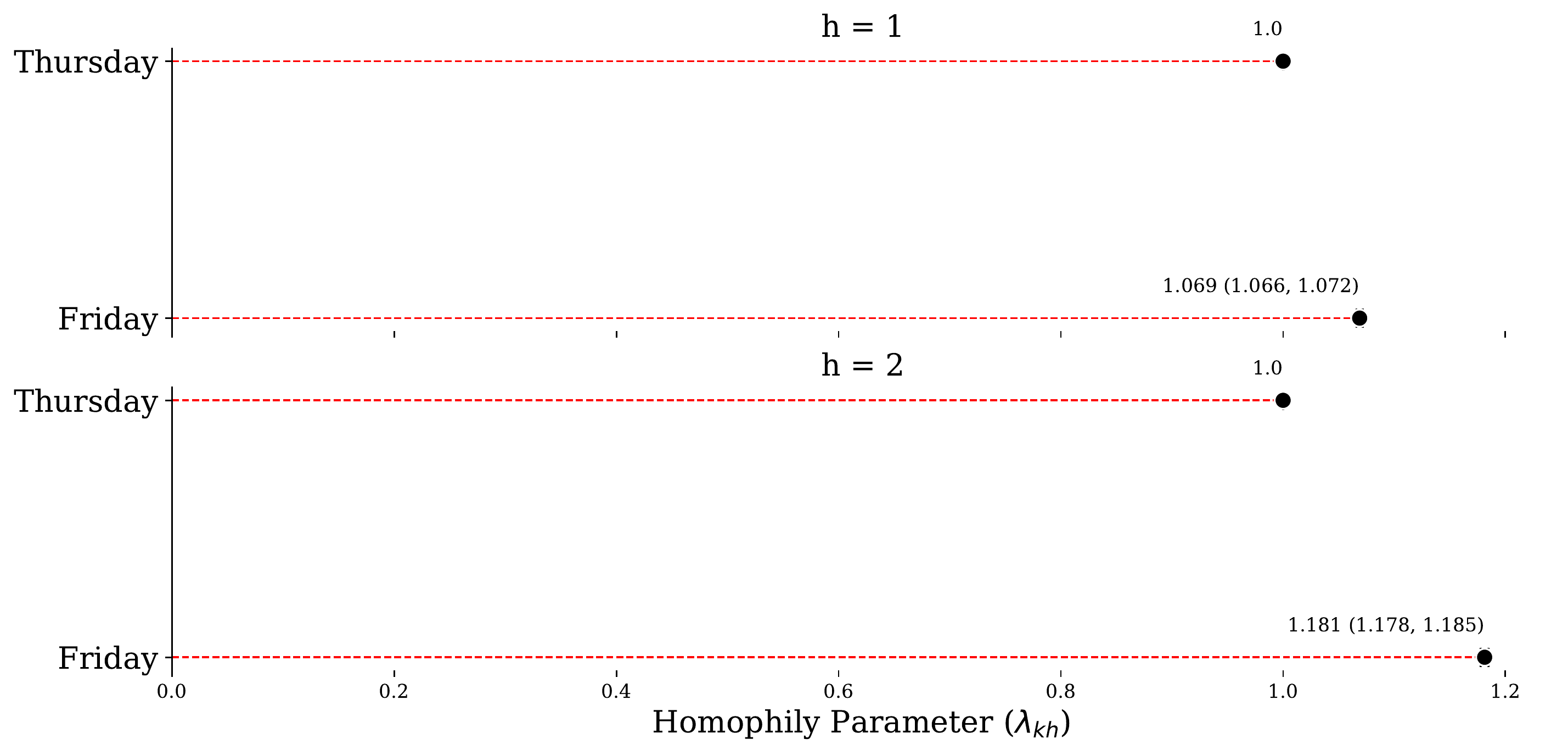}
\caption{The homophily coefficients' posterior means and 95\%  credible intervals for the primary school face-to-face contact networks. The top and bottom plots give estimates for the degree of homophily along the first and second latent dimensions, respectively.}
\label{fig:primary_homophily}
\end{figure}

\clearpage

\bibliography{reference}

\end{document}